\documentclass[a4paper,UKenglish,cleveref]{lipics-v2019}

\pdfoutput=1

\usepackage{microtype}
\usepackage{mathtools,stmaryrd}
\usepackage{xspace}

\usepackage{todonotes}

\usepackage{tikz}

\usetikzlibrary{matrix}
\usetikzlibrary{arrows}
\usetikzlibrary{shapes}
\usetikzlibrary{arrows,calc,backgrounds,shapes,shapes.geometric,patterns,positioning,fit,decorations.pathmorphing,decorations.pathreplacing,shadows}

\usepackage{amssymb, latexsym, marvosym}

\usepackage{thmtools}
\usepackage{thm-restate}

\usepackage[calc]{adjustbox}
\usepackage{wrapfig}

\colorlet{softyellow}{yellow!15!white}
\colorlet{softorange}{orange!15!white}
\sethlcolor{softyellow}

\newcommand{\theo}{\begin{theorem}}
\newcommand{\etheo}{\end{theorem}}

\newcommand{\lem}{\begin{lemma}}
\newcommand{\elem}{\end{lemma}}
\newcommand{\elemma}{\end{lemma}}

\newcommand{\corol}{\begin{corollary}}
\newcommand{\ecorol}{\end{corollary}}

\theoremstyle{definition}
\newtheorem{defn}[theorem]{Definition}
\newcommand{\bdefn}{\begin{defn}}
\newcommand{\edefn}{\end{defn}}

\theoremstyle{plain}
\newcommand{\prop}{\begin{proposition}}
\newcommand{\eprop}{\end{proposition}}

\newtheorem{Conjecture}[theorem]{Conjecture}
\newcommand{\conj}{\begin{Conjecture}}
\newcommand{\econj}{\end{Conjecture}}

\newtheorem{Example}[theorem]{\small Example}
\newcommand{\ex}{\begin{Example}\small \rm}
\newcommand{\eex}{\end{Example}}

\newtheorem{Examples}{\small Examples}
\newcommand{\exx}{\begin{Examples}\small\rm\begin{enum}}
\theoremstyle{claimstyle}
\newtheorem*{observation}{Observation}

\newcounter{enum}
\newenvironment{enum}{\begin{list}{{\upshape(\arabic{enum})}}%
{\setlength{\labelwidth}{5mm}\setlength{\leftmargin}{10mm}%
\setlength{\itemindent}{0pt}\usecounter{enum}}}{\end{list}}

\newcounter{menum}

\newcounter{aenum}

\newcommand{\look}{\begin{proof}}

\newcommand{\hx}{\end{proof}}

\newcommand{\Iff}{\Longleftrightarrow}
\newcommand{\ra}{\rightarrow}
\renewcommand{\iff}{\leftrightarrow}
\newcommand{\E}{\exists}
\newcommand{\A}{\forall}
\renewcommand{\phi}{\varphi}
\renewcommand{\theta}{\vartheta}
\renewcommand{\emptyset}{\varnothing}

\renewcommand{\AA}{{\mathfrak A}}

\newcommand{\MM}{{\mathfrak M}}

\renewcommand{\epsilon}{\varepsilon}

\newcommand{\FO}{{\rm FO}}
\newcommand{\RA}{{\rm RA}}

\newcommand{\LFP}{{\rm LFP}}

\newcommand{\posLFP}{{\rm posLFP}}

\newcommand{\N}{{\mathbb N}}
\newcommand{\Ninf}{{\mathbb N}^{\infty}}
\newcommand{\Sinf}{{\mathbb S}^{\infty}}

\newcommand{\B}{{\mathbb B}}

\newcommand{\W}{{\mathbb W}}

\newcommand{\Gg}{{\cal G}}

\newcommand{\Ss}{{\cal S}}

\renewcommand{\bar}{\overline}
\newcommand{\tup}[1]{\mathbf{#1}}

\newcommand{\cutout}[1]{}

\DeclareMathOperator{\lfp}{\mathbf{lfp}}
\DeclareMathOperator{\gfp}{\mathbf{gfp}}

\renewcommand{\st}{\,.\,} 

\newcommand{\Strat}{\mathrm{Strat}}
\DeclareMathOperator{\Plays}{\mathrm{Plays}}

\newcommand{\On}{\mathrm{On}}

\DeclareMathOperator{\Atoms}{\mathrm{Atoms}}
\DeclareMathOperator{\NegAtoms}{\mathrm{NegAtoms}}

\DeclareMathOperator{\Lit}{\mathrm{Lit}}
\DeclareMathOperator{\nnf}{\mathrm{nnf}}

\newcommand{\Bool}{\mathbb{B}}

\newcommand{\Nat}{\mathbb{N}}

\newcommand{\Trop}{\mathbb{T}}
\newcommand{\Vit}{\mathbb{V}}
\newcommand{\Luk}{\mathbb{L}}

\newcommand{\PosBool}{\mathsf{PosBool}}

\newcommand{\nn}[1]{\bar{#1}}

\newcommand{\nnX}{\nn{X}}
\newcommand{\nnx}{\nn{x}}

\newcommand{\Mod}{\mathsf{Mod}}
\newcommand{\tx}{{\tup x}}
\newcommand{\ty}{{\tup y}}
\newcommand{\ta}{{\tup a}}
\newcommand{\tb}{{\tup b}}
\newcommand{\tuple}[1]{{\tup {#1}}}
\newcommand{\dcup}{\mathbin{\dot{\cup}}}
\newcommand{\bigmid}{\;\big|\;}
\newcommand{\Bigmid}{\;\Big|\;}
\newcommand{\eqInfo}[1]{\overset{\clap{\scriptsize #1}}{=}}

\newcommand{\eqIH}{\eqInfo{IH}}

\newcommand{\Inf}{\bigsqcap}
\newcommand{\Sup}{\bigsqcup}
\newcommand{\join}{\sqcup}

\newcommand{\mxls}{\mathrm{maximals}\,}
\newcommand{\ps}[1]{[\![ #1 ]\!]}

\newcommand{\ext}[1]{[\![ #1 ]\!]}

\newcommand{\mgpi}{\pi^{\star}} % most general K-interpretation

\newcommand{\lfpfml}[4]{[\lfp{} {#1} \, {#2}. \; {#3}]({#4})}
\newcommand{\gfpfml}[4]{[\gfp{} {#1} \, {#2}. \; {#3}]({#4})}
\newcommand{\Lfpfml}[4]{\big[\!\lfp{} {#1} \, {#2}. \; {#3}\big]({#4})}
\newcommand{\Gfpfml}[4]{\big[\!\gfp{} {#1} \, {#2}. \; {#3}\big]({#4})}

\newcommand{\lm}{\pi}
\newcommand{\phia}{\phi(\ta)}
\newcommand{\eval}[2]{#1 \ext {#2}}
\renewcommand{\Game}{\Gg} %note: \Game is defined by amssymb

\newcommand{\nbs}{\normalbaselineskip}

\newcommand{\nsmall}{\\[0.3\baselineskip]}

\newcommand{\litcount}[2]{\#_{#1}(#2)}
\newcommand{\Litcount}[2]{\litcount {#1} {#2}}
\newcommand{\Cutsym}{\scalebox{.8}[1]{\text{\raisebox{-1pt}{\Large\Rightscissors}}}}
\newcommand{\Cutsymsmall}{\scalebox{.7}[1]{\text{\raisebox{-1pt}{\small\Rightscissors}}}}
\newcommand{\Pos}{V}
\newcommand{\cut}[1]{|_{#1}} % write e.g. \Ss \cut n
\newcommand{\TT}{\mathcal T}
\newcommand{\WinStrat}[1]{\mathcal{W}(#1)}
\newcommand{\hash}[1]{{#1}^{\#}}
\newcommand{\hashhat}[1]{\widehat{#1}^{\#}}
\newcommand{\prefix}{\sqsubseteq}
\newcommand{\prefixneq}{\sqsubset}

\DeclareMathOperator*{\Prod}{\widehat{\prod}}
\DeclareMathOperator*{\OrdProd}{\widetilde{\prod}}
\newcommand{\finsub}{\subseteq_{\text{fin}}}
\newcommand{\cntf}[2]{\#_{#1}(#2)}
\newcommand{\cnt}[1]{\##1}
\newcommand{\pchain}[1]{\widetilde{#1}}

\colorlet{bggray}{gray!30} %background in tikzpictures (to highlight parts)
\colorlet{wingray}{gray!40} %highlighted edges to indicate winning strategies
\tikzset{
	dot/.style={draw=black, fill=black, circle, inner sep=0pt, minimum size=5pt},
	win/.style={preaction={shorten <=1pt, shorten >=1pt,draw=wingray, -, double=wingray, double distance=5pt}},
	prio/.style={font=\scriptsize},
  p0/.style={draw=black, fill=white, rounded rectangle, inner sep=2pt, minimum size=.55cm, minimum width=1cm},
  p1/.style={draw=black, fill=white, rectangle, inner sep=2pt, minimum size=.55cm},
  tt/.style={draw=gray, fill=white, rectangle, dashed, inner ysep=2pt, inner xsep=4pt, minimum size=.55cm},
}

\title{Generalized Absorptive Polynomials and Provenance Semantics for Fixed-Point Logic}
\author{Katrin M. Dannert}{RWTH Aachen University, Germany}{dannert@logic.rwth-aachen.de}{}{Supported by the DFG RTG 2236 UnRAVeL.}
\author{Erich Grädel}{RWTH Aachen University, Germany}{graedel@logic.rwth-aachen.de}{}{}
\author{Matthias Naaf}{RWTH Aachen University, Germany}{naaf@logic.rwth-aachen.de}{}{}
\author{Val Tannen}{Univ.~of~Pennsylvania, U.S.A.}{val@cis.upenn.edu}{}{}

\authorrunning{K. Dannert, E. Grädel, M. Naaf,  and V. Tannen}

\Copyright{Katrin M. Dannert, Erich Grädel, Matthias Naaf, and Val Tannen}

\ccsdesc{Theory of Computation~Finite Model Theory}

\keywords{Finite Model Theory, Semiring Provenance, Absorptive Semirings, Fixed-Point Logics}

\nolinenumbers %uncomment to disable line numbering
\hideLIPIcs  %uncomment to remove references to LIPIcs series (logo, DOI, ...), e.g. when preparing a pre-final version to be uploaded to arXiv or another public repository

\begin{document}

\maketitle

\begin{abstract} 
Semiring provenance is a successful approach, originating in database
theory, to providing detailed information on how atomic facts combine
to yield the result of a query. In particular, general provenance
semirings of polynomials or formal power series provide precise
descriptions of the evaluation strategies or ``proof trees'' for the
query. By evaluating these descriptions in specific application
semirings, one can extract practical information for instance about
the confidence of a query or the cost of its evaluation.

This paper develops semiring provenance for very general logical
languages featuring the full interaction between negation and
fixed-point inductions or, equivalently, arbitrary interleavings of 
least and greatest fixed points. This also opens the door to provenance analysis applications for
modal $\mu$-calculus and temporal logics, as well as for finite and
infinite model-checking games.

Interestingly, the common approach based on Kleene's Fixed-Point Theorem for $\omega$-continuous semirings is not 
sufficient for these general languages.  We
show that an adequate framework for the provenance analysis of full
fixed-point logics is provided by semirings that are (1) \emph{fully continuous}, and (2) \emph{absorptive}. 
Full continuity guarantees
that provenance values of least and greatest fixed-points are
well-defined. Absorptive semirings provide a symmetry between least
and greatest fixed-points and make sure that provenance values of
greatest fixed points are informative.

We identify \emph{semirings of generalized absorptive polynomials} $\Sinf[X]$ and 
prove universal properties that make them the most general appropriate semirings for our framework.
These semirings have the further property of being (3) \emph{chain-positive}, which is responsible for having \emph{truth-preserving} interpretations
that give non-zero values to all true formulae.
We relate the provenance analysis of fixed-point formulae with
provenance values of plays and strategies in the associated
model-checking games. Specifically, we prove that the provenance value
of a  fixed point formula gives precise information on the
evaluation strategies in these games.
\end{abstract}

\newpage
\section{Introduction}

Provenance analysis for a logical statement $\psi$, evaluated on a finite structure $\AA$, aims at providing precise
information why $\psi$ is true or false in $\AA$. The approach of \emph{semiring provenance}, going back to
\cite{GreenKarTan07} relies on the idea of annotating the atomic facts %not just by true or false, but 
by values from a commutative semiring,
and to propagate these values through the statement $\psi$, keeping track whether information is used
alternatively (as in disjunctions or existential quantifications) or jointly (as in conjunctions or universal quantifications).
Depending on the chosen semiring, the provenance value  may then give practical information
for instance concerning the \emph{confidence} we may have that $\AA\models\psi$, the \emph{cost} of the evaluation
of $\psi$ on $\AA$, the number of successful evaluation strategies for $\psi$ on $\AA$ in a game-theoretic sense,
and so on. Beyond such provenance evaluations in specific application semirings, 
more general and more precise information is obtained by evaluations in so-called \emph{provenance semirings}
of polynomials or formal power series. Take, for instance, an abstract set $X$ of \emph{provenance tokens}
that are used to label the atomic facts of a structure $\AA$, and consider the semiring $\N[X]$ of polynomials
with indeterminates in $X$ and coefficients from $\N$, which is the 
commutative semiring that is freely generated (`most general') over $X$.
Such a labelling of the atomic facts then extends to a provenance valuation $\pi\ext{\psi}\in\N[X]$ for
every Boolean query $\psi$ from positive relational algebra $\RA^+$ and, indeed, every negation-free first-order sentence $\psi\in\FO^+$.
This provenance valuation gives precise information 
about the combinations of atomic facts that imply the truth of $\psi$ in $\AA$. 
Indeed, we can write $\pi\ext{\psi}$ as a sum of monomials $m\: x_1^{e_1}\cdots x_k^{e_k}$.
Each such monomial indicates that we have precisely $m$ evaluation strategies (or `proof trees') to determine that $\AA\models\psi$
that make use of the atoms labelled by $x_1,\dots x_k$, and the atom labelled by $x_i$ is used precisely $e_i$
times by the strategy, see \cite{GreenKarTan07,GraedelTan17}. 
 
\smallskip\noindent{\bf Provenance for least fixed points. } 
A similar analysis has been carried out for Datalog \cite{DeutchMilRoyTan14,GreenKarTan07}. Due to the need of unbounded least fixed-point iterations 
in the evaluation of Datalog queries, the underlying semirings have to satisfy the additional property of being $\omega$-continuous.
By Kleene's Fixed-Point Theorem, systems of polynomial equations then have least fixed-point solutions that can 
be computed by induction, reaching the fixed-point after at most $\omega$ stages. 
Most of the common application semirings are $\omega$-continuous, or can easily be extended to one that is so; 
however, the most general $\omega$-continuous provenance semiring over $X$  is no longer a semiring of polynomials but
the semiring of formal power series over $X$, denoted $\N^\infty[\![X]\!]$, with coefficients in $\Ninf\coloneqq\N\cup\{\infty\}$.
As above, provenance valuations $\pi\ext{\psi}\in\N^\infty[\![X]\!]$ give precise information about the possible
evaluation strategies for a Datalog query $\psi$ on $\AA$. Even though $\AA$ is assumed to be finite there
may be infinitely many such strategies, but each of them can use each atomic fact only a finite number of times; to put it 
differenty, `proof trees' for $\AA\models\psi$ are still finite. This is closely related to the provenance analysis of reachability games
on finite graphs \cite{DeutchMilRoyTan14,GraedelTan20}.

\smallskip\noindent{\bf Negation: a stumbling block for wider applications.} 
Semiring provenance has been applied to a number of other scenarios, such as nested relations, XML, 
SQL-aggregates, graph databases (see, e.g., the survey \cite{GreenTan17} as well as \cite{RamusatManSen18,Senellart17}),
and it is fair to say that in databases, semiring provenance analysis has been rather successful. 
However, its impact outside of databases has been very limited, despite the fact that the main
questions addressed by provenance analysis, namely which parts of a large heterogeneous input structure are responsible 
for the evaluation of a logical statement, and the applications to cost, confidence, access control and so on are 
clearly interesting and relevant in many other branches of logics 
in computer science.
The  main obstacle for extending semiring provenance to such fields have been difficulties with handling negation. 
For a long time, semiring provenance has essentially been restricted to negation-free query languages,
and although there have been algebraically interesting attempts to cover difference of relations
\cite{AmsterdamerDeuTan11,GeertsPog10,GeertsUngKarFunChr16,GreenIveTan09}, 
they have not resulted in systematic tracking of \emph{negative information}.
While there are many applications in databases where one can get quite far with
using positive information only, logical applications in most other areas are based on
formalisms that use negation in an essential way, often in combination with recursion or fixed-points.

\smallskip\noindent{\bf Provenance semirings for logics with negation and recursion. }
This paper is part of larger project with the objectives to 
\begin{itemize}
\item develop semiring provenance
systematically for a wide range of logics, including those 
featuring the notoriously difficult interaction between full
negation and recursion,
\item to employ algebraic methods for provenance analysis,
in particular universal semirings of polynomials to obtain the most general provenance information,
\item to exploit the connections between logics and various kinds of games
and to use semiring valuations for an analysis of strategies in such games, and
\item to explore practical applications of semiring provenance in new areas of logics
in computer science, where this has not been used so far, such as knowledge representation,
verification, and machine learning.
\end{itemize}
This  project has been initiated in  \cite{GraedelTan17}, where a provenance analysis of
full first-order logic has been proposed. In this approach, negation is dealt with by transformation into
negation normal form%
\footnote{Of course, transformation to negation normal form is a common approach in logic.
But while this is often just a matter of convenience and done for simplification, its seems indispensable
for provenance semantics.
Indeed, beyond Boolean semantics, negation is not a compositional logical operation:
the provenance value of $\neg\phi$ is not necessarily determined by the  provenance value of $\phi$.}
and, algebraically,  by new provenance semirings  
of dual-indeterminate polynomials,
which are obtained by taking quotients of traditional semirings
of polynomials, such as $\N[X]$ 
by congruences  generated by products of positive and negative provenance tokens,
see Sect.~\ref{Sect:prelim} for details. In particular, the semiring $\N[X,\nnX]$ of dual-indeterminate polynomials is the most general provenance semiring for full first-order logic  $\FO$.  These ideas have been
used in \cite{DannertGra19a,DannertGra19b} to provide a provenance analysis of modal and guarded
fragments of first-order logic, and to explore applications in description logic.
Further, this approach has been 
applied to database repairs in \cite{XuZhaAlaTan18}, and it has been shown how 
this treatment of negation, or absent information,
can be used to explain and repair missing query answers and the failure of integrity constraints in 
databases.

While the connection between provenance analysis of first-order logic and
semiring valuations of games had only been hinted at in \cite{GraedelTan17}, 
it has then been developed more systematically in \cite{GraedelTan20}, first for games
on acyclic graphs, which admit only finite plays, and then also for reachability games
on acyclic game graphs. The latter are tightly connected with least fixed-point inductions,
used positively. Combining the approach from \cite{GraedelTan17}
with the provenance analysis of least fixed-point inductions in $\omega$-continuous semirings 
of formal power series, one obtains, by an analogous quotient construction, 
the semiring $\N^\infty[\![X,\nnX]\!]$ of dual-indeterminate
power series \cite{GraedelTan20}. This is the most general provenance semiring for Datalog with negated input predicates and, more generally, also for 
$\posLFP$, the fragment of full fixed-point logic that consists of formulae in negation normal form   
such that all its fixed-point operators are least fixed-points. This is a powerful fixed-point calculus,
which suffices to capture all polynomial-time computable properties of ordered finite structures \cite{Graedel+07}. 
An important simplification of dealing with $\posLFP$ is that the game-based analysis of model checking 
only requires reachability games rather than the much more complicated parity games that are needed
for full LFP. At the end of \cite{GraedelTan20} the problem of generalising semiring valuations and strategy analysis to infinite games with more general objectives than reachability has been discussed.
In particular, a provenance approach for safety games has been proposed, 
with absorptive semirings as the central algebraic tools, and absorption-dominant strategies as
a relevant game-theoretic notion.  

\smallskip\noindent{\bf  Greatest fixed points.}
What has been missing so far, and what we want to provide in this paper, 
is an adequate and systematic treatment of greatest fixed points. There is a strong motivation for this:
If provenance analysis  should ever have an impact in fields such as verification (and we strongly believe it should) 
then dealing with greatest fixed points, e.g. for safety conditions or bisimulation, and with alternations between least and greatest fixed points is indispensable. 
The relevant formalisms in verification (such as LTL, CTL, mu-calculus etc.) are negation closed and based on both 
least and greatest fixed-points, with strict alternation hierarchies (even for finite structures), 
and without possibilities to eliminate greatest fixed-points.
Even in finite model theory, where greatest fixed points can in principle be eliminated from LFP by means of the Stage Comparison Theorem  \cite{Moschovakis74,Graedel+07}, it is usually not desirable to do so. Natural properties involving greatest fixed points (such as bisimilarity)
would become very complicated to express, with the need to double the arity of the fixed-point variables. 
In addition, provenance valuations provide a refined semantics, and statements that are equivalent in the Boolean sense need not  have the same provenance value. Therefore we here do not propose an approach
that first tries to simplify formulae (e.g. by eliminating fixed-point alternations) and
then computes semiring valuations for the translated formulae, but instead
lay  foundations of a provenance analysis for the general logics with arbitrary interleavings of least and greatest fixed points, such as full LFP  or the modal $\mu$-calculus (and for infinite games with more general objectives than reachability).

\smallskip\noindent{\bf  Provenance semirings for arbitrary fixed points.} 
We first address the question, what kind of semirings are adequate for a meaningful and informative
provenance analysis of unrestricted fixed point logics (Sect.\ \ref{Sect:Semirings}). The common approach 
for dealing with least fixed point inductions, based on $\omega$-continuous semirings 
and Kleene's Fixed-Point Theorem, is not sufficient to
guarantee that both least and greatest fixed point  are well-defined.
Instead, we require that the semirings are \emph{fully continuous} which means that every
chain $C$ has not only a supremum $\Sup C$, but also an infimum $\Inf C$,
and that both semiring operations are compatible with these suprema and infima.
For an informative provenance semantics, there is a second important condition that is 
connected with the symmetry between least and greatest fixed point computations.
In the Boolean setting, fixed-point logic is based on complete lattices
which are inherently symmetric. Moreover, conjunction and disjunction are dual
in the sense that one leads to larger lattice elements while the other is decreasing.
In the semiring setting, we compute fixed points with respect to the natural order induced by addition.
The only constraint that relates this order with multiplication is distributivity, but this alone does not suffice to ensure a similar duality.
We achieve this by requiring that the semiring is \emph{absorptive}.
This means that $a+ab=a$ for all $a,b$, and we shall see that %in naturally ordered semirings,
this is equivalent with $1$ being the greatest element or with multiplication being decreasing, 
giving us the desired duality with $0$ and addition.
As a result, absorptive and fully continuous semirings guarantee a well-defined and informative 
provenance semantics for arbitrary fixed-point formulae.

\smallskip\noindent{\bf Generalized absorptive polynomials.}
For a most general provenance analysis, we further want the semiring semantics to be \emph{truth-preserving}, which means that it gives non-zero values to true formulae.
In positive semirings, this is guaranteed if infima of non-zero values are also non-zero, which we call \emph{chain-positivity}.
Our fundamental examples of absorptive, fully continuous, and chain-positive semirings are the semirings $\Sinf[X]$ of \emph{generalized absorptive polynomials}
and its dual-indeterminate version $\Sinf[X,\nnX]$, as introduced in \cite{GraedelTan20}.  Informally such a polynomial is
a sum of monomials, with possibly infinite exponents, that are maximal with respect to absorption.
For instance a monomial $x^2y^\infty z$ occurring in a provenance value $\pi\ext{\psi}$ indicates
an absorption-dominant evaluation strategy that uses the atom labelled by $x$ twice,
the atom labelled by $y$ an infinite number of times, and the atom labelled by $z$ once. 
This monomial absorbs all those that have larger exponents for all variables, such as for instance $x^3y^\infty z^\infty u$,
but not, say, $x^\infty y^3$. Absorptive polynomials thus describe model-checking proofs or evaluation strategies
with a minimal use of atomic facts. 
A precise definition and analysis of these semirings will be given in Sect.~\ref{Sect:Sinf}.
We prove that they do indeed have universal properties (see Theorem~\ref{universality-of-Sinf}) that make $\Sinf[X, \nnX]$
the most general absorptive semiring for LFP
and thus also an indispensable tool to prove general results
about provenance semantics in absorptive, fully continuous semirings.

\smallskip\noindent{\bf Game-theoretic analysis.} In the final Sect.~\ref{Sect:Games} we illustrate the power of provenance interpretations for LFP in absorptive, fully-continuous semirings, and particularly in $\Sinf[X,\nnX]$ by relating them to provenance values of plays and strategies in the associated
model-checking games which in this case are parity games.
Specifically we prove that, as in the case of $\FO$ and $\posLFP$, the provenance value of an LFP-formula $\phi$
gives precise information on the evaluation strategies in these games.

\section{Preliminaries: Commutative Semirings}\label{Sect:prelim}

\bdefn A \emph{commutative semiring} is an algebraic structure 
$(K,+,\cdot,0,1)$, with $0\neq1$,  such that $(K,+,0)$
and $(K,\cdot,1)$ are commutative monoids, $\cdot$
distributes over $+$, and $0\cdot a=a\cdot 0=0$. 
It is \emph{naturally ordered} if the relation
$a \le b :\Iff a + c = b$ for some $c \in K$ is a partial order.
Further, a commutative semiring is \emph{positive} if 
%+-positive 
$a+b=0$ implies $a=0$ and $b=0$
and if it has no divisors of 0 (i.e., $a\cdot b=0$ implies
that $a=0$ and $b=0$).
\edefn
All semirings considered in this paper are commutative
and naturally ordered (which excludes rings). In the following we just write `semiring'
to denote a commutative, naturally ordered semiring.
Standard semirings considered in
provenance analysis are in fact also positive, 
but for an appropriate treatment 
of negation we need semirings (of dual-indeterminate polynomials
or power series) that have divisors of 0.
Notice that a semiring $K$ is positive if, and only if, 
the unique function $h: K\ra \{0,1\}$ with 
$h^{-1}(0)=\{0\}$ is a homomorphism into the Boolean semiring $\Bool$ defined below.

Elements of semirings will be used as truth values for
logical statements.
The intuition is that + describes the \emph{alternative use}
of information, as in disjunctions or existential quantifications,
whereas $\cdot$ stands for the \emph{joint use} of information,
as in conjunctions or universal quantifications. Further, 0 is the value
of false statements, whereas any element $a\neq 0$ of a semiring $K$ stands
for a ``nuanced''  interpretation of true.   
We briefly discuss some specific semirings that provide interesting
information about a logical statement.
\begin{itemize}
\item The \emph{Boolean semiring} $\Bool=(\{0,1\},\vee,\wedge,0,1)$ is the standard habitat of 
logical truth. 
\item $\Nat=(\Nat,+,\cdot,0,1)$ is used for counting evaluation strategies for a logical statement.
\item $\Trop=(\mathbb{R}_{+}^{\infty},\min,+,\infty,0)$ 
is called the \emph{tropical} semiring. It can be used for measuring the cost of 
evaluation strategies.
\item The \emph{Viterbi} semiring $\Vit=([0,1],\max,\cdot,0,1)$
is used to compute \emph{confidence scores}
for logical statements. It is in fact isomorphic to $\Trop$.
\item The \emph{min-max} semiring on a totally ordered set $(A,\leq)$
with least element $a$ and greatest  element $b$
is the semiring $(A,\max,\min,a,b)$. 
\end{itemize}

Beyond these \emph{application semirings}, 
(most general) abstract provenance 
can be calculated in freely generated (universal)
%there are  important universal
\emph{provenance semirings} of polynomials or formal power series.
The abstract provenance can
%in a general semiring
then be specialised 
via homomorphisms to provenance values in
different application semirings as needed.
\begin{itemize}
\item For any set $X$, the semiring $\Nat[X]=(\Nat[X],+,\cdot,0,1)$
consists of the multivariate polynomials in indeterminates from $X$
with coefficients from $\Nat$. This is the commutative
semiring freely generated by $X$. 
Admitting also infinite sums of monomials we obtain the 
semiring $\N^\infty[\![X]\!]$ of formal power series over $X$, with coefficients in $\Ninf\coloneqq\N\cup\{\infty\}$.
\item Given two disjoint sets $X,\nnX$ of ``positive''  and ``negative''
provenance tokens, together with a one-to-one correspondence $X\leftrightarrow\nnX$, mapping each 
positive token $x$ to its corresponding negative token $\nnx$, the
semiring $\N[X,\nnX]$ is the quotient
of the semiring of polynomials $\N[X\cup\nnX]$ by the
congruence generated by the equalities $x\cdot\nnx=0$ for
all $x \in X$.  This is the same as quotienting by the ideal generated 
by the polynomials $x\nnx$ for all $x \in X$.
%Two polynomials $g,g'\in\N[X\cup \nnX]$
%are congruent if, and only if, they become identical after deleting
%the monomials that contain complementary tokens. Hence, 
The congruence classes in $\N[X,\nnX]$ are in one-to-one correspondence
with the polynomials in $\N[X\cup\nnX]$ such that none of their monomials contain 
complementary tokens. We call these \emph{dual-indeterminate polynomials}.
$\N[X,\nnX]$ is freely generated by $X\cup\nnX$ for homomorphisms such
that $h(x)\cdot h(\nnx)=0$.
By a completely analogous quotient construction, we obtain
the semiring $\N^\infty[\![X,\nnX]\!]$ of dual-indeterminate power series.
\item By dropping coefficients from $\Nat[X]$, we get the semiring
$\B[X]$ whose elements are just finite sets
of distinct monomials. It is the free idempotent semiring over $X$.
By dropping also exponents, we get the semiring $\W[X]$ of 
finite sums of monomials that are linear in each argument.
It is sometimes called the Why-semiring.
\item The semiring $(\PosBool(X), \lor, \land, \text{false}, \text{true})$
consists of the positive Boolean expressions over the variables $X$,
where we identify logically equivalent expressions.
\end{itemize}

\section{Provenance Semantics for Fixed-Point Logic}\label{Sect:LFP}

Semiring provenance is well understood for first-order logic
and for logics with only least fixed-points, used positively.
To extend it to logics with arbitrary interleavings of least and greatest fixed points,
we discuss the general fixed-point logic LFP that extends first-order logic by 
least and greatest fixed-point operators, but our insights also apply to weaker logics such as 
the modal $\mu$-calculus, dynamic logics, or temporal logics such as CTL.

\medskip\noindent{\bf Least Fixed-Point Logic.}
%Throughout this paper, we assume a finite relational vocabulary $\tau$ and only consider finite structures.
Least fixed-point logic, denoted LFP,   
extends first order logic by 
least and greatest fixed points of definable 
monotone operators on relations:
If  $\psi(R,\tup x)$ is a formula 
of vocabulary $\tau\cup\{R\}$, in which the  relational variable $R$ occurs only positively and 
the length of $\tup x$ matches the arity of $R$, then 
$[\lfp R\tup x\st \psi](\tup x)$ and $[\gfp R\tup x \st\psi](\tup x)$
are also formulae (of vocabulary $\tau$).
The semantics of these formulae is that $\tup x$ is contained in the least
(respectively the greatest) fixed point of the update operator
$F_\psi:R\mapsto\{\tup a:  \psi(R,\tup a)\}$. Due to the positivity of $R$ in $\psi$,
any such operator $F_\psi$ is monotone and has, by the Knaster-Tarski-Theorem, a least 
fixed point $\lfp(F_\psi)$ and a greatest fixed point $\gfp(F_\psi)$. 
See e.g.\ \cite{Graedel+07} for background on $\LFP$.
The duality between least and greatest fixed points implies
that 
$[\gfp R\tup x\st \psi](\tup x)\equiv \neg [\lfp R\tup x\st \neg\psi[R/\neg R]](\tup x)$.
By this duality together with de Morgan's laws, every LFP-formula can
be brought into \emph{negation normal form}, where negation applies to
atoms only.
The fragment $\posLFP$ of LFP consists of the formulae in negation normal form
in which all fixed-point operators are least fixed-points.
It is well-known that LFP, and even $\posLFP$,  captures all polynomial-time computable
properties of ordered finite structures \cite{Graedel+07}.

\medskip\noindent{\bf Provenance Semantics.}
Instead of truth-values, we now assign semiring values to literals.
For a finite universe $A$ and a finite relational vocabulary $\tau$  we denote the set of atoms as
$\Atoms_A(\tau)=\{R\tup{a}\colon R\in\tau,\ \tup{a}\in A^{\text{arity}(R)}\}$.
The set $\NegAtoms_A(\tau)$ contains all
negations $\neg R\tup{a}$ of atoms in $\Atoms_A(\tau)$ and we define the set of 
$\tau$-literals on $A$ as
\begin{align*}
 \Lit_A(\tau)\coloneqq\Atoms_A(\tau)\cup\NegAtoms_A(\tau)\cup \{a=b\colon a,b\in A\}
 \cup \{a\neq b\colon a,b\in A\}.
\end{align*}

\bdefn
 For any semiring $K$, a $K$-\emph{interpretation}
 (for $\tau$ and $A$) is a function $\pi\colon \Lit_A(\tau)\rightarrow K$ 
mapping true equalities and inequalities to $1$ and false ones to $0$.
\edefn

We can extend $K$-interpretations $\pi$ to provide provenance values $\pi \ext \phi$
for any first-order formula $\phi$ in a natural way \cite{GraedelTan17},
by interpreting disjunctions and existential quantification via addition,
and conjunctions and universal quantification via multiplication. Negation is not interpreted directly by an algebraic operation. 
We deal with it syntactically, by evaluating the negation normal form $\nnf(\psi)$ instead.
To interpret fixed-point formulae
$[\lfp R\tup x\st \psi](\tup a)$ and $[\gfp R\tup x \st\psi](\tup a)$,
%we further assume that $K$ is \emph{naturally ordered} and 
we generalize the update operators $F_\psi$ to semiring semantics.
If $R$ has arity $m$, then its $K$-interpretations on $A$ are functions $g:A^m\ra K$.
These functions are ordered, by $g\leq g'$ if, and only if, $g(\tup a)\leq g'(\tup a)$
for all $\tup a\in A^m$ (recall that our semirings are naturally ordered).
Given a $K$-interpretation $\pi\colon\Lit_A(\tau)\ra K$, we denote by $\pi[R\mapsto g]$
the $K$-interpretation of $\Lit_A(\tau)\cup \Atoms_A(\{R\})$ obtained from $\pi$ by
adding values $g(\tup c)$ for the atoms $R\tup c$. (Notice that $R$ appears only positively in $\phi$,
so negated $R$-atoms are not needed).
The formula $\phi(R,\tup x)$ now defines, together with $\pi$, a monotone update
operator $F_\pi^\phi$ on functions $g: A^m\ra K$. More precisely, it maps
$g$ to the function
\[   F_\pi^\phi(g) \colon \; \tup a\mapsto   \pi[R\mapsto g]\ext{\phi(R,\tup a)}. \]

We obtain a well-defined provenance semantics for LFP if we can make sure that
the update operators $F_\pi^\phi$ have least and greatest fixed-points
$\lfp(F_\pi^\phi)$, $\gfp(F_\pi^\phi) \colon A^m\ra K$.
However, this is not guaranteed in all semirings, and also the common approach to least fixed-point inductions
based on $\omega$-continuous semirings is not sufficient here,
as these, in general, do not guarantee the existence of \emph{greatest} fixed points.
This raises the fundamental question: which semirings are really appropriate for LFP? 
We shall discuss this in detail in the next section. Once we have fixed a notion of appropriate semirings for LFP,
we obtain a provenance semantics for LFP as follows.

\bdefn
 A $K$-interpretation $\pi\colon \Lit_A(\tau)\rightarrow K$ 
 in an \emph{appropriate} semiring $K$
 extends to a $K$-\emph{valuation}
 $\pi\colon \LFP(\tau)\rightarrow K$ by mapping an $\LFP$-sentence $\psi(\tup{a})$ to a value $\pi\ext\psi$ using the following rules
 \begin{align*}
  &\pi\ext{\psi\vee\phi}\coloneqq\pi\ext\psi + \pi\ext\phi & &\pi\ext{\psi\wedge\phi}\coloneqq\pi\ext\psi \cdot \pi\ext\phi 
  & & \pi\ext{\exists x\psi(x)}\coloneqq\sum_{a\in A}\pi\ext{\phi(a)}\\
  &\pi\ext{\forall x\psi(x)}\coloneqq\prod_{a\in A}\pi\ext{\phi(a)} & &
  \mathrlap{\pi \ext{[\lfp R\tup x. \phi(R,\tup x)](\tup a)} \coloneqq \lfp(F_\pi^\phi)(\tup a)} \\
  &\pi\ext{\neg\psi}\coloneqq\pi\ext{\nnf(\psi)} & &
  \mathrlap{   \pi \ext{[\gfp R\tup x. \phi(R,\tup x)](\tup a)} \coloneqq \gfp(F_\pi^\phi)(\tup a).}
 \end{align*}
\edefn

We remark that there is an important difference between the classical Boolean
semantics and provenance semantics concerning the relationship
of fixed-point logics with first-order logic. The (Boolean) evaluation of a fixed-point
formula on a finite structure is computed by fixed-point inductions 
that terminate after a polynomial number of stages (with respect to the size of the structure).
Hence, on any fixed finite universe, a fixed-point formula can be unraveled to
an equivalent first-order formula. This is not the case for the provenance valuations in
infinite semirings. Even for very simple Datalog queries, a fixed-point induction
need not terminate after a finite number of steps. Provenance valuations 
provide more information that just the truth or falsity of a statement, and
in a general setting, this provenance information, for instance about the
number and properties of successful evaluation strategies, may also be infinite.

\section{Semirings for Fixed-Point Logic}
\label{Sect:Semirings}

Given a naturally ordered semiring $K$,
a \emph{chain} is a totally ordered subset $C \subseteq K$.
For $\circ \in \{ +, \cdot \}$ we  write $a \circ C$ for $\{a \circ c \mid c \in C \}$.
Provided they exist, we write $\Sup C$ and $\Inf C$ for the
\emph{supremum} (least upper bound) and \emph{infimum} (greatest lower bound) of  $C \subseteq K$,
and further $\bot$ and $\top$ for the least and greatest elements of $K$.
We say that a function $f : K_1 \to K_2$ is fully chain-continuous
or, for short, \emph{fully continuous}
if it preserves suprema and infima
of nonempty
chains, i.e., $f(\Sup C) = \Sup f(C)$ and $f(\Inf C) = \Inf f(C)$
for all chains $\emptyset \neq C \subseteq K_1$.

\bdefn
A naturally ordered semiring $K$ is \emph{fully chain-complete}
if every chain $C \subseteq K$ has a supremum $\Sup C$
and an infimum $\Inf C$ in $K$.
It is additionally \emph{fully continuous} if its operations are fully continuous in both arguments,
i.e., 
$a \circ \Sup C = \Sup (a \circ C)$ and
$a \circ \Inf C = \Inf (a \circ C)$
for all $a \in K$,
chains $\emptyset \neq C \subseteq K$ and
$\circ \in \{ +, \cdot \}$.
\edefn

Examples of fully continuous semirings include
the Viterbi semiring, $\N^\infty$ and formal power series 
$\N^\infty[\![X]\!]$ and $\N^\infty[\![X,\nnX]\!]$.
For positive least fixed-point inductions, as in Datalog \cite{GreenKarTan07} or $\posLFP$ \cite{GraedelTan20},  the common approach is to 
use \emph{$\omega$-continuous} semirings. There, only suprema of $\omega$-chains are required and both operations
must preserve suprema. It would be tempting to work with a minimal generalization that imposes similar properties
for descending $\omega$-chains, using a dual version of Kleene's Fixed-Point Theorem.
However the following example shows that this approach will not work in general
with alternating fixed points.

\ex
\label{ex:noncontinuousLukasiewicz}
Let $K$ be a naturally ordered semiring that has both suprema of ascending
$\omega$-chains and infima of descending $\omega$-chains and let
$f : K \times K \to K$ be a function
that preserves these suprema and infima in each argument.
For each $x \in K$, we can consider the function $g_x : K \to K$, 
$g_x(y) = f(x, y)$
and, further, the function $G : K \to K$, $G(x) = \gfp(g_x)$.
Note that $G$ is well-defined due to
the preservation property of $f$ and a dual version of Kleene's Fixed-Point Theorem.
Now consider $\lfp(G)$. 
To guarantee the existence of this fixed point via Kleene's theorem, 
$G$ has to preserve suprema of $\omega$-chains.
This is, however, not the case, in general.
One counterexample is the the function $f(x,y) = x \diamond y$
in the (fully continuous) {\L}ukasiewicz semiring $\Luk = ([0,1], \max, \diamond, 0, 1)$
with $a \diamond b = \max(0, a + b - 1)$ on the $\omega$-chain $(x_n)_{n < \omega}$ defined by $x_n = 1 - \frac 1 {1+n}$.
Then $G(\Sup_{n < \omega} x_n) = G(1) = \gfp(g_1) = 1$,
whereas $\Sup_{n < \omega} G(x_n) = \Sup_{n < \omega} \gfp\big(g_{x_n}\big) = \Sup_{n < \omega} 0 = 0$.
\eex

Instead, we rely on $K$ being fully chain-complete
to guarantee the existence of fixed points of monotone functions.
We can then extend \cite{Moschovakis74} the Kleene iteration
$\bot$, $f(\bot)$, $f^2(\bot)$, $f^3(\bot)$, $\dots$ for $\lfp(f)$
to a transfinite fixed-point iteration $(x_\beta)_{\beta \in \On}$ by setting
$x_0 = \bot$, $x_{\beta+1} = f(x_\beta)$ for ordinals $\beta$ and
$x_\lambda = \Sup \{ x_\beta \mid \beta < \lambda \}$ for limit ordinals $\lambda$.
If $f$ is monotone, this iteration forms a chain and is well-defined due to the chain-completeness of $K$.
The iteration for $\gfp(f)$ can be defined analogously
by %using the infimum 
$x_\lambda = \Inf \{ x_\beta \mid \beta < \lambda \}$ for limit ordinals
and it follows that both $\lfp(f)$ and $\gfp(f)$ exist in fully chain-complete semirings.

\prop
\label{propMonotoneFixpoint}
For a monotone function $f : K \to K$ on a fully chain-complete semiring,
both $\lfp(f)$ and $\gfp(f)$ exist.
\eprop

\begin{proof}
	Consider the fixed-point iteration $(x_\beta)_{\beta \in \On}$ for $\lfp(f)$ defined above.
	As $K$ is a set, there must be an ordinal $\alpha \in \On$ with
	$x_\alpha = x_{\alpha+1} = f(x_\alpha)$, so $x_\alpha$ is a fixed point of $f$.
	To see that $x_\alpha$ is the least fixed point,
	let $x'$ be any fixed point of $f$.
	Clearly, $\bot \le x'$ and, by monotonicity, $f(\bot) \le f(x') = x'$.
	By induction, it follows that $x_\beta \le x'$ for all 
	$\beta \in \On$.
	The proof for $\gfp(f)$ is analogous.
\end{proof}

Coming back to the question of appropriate semirings for LFP,
we observe that the monotonicity of the semiring operations $\cdot$ and $+$
lifts to monotonicity of update operators $F_\pi^\phi$.
Hence \cref{propMonotoneFixpoint} ensures that their least and greatest fixed points always exist.

\begin{theorem}
\label{propSemanticsWelldefined}
Semiring semantics for $\LFP$ is well-defined in fully chain-complete semirings.
\end{theorem}

\begin{proof}
Clearly, the semantics of FO operators ($\lor$, $\land$, $\E$, $\A$) are well-defined
(for quantifiers, recall that we assume a finite universe and thus only have finite sums and products).
What remains to prove is that the fixed points $\lfp(F_\pi^\phi)$ and $\gfp(F_\pi^\phi)$ are well-defined.
Recall that an update operator $F_\pi^\phi$ does not operate on the semiring $K$,
but on functions $A^k \to K$.
These functions form a semiring under pointwise operations
that inherits most of the properties from $K$.
Most importantly, it inherits chain-completeness and continuity.
By \cref{propMonotoneFixpoint}, it thus suffices to prove
that update operators $F_\pi^\phi$ are always monotone.

Towards the proof, we say that
$\pi \ext \phi$ is monotone in $\pi$, if $\pi_1 \le \pi_2$ (pointwise comparison)
implies $\pi_1 \ext \phi \le \pi_2 \ext \phi$.
We split the monotonicity proof into two steps.

\begin{claim*}[1]
	Let $K$ be a fully chain-complete semiring and $\theta(R, \tx)$ an LFP-formula.
	If $\pi \ext \theta$ is monotone in $\pi$, then the update operator $F_\pi^\theta$ is monotone.
\end{claim*}

\begin{claimproof}
	Let $k$ be the arity of $R$ and let $g_1, g_2 : A^k \to K$ with $g_1 \le g_2$.
	To simplify notation, let $g_1' = F_\pi^\theta(g_1)$ and $g_2' = F_\pi^\theta(g_2)$.
	Due to $g_1 \le g_2$, we also have $\pi[R/g_1] \le \pi[R/g_2]$. Then $g_1' \le g_2'$, as for all $\ta \in A^k$:
	$
	g_1'(\ta) = {\lm[R/g_1]} \ext {\theta(\ta)}
	\le
	{\lm[R/g_2]} \ext {\theta(\ta)} = g_2'(\ta)
	$,
	due to the monotonicity assumption on $\pi \ext \theta$.
\end{claimproof}

\begin{claim*}[2]
Let $K$ be a fully chain-complete semiring. Then $\lm \ext \phi$ is monotone in $\lm$.
%  , so given two $K$-interpretations $\lm_1$ and $\lm_1$, the following implication holds for all LFP-sentences $\phi(\ta)$:
%	$\lm_1 \le \lm_2 \implies {\lm_1} \ext {\phi(\ta)} \le {\lm_2} \ext {\phi(\ta)}$
\end{claim*}

\begin{claimproof}
	Fix $K$-interpretations $\pi_1 \le \pi_2$. We proceed by induction on the negation normal form of $\phi$.
	\begin{itemize}
		\item For literals, $\pi_1 \ext {R\ta} = \pi_1(R\ta) \le \pi_2(R\ta) = \pi_2 \ext {R \ta}$.
		The same holds for negative literals (and similarly for equality atoms).
		
		\item If $\phi = \phi_1 \lor \phi_2$, then ${\lm_i} \ext \phi = {\lm_i} \ext {\phi_1} + {\lm_i} \ext {\phi_2}$ for $i \in \{1,2\}$.
    By induction, $\pi_1 \ext {\phi_1} \le \pi_2 \ext {\phi_1}$ and $\pi_1 \ext {\phi_2} \le \pi_2 \ext {\phi_2}$.
    The claim then follows by monotonicity of $+$.
    The cases for $\land$, $\E$ and $\A$ are analogous.
		
%		\item If $\phi = \E x. \theta(x)$, then ${\lm_i} \ext \phi = \sum_{a \in A} {\lm_i} \ext {\theta(a)}$.
%		Recall that the universe $A$ is finite, so the claim again follows by induction and monotonicity of $+$. The case for $\A$ is analogous.
		
		\item If $\phi = \lfpfml R \tx \theta \ty$ with $R$ of arity $k$, we proceed by induction on the fixed-point iterations
		$(g_\beta)_{\beta \in \On}$ for $\pi_1$ and $(f_\beta)_{\beta \in \On}$ for $\pi_2$.
    Notice that these are functions $g_\beta, f_\beta \colon A^k \to K$.
		By the induction hypothesis and \textsf{Claim (1)}, $F_{\lm_1}^\theta$ and $F_{\lm_2}^\theta$ are monotone
		and hence the fixed-point iterations are well-defined.
		We prove by induction that $g_\beta \le f_\beta$ for all $\beta \in \On$.
    The proof for $\phi = \gfpfml R \tx \theta \ty$ is completely analogous.
		\begin{itemize}
			\item For $\beta = 0$, we have $g_0,f_0 : A^k \to K,\, \ta \mapsto 0$. In particular, $g_0 \le f_0$.
			
			\item For successor ordinals $\beta+1$, we have $\pi_1[R/g_\beta] \le \pi_2[R/f_\beta]$ by the induction hypothesis for $\beta$.
			Applying the outer induction hypothesis for $\theta$ then yields:
			\[
			g_{\beta+1}(\ta)
			= F_{\lm_1}^\theta(g_\beta)(\ta)
			=   {\lm_1[R/g_\beta]} \ext {\theta(\ta)}
			\le {\lm_2[R/f_\beta]} \ext {\theta(\ta)}
			= F_{\lm_2}^\theta(f_\beta)
			= f_{\beta+1}.
			\]
			
			\item For limit ordinals $\lambda$, we have
			$g_\lambda = \Sup \{ g_\beta \mid \beta < \lambda \} \le \Sup \{ f_\beta \mid \beta < \lambda \} = f_\lambda$
			since we know that $g_\beta \le f_\beta$ for all $\beta < \lambda$.
			
		\end{itemize}
		This ends the induction on $\beta$.
    By choosing a sufficiently large ordinal $\beta$, we can conclude
		\[
		\pi_1 \ext {\phi(\ta)} = \lfp(F_{\pi_1}^\theta)(\ta)
		= g_\beta(\ta) \le f_\beta(\ta)
		= \lfp(F_{\pi_2}^\theta)(\ta) = \pi_2 \ext {\phi(\ta)}. \qedhere
		\]
	\end{itemize}
\end{claimproof}
\noindent
Together, the two claims entail the monotonicity of update operators.
\end{proof}

We further remark that full chain-completeness is more general
than the common notion of complete lattices, used in the Knaster-Tarski fixed-point theory,
as we only require suprema (and infima) of chains instead of arbitrary sets.
However, based on results in \cite{Markowsky76} it follows that the two notions coincide
for the semirings we are interested in.

\prop
\label{propIdempotentLattice}
If $K$ is an idempotent, fully chain-complete semiring,
then its natural order forms a complete lattice,
i.e., suprema and infima of arbitrary sets exist.
\eprop

\begin{proof}
	We first show that addition coincides with finite suprema,
	i.e.\ $a + b = \Sup \{a,b\}$ for $a,b \in K$.
	Clearly, $a \le a+b$ and $b \le a+b$, so $\Sup \{a,b\} \le a+b$.
	The other direction follows from idempotence:
	$a + b \le \Sup \{a,b\} + \Sup\{a,b\} = \Sup \{a,b\}$.
	
	Hence suprema of arbitrary finite sets exist (by summation).
	Due to an old result of Markowsky \cite{Markowsky76}, chain-completeness and
	finite suprema imply the existence of suprema of arbitrary (possibly infinite) sets. Infima can be expressed via suprema, so $K$ forms a complete lattice
	under its natural order.
\end{proof}

The following fundamental property for provenance analysis (cf.\ \cite{GraedelTan17})
establishes a closer connection between logic (the semantics of $\phi$)
and algebra (the semiring homomorphism $h$) and enables us to compute provenance information
in a general semiring and then specialize the result to
application semirings by applying homomorphisms,
most prominently by working with polynomials and applying polynomial evaluation.

\begin{proposition}[Fundamental Property]
\label{propFundamentalProperty}
Let $K_1$, $K_2$ be fully chain-complete semirings
and let $h : K_1 \to K_2$ be a fully continuous semiring homomorphism with $h(\top) = \top$.
Then for every $K_1$-interpretation $\pi$, the mapping $h \circ \pi$ is a $K_2$-interpretation and
for every $\phi\in\LFP$, we have $h(\pi \ext \phi) = (h \circ \pi) \ext \phi$.
\end{proposition}

\noindent
As diagram:\qquad
\begin{tikzpicture}[baseline,node distance=2cm,font=\small]
\node [baseline, anchor=base] (lit) {$\Lit_A(\tau)$};
\node [below left of=lit] (litS) {$K_1$};
\node [below right of=lit] (litT) {$K_2$};

\node [right=6cm of lit.base, anchor=base] (fol) {\normalfont LFP};
\node [below left of=fol] (folS) {$K_1$};
\node [below right of=fol] (folT) {$K_2$};

\node [align=center] at ({$(lit.center)!0.5!(fol.center)$} |- {$(lit.center)!0.5!(litS.center)$}) {$\implies$};

\path[draw,->,shorten <=2pt, shorten >=2pt, font=\scriptsize]
(lit) edge node [above left] {$\pi$} (litS)
(lit) edge node [above right] {$h \circ \pi$} (litT)
(litS) edge node [above] {$h$} (litT)
(fol) edge node [above left] {$\pi$} (folS)
(fol) edge node [above right] {$h \circ \pi$} (folT)
(folS) edge node [above] {$h$} (folT)
;
\end{tikzpicture}

\begin{proof}
  The proof is by induction on the structure of $\phi$.
  For fixed-point formulae, we consider the fixed-point iterations in $K_1$ and $K_2$,
  and we prove that all steps of the iterations are preserved by $h$.
  Here we need the assumption that $h$ is fully continuous.
  Formally, we prove that for all LFP-formulae $\phi(\tx)$ in negation normal form,
 	$h(\pi \ext \phia) = (h \circ \phi) \ext \phia$ holds for all $K$-interpretations $\pi$ and all tuples $\ta$ from the universe $A$.
 	\begin{itemize}
 		\item For literals, we have $h(\pi \ext {R\ta}) = h(\pi(R\ta)) = (h \circ \pi)(R\ta) = (h \circ \pi) \ext {R\ta}$.
 		
 		\item For $\phi = \phi_1 \land \phi_2$ (and, analogously, for $\lor$, $\E$, $\A$) we use that $h$ is a semiring homomorphism:
 		$
 		h(\lm \ext \phi) =
 		h(\lm \ext {\phi_1} \cdot \lm \ext {\phi_2}) =
 		h(\lm \ext {\phi_1}) \cdot h(\lm \ext {\phi_2}) =
 		(h \circ \lm) \ext {\phi_1} \cdot (h \circ \lm) \ext {\phi_2} =
 		(h \circ \lm) \ext \phi
 		$.
 		
 		\item For $\phi = \gfpfml R \tx \theta \ty$ with $R$ of arity $k$,
 		we consider the fixed-point iteration $(g_\beta)_{\beta \in \On}$ for $\pi$ in $K_1$ and
 		the iteration $(f_\beta)_{\beta \in \On}$ for $h \circ \pi$ in $K_2$.
 		We show by induction that $h \circ g_\beta = f_\beta$ for all ordinals $\beta \in \On$,
 		so $h$ preserves all steps of the fixed-point iteration.
 		\begin{itemize}
 			\item For $\beta = 0$, we have $g_0,f_0 : A^k \to K,\, \ta \mapsto \top$.
 			Then $h \circ g_0 = f_0$, as $h(\top) = \top$.
 			
 			\item For successor ordinals, we can apply the induction hypothesis. By definition,
 			\begin{align*}
 			g_{\beta+1}(\ta) = F_\lm^\theta(g_\beta)(\ta) &= {\lm[R/g_\beta]} \ext {\theta(\ta)},\\
 			f_{\beta+1}(\ta) = F_{h \circ \lm}^\theta(f_\beta)(\ta) &= {(h \circ \lm)[R/f_\beta]} \ext {\theta(\ta)}
 			\overset{(*)}{=}
 			{(h \circ \lm[R/g_\beta])} \ext {\theta(\ta)}.
 			\end{align*}
 			In $(*)$, we use the induction hypothesis $h \circ g_\beta = f_\beta$.
 			Using the (outer) induction hypothesis on $\theta$, we obtain
 			\begin{align*}
 			(h \circ g_{\beta+1})(\ta) = h({\lm[R/g_\beta]} \ext {\theta(\ta)}) =
 			{(h \circ \lm[R/g_\beta])} \ext {\theta(\ta)}) = f_{\beta+1}(\ta).
 			\end{align*}
 			
 			\item For limit ordinals, we exploit that $h$ is fully continuous:
 			\begin{align*}
 			h(g_\lambda(\ta)) &= h(\Inf \{ g_\beta(\ta) \mid \beta < \lambda \}) \\
 			&= \Inf \{ h(g_\beta(\ta)) \mid \beta < \lambda \}
 			= \Inf \{ f_\beta(\ta) \mid \beta < \lambda \} = f_\lambda(\ta).
 			\end{align*}
 		\end{itemize}
 		
 		\noindent
 		This closes the proof for $\gfp$-formulae, as for sufficiently large $\beta$, we have
 		\[
 		h(\lm \ext \phia) = h(g_\beta(\ta)) = f_\beta(\ta) = {(h \circ \lm)} \ext \phia.
 		\]
 		The proof for $\lfp$-formulae is analogous. \qedhere
 	\end{itemize}
\end{proof}

\medskip\noindent{\bf Fully continuous semirings. }
While fully chain-complete semirings suffice to guarantee well-defined semantics,
our main results (the universal property in \cref{universality-of-Sinf} and
the connection to games in Sect.~6) require the technically slightly stronger notion
of fully continuous semirings, in which addition and multiplication
preserve suprema and infima of chains.
This is an adaption of the standard notion of $\omega$-continuity to our setting
%of fully chain-complete semirings
and
all natural examples of fully chain-complete semirings we are aware of
are in fact fully continuous.
%, so this is not a severe restriction in practice.
On a different note, the notion of chain-completeness is based on chains of arbitrary length.
We do not know whether working with ascending and descending $\omega$-chains
would suffice in all cases, but we show in Sect.~\ref{Sect:Sinf} that it
suffices in absorptive, fully continuous semirings.

\smallskip\noindent{\bf Absorptive and chain-positive semirings. }
Although the \emph{existence} of fixed points is guaranteed in fully continuous semirings,
we observe (in \cref{ex:infpathViterbi} below) that one may have valuations of greatest fixed-point formulae in such semirings that are not really informative
and do not provide useful insights why a formula holds.
This can be tied to two separate problems: the \emph{lack of symmetry} between least and greatest fixed-point inductions
in some such semirings, and the fact that such semirings are not necessarily \emph{truth-preserving},
i.e.\ they may evaluate true statements to 0.
To deal with these problems we propose to work with fully continuous
semirings that are
\emph{absorptive}, to provide useful provenance information for greatest fixed points,
and  \emph{chain-positive} to guarantee truth-preservation.

We first address the issue of symmetry between least and greatest fixed points.
In the Boolean setting, these are computed in the complete lattice of subsets
which is inherently symmetric.
For instance, a greatest fixed point of a monotone operator
is the complement of the least fixed point of the dual operator
(which is essential for a negation normal form).
Moreover, conjunction and disjunction are symmetric in the sense
that one increases values, acting as set union in the lattice of subsets,
while the other is decreasing.
In the semiring setting, we compute fixed points with respect to the natural order induced by addition.
This order is always a complete lattice in absorptive semirings (in fact, idempotent semirings suffice)
and it is clear that addition is increasing in the sense that $a+b \ge a$ for all $a,b$.
The issue is with multiplication: The only constraint that relates addition and multiplication is distributivity, but this alone does not suffice to ensure a symmetry similar to the Boolean setting.
We achieve this by requiring that the semiring is absorptive.

\bdefn
A semiring $K$ is \emph{absorptive} if $a+ab=a$ for all $a,b \in K$,
which is equivalent to saying that $1+b=1$, for all $b\in K$.
\edefn

Clearly, every absorptive semiring is \emph{idempotent}: $a+a=a$ for all $a$.
For naturally ordered semirings, absorption indeed provides
symmetry: multiplication becomes decreasing and $1$ becomes the greatest element,
symmetric to addition and the least element $0$. 

\begin{proposition}
\label{propAbsorptionEquivalences}
In a naturally ordered semiring $K$,
the following are equivalent:
\begin{enumerate}
\item $K$ is absorptive,
\item $K$ has the greatest element $\top = 1$, i.e., $a \le 1$ for all $a \in K$,
\item multiplication in $K$ is decreasing, i.e., $a \cdot b \le b$ for all $a,b \in K$.
\end{enumerate}
\end{proposition}

\begin{proof}
	If $K$ is absorptive, then $1 + 1 \cdot a = 1$ and hence $a \le 1$ for all $a \in K$.
	Absorption further implies $ab \le a$ for all $a,b \in K$.
	Conversely, $\top = 1$ entails $1 \le 1+a \le 1$, and multiplication with $b$ gives $a + ab = a$.
	If multiplication is decreasing, then $a \ge a \cdot (1+b) = a + ab$.
	Together with $a \le a + ab$ (by natural order), this implies absorption.
\end{proof}

This symmetry helps, for instance,
to avoid problems of increasing multiplication as in $\N^\infty$.
Fixed-point theory often relies on symmetry and it is thus no surprise
that more symmetry leads to more useful provenance information.
This can be seen in the following example when comparing the
computations of greatest fixed-points in the non-absorptive semiring $\N^\infty$ and the more informative Viterbi semiring.

\ex
\label{ex:infpathViterbi}
The existence of an infinite path from $u$ in a graph $G$ is expressed by the LFP-formula
\begin{center}
$\displaystyle \phi(u) = \gfpfml R x {\E y (E x y \land R y)} u$
\hspace{2cm}
\smash{
\begin{tikzpicture}[baseline, yshift=2.5pt]
\node [dot, label={above:$u$}] (u) {};
\node [dot, label={above:$v$}, right of=u] (v) {};
\path [draw, ->, >=stealth', shorten <=2pt, shorten >=2pt]
(u) edge (v)
(v) edge [in=-30, out=35, looseness=15]
%node [right,font=\scriptsize] {$\frac 1 2$}
(v);
\end{tikzpicture}
}
\end{center}

\noindent For the Boolean semiring $\B = \{0,1\}$ there is a unique $\B$-interpretation $\pi$ that defines the displayed graph $G$. 
Provenance semantics in $\B$ coincides with standard semantics and we indeed obtain $\pi \ext{\phi(u)} = 1$.
The Viterbi semiring $\Vit$ instead allows us to assign confidence scores to the edges.
If we set $\pi(E u v) = \pi(E v v) = 1$ as in the Boolean interpretation,
we again obtain an overall confidence of $\pi \ext{\phi(u)} = 1$.
However, if we instead lower the score of the self-loop to $\pi(E v v) = 1 -\epsilon$,
we obtain an overall confidence of $\pi \ext{\phi(u)} = 0$
due to the fixed-point iteration $1, 1-\epsilon, (1-\epsilon)^2, \dots$.
So while $\pi$ still defines the model shown above,
the formula evaluates to $0$ which we usually interpret as \emph{false},
illustrating that the Viterbi semiring is \emph{not truth-preserving}.
Since the loop occurs infinitely often in the unique infinite path from $u$, the value $0$ makes sense 
as a confidence score. Thus, although it is not truth-preserving, the Viterbi semiring does provide useful information.

Consider next the semiring of formal power series $\N^\infty \ps X$.
If we choose $\pi(E u v) = x$ and $\pi(E v v) = y$
(and keep the values $0$ or $1$ for the remaining literals),
%(and map all remaining literals to $0$ or $1$ such that $\pi$ defines the given graph),
then $\pi \ext{\phi(u)} = 0$, as result of
the iteration $\top$, $y \cdot \top$, $y^2 \cdot \top$, $y^3 \cdot \top$, $\dots$
with infimum $0$ at node $v$
(here, $\top$ is the power series in which all monomials have coefficient $\infty$).
Thus, $\N^\infty \ps X$ is \emph{not truth-preserving} either.

In the semiring $\N^\infty$, used
to count proofs of formulae in $\FO$ and $\posLFP$,
the consideration of greatest fixed points imposes problems:
Intuitively, the graph only has one infinite path that we would view as a proof of $\phi(u)$.
But setting $\pi(E u v) = \pi(E v v) = 1$ results in $\pi \ext{\phi(u)} = \infty$,
since the iteration for the evaluation of $\phi$ at $v$ is $\infty$, $1 \cdot \infty$, $1 \cdot \infty$, $\dots$
which stagnates immediately.
Although $\N^\infty$ is truth-preserving,
the example hints at another general issue:
Multiplication with non-zero values in $\N^\infty$
always increases values.
The same is true for addition, so fixed-point iterations of $\gfp$-formula
are likely to result in $\infty$ and \emph{do not give meaningful provenance information}, e.g. about the number of proofs.
Since the computation in $\N^\infty \ps X$ yields 0, we further see that
we cannot obtain the result in $\N^\infty$ from the computation in $\N^\infty \ps X$ by polynomial evaluation.
Hence evaluation of formal power series
\emph{does not preserve provenance semantics} in general.
This is a further reason why formal power series are not the right provenance semirings for LFP.
\eex

A further motivation for absorptive semirings is
that they give information about reduced proofs of a formula.
The property $a + ab = a$ implies, for example,
that a proof containing two literals mapped to
$a$ and $b$, thus having the value $ab$, is absorbed by a proof
only using one literal, with provenance value $a$.
To see why this is useful when working with greatest fixed-points,
we consider an example in the Why-semiring $\W[X]$.
This semiring results from polynomials $\N[X]$
by dropping both coefficients and exponents,
which makes it finite and thus truth-preserving, but not absorptive.
This is similar to $\N^\infty$ and although $\W[X]$ provides more information
about greatest fixed-points, the lack of absorptivity also here leads
to undesired provenance information.
Another benefit of absorption is that, unlike formal power series $\Ninf \ps X$,
provenance information is always finitely representable (see Sect.~\ref{Sect:Sinf}).

\ex
\label{ex:infpathWhy}
Recall the formula from the previous example,
now interpreted on a different graph:
\begin{center}
$\displaystyle \phi(u) = \gfpfml R x {\E y (E x y \land R y)} u$
\hspace{2cm}
\begin{tikzpicture}[baseline, yshift=2.5pt]
\node [dot, label={above:$u$}] (u) {};
\node [dot, label={above:$v$}, right of=u] (v) {};
\path [draw, ->, >=stealth', shorten <=2pt, shorten >=2pt]
(u) edge node [below, font=\scriptsize] {$y$} (v)
(u) edge [in=-150, out=145, looseness=15] node [left, font=\scriptsize] {$x$} (u);
\end{tikzpicture}
\end{center}

We consider the $\W[X]$-interpretation $\pi$ with
$\pi(Euu) = x$ and $\pi(Euv)=y$ that defines the above graph (with $X = \{x,y\}$).
Here there is only one infinite path which uses the edge labelled $x$
infinitely often. As $\W[X]$ is obtained by dropping exponents,
it does not allow to count the usage of $x$, so the path simply corresponds to the monomial $x$.

However, the iteration $\top$, $x\top$, $x^2\top = x\top$ at node $u$ leads to
$\pi \ext {\phi(u)} = x\top = x + xy$, which additionally contains the monomial $xy$.
As there is no infinite path using both edges, $xy$ does not
correspond to an evaluation strategy of $\phi(u)$ on the given graph.
The problem here is that $\top \neq 1$ (recall that $\top=1$ is equivalent to absorption).
And indeed, absorption would imply $x + xy = x$ as expected.
Making $\W[X]$ absorptive results in the semiring $\PosBool(X)$
which provides useful provenance information,
but is strictly less informative than generalized absorptive polynomials.
\eex

\medskip
We next address the issue of truth-preservation, as seen in \cref{ex:infpathViterbi}.
Formally, we define it as follows.
As in  \cite{GraedelTan17},
we say that a $K$-interpretation $\pi\colon \Lit_A(\tau)\rightarrow K$
is \emph{model-defining} if for all atoms $R\tup{a}$ exactly one of the two values
$\pi\ext{R\tup{a}}$ and $\pi\ext{\neg R\tup{a}}$ is zero.
A model-defining $K$-interpretation induces a unique structure 
$\mathfrak{A}_\pi$
with universe $A$ and $\tup{a}\in R^\mathfrak{A}$ if, and only if, $\pi(R\tup{a})\neq 0$.
For a truthful provenance analysis for a logic $L$, this
should lift from literals to arbitrary sentences $\phi\in L$.
If this is guaranteed, then $K$ is truth-preserving.
%If this is always the case, then the semiring $K$ is truth-preserving.

\bdefn
A semiring $K$ is \emph{truth-preserving} for a logic $L$,
if $\mathfrak{A}_\pi \models \phi$ if, and only if, $\pi \ext \phi \neq 0$,
for all model-defining $K$-interpretations $\pi$ and all LFP-sentences $\phi$.
\edefn

We shall define and investigate in the next section the semiring of
generalized absorptive polynomials $\Sinf[X]$ which, contrary to other fully continuous
and absorptive semirings, is truth-preserving due to the following
algebraic property.

\bdefn
A fully chain-complete semiring $K$ is \emph{chain-positive}
if for each non-empty chain $C \subseteq K$ of non-zero elements,
the infimum $\Inf C$ is non-zero as well.
\edefn

Chain-positivitiy is sufficient to guarantee that a positive, fully chain-complete semiring is truth-preserving.
This can be seen by a straight-forward induction on fixed-point iterations.
A more elegant proof makes use of the fundamental property together with the observation
that a positive, fully chain-complete semiring $K$ is chain-positive if, and only if,
the unique function $h : K \to \B$ with $h^{-1}(0) = \{0\}$ is a fully continuous semiring homomorphism
(this is easy to see by case distinction on the values in $\B$).

\begin{proposition}
\label{propChainpositiveTruthpreserving}
Every chain-positive, positive semiring is truth-preserving for $\LFP$.
\end{proposition}

\begin{proof}
  Let $K$ be such a semiring and consider the function $h : K \to \B$ with $h^{-1}(0) = \{0\}$.
  Then $h$ is a fully continuous homomorphism and we can apply the fundamental property:
	Let $\pi$ be a model-defining $K$-interpretation with induced model $\AA_\pi$,
  and let $\phi$ be an LFP-sentence.
	Notice that $h \circ \pi$ is a $\B$-interpretation that induces the same model as $\pi$.
	Provenance semantics in $\B$ coincides with standard semantics,
	hence $\AA_\pi \models \phi \Iff (h \circ \pi) \ext {\phi} = 1$.
	By the fundamental property,
	$h(\pi \ext \phi) = (h \circ \pi) \ext {\phi} = 1$,
	and this is equivalent to $\pi \ext {\phi} \neq 0$ by definition of $h$.
\end{proof}

Chain-positivity is not an indispensible requirement for provenance analysis,
as shown by the Viterbi semiring (which is absorptive and fully continuous). However, we need this property for provenance semirings
which should give insights into proofs or evaluation strategies and thus have to preserve truth.

\section{Generalized Absorptive Polynomials}\label{Sect:Sinf}

We now discuss the semirings $\Sinf[X]$ and $\Sinf[X,\nnX]$ of generalized absorptive polynomials.
They were introduced in \cite{GraedelTan20} and generalize the
semiring of absorptive polynomials $\mathsf{Sorp}(X)$ from \cite{DeutchMilRoyTan14}
by admitting exponents in $\N^\infty$ to guarantee chain-positivity.
We show that these semirings are, in a well-defined sense,  
the most general absorptive, fully continuous semirings
and we argue that $\Sinf[X, \nnX]$ is the right provenance semiring for LFP.

\bdefn Let $X$ be a \emph{finite} set of provenance tokens. We generalize the notion
of a monomial over $X$  to admit exponents from $\Ninf$. Monomials are here 
functions $m:X\rightarrow\Ninf$, written $x_1^{m(x_1)}\cdots x_n^{m(x_n)}$. 
Multiplication adds the exponents, and $x^\infty\cdot x^n =x^\infty$.
We say that $m_2$ \emph{absorbs} $m_1$, denoted $m_2 \succeq m_1$, if $m_2$ has smaller exponents
than $m_1$, i.e., $m_2(x)\leq m_1(x)$ for all $x\in X$.
This is the pointwise partial order given by the reverse order on $\Ninf$.
\edefn

The set of monomials inherits a lattice structure from $\Ninf$
and is, of course, infinite.
However, it has some crucial finiteness properties.

\prop
\label{thm:chainsFinite}
Every antichain of monomials is finite. Further, while there are infinitely descending chains
of monomials, such as $1=x^0\succ x^1\succ x^2\succ \dots$ there are
no infinitely ascending such chains.
\eprop

Indeed, $(\Ninf,\leq)$ is a well-order. The set of monomials $m:X\rightarrow \Ninf$ with the \emph{reverse order}
of the absorption order is isomorphic to $(\Ninf)^k$ with $k=|X|$ and
with the component-wise order inherited from $(\Ninf,\leq)$. This is a well-quasi-order
and therefore has no infinite descending chains and no infinite antichains.
This implies that in the set of monomials over $X$ with the absorption order,
all ascending chains and all antichains are finite.

\bdefn
We define $\Sinf[X]$ as the set of antichains of monomials
with indeterminates from $X$ and exponents in $\Ninf$.
We write such antichains as formal sums of their monomials 
and call them \emph{generalized absorptive polynomials}. 
Addition and multiplication of polynomials proceed as usual,
but keeping only the maximal monomials (w.r.t.\ $\succeq$) in the result (and disregarding coefficients).
%followed by replacing nonzero coefficients with 1 and 
%keeping only the maximal monomials (w.r.t.\ $\succeq$) in the result.
%The empty antichain corresponds to the 0 polynomial. The 1 polynomial
%consists of just the monomial in which every indeterminate has exponent 0.
\edefn

Since antichains of monomials are finite,
there is no difference between polynomials and power series here
and moreover, $\Sinf[X]$ is countable.
The natural order on $\Sinf[X]$ can be characterized by monomial absorption:
$P \le Q$ if, and only if,
for each $m \in P$ there is $m' \in Q$ with $m' \succeq m$.
With \cref{thm:chainsFinite},
it follows that there are no infinitely ascending chains of polynomials,
and further that the supremum of $S \subseteq \Sinf[X]$
is $\Sup S = \mxls(\bigcup S)$
which is the set of $\succeq$-maximal monomials in $\bigcup S$ (see below for the proof).
Due to the exponent $\infty$ and the finiteness of $X$,
there is a smallest monomial $m_\infty\neq 0$ with $m_\infty(x) = \infty$ for all $x \in X$.
This ensures chain-positivity of $\Sinf[X]$.
%as the infimum of each positive chain must at least contain $m_\infty$.

In order to provide proofs of the algebraic properties of $\Sinf[X]$,
we begin with simple observations that hold in all absorptive, fully continuous semirings.
In these semirings, powers of an element $a$ always form a descending $\omega$-chain $a \ge a^2 \ge a^3 \ge \dots$
and we denote its infimum by $a^\infty$, which we call the \emph{infinitary power} of $a$.

\lem[Splitting Lemma]
\label{lemSplitting}
Let $K$ be a fully continuous semiring
and let $(a_i)_{i < \omega}$ and $(b_i)_{i < \omega}$ be two descending $\omega$-chains.
Then, $\Inf_{i < \omega} (a_i \circ b_i) = \big(\Inf_{i < \omega} a_i\big) \circ \big(\Inf_{j < \omega} b_j\big)$,
with $\circ \in \{ +, \cdot \}$. Analogous statements hold for suprema.
\elem

\begin{proof}
	We only show the statement for infima, the proof for suprema is analogous.
	We have the following equality, where $(*)$ holds since $K$ is fully continuous:
	\[
	\Inf_{i < \omega} a_i \circ b_i \overset{(1)}{=}
	\Inf_{i < \omega} \Inf_{j < \omega} a_i \circ b_j \overset{(*)}{=}
	\Inf_{i < \omega} (a_i \circ \Inf_{j < \omega} b_j) \overset{(*)}{=}
	\Inf_{i < \omega} a_i \;\circ\; \Inf_{j < \omega} b_j
	\]
	We prove both directions of $(1)$.
	Fix $i,j$ and let $k = \max(i,j)$.
	Then $a_i \circ b_j \ge a_k \circ b_k \ge \Inf_k a_k \circ b_k$ by monotonicity of $\circ$.
	As $i,j$ are arbitrary, this proves $\Inf_i \Inf_j a_i \circ b_j \ge \Inf_k a_k \circ b_k$.
	
	For the other direction, we have $a_i \circ b_i \ge a_i \circ \Inf_j b_j$ for every $i$ by monotonicity of $\circ$. By continuity, $a_i \circ b_i \ge \Inf_j a_i \circ b_j$ for every $i$, and
	thus $\Inf_i a_i \circ b_i \ge \Inf_i \Inf_j a_i \circ b_j$.
\end{proof}

\lem[Infinitary Power]
\label{lemInfpow}
Let $K$ be an absorptive, fully continuous semiring.
Then,
\begin{enum}
	\item \label{lemInfpow-addmult}
	$(a+b)^\infty = a^\infty + b^\infty$
	and
	$a^n \cdot a^\infty = a^\infty$,
	for $a,b \in K$ and $n \in \N^\infty$,
	
	\item \label{lemInfpow-inf}
	$\big( \Inf_{i < \omega} a_i \big)^\infty = \Inf a_i^\infty$
	for any descending $\omega$-chain $(a_i)_{i < \omega}$ in $K$.
\end{enum}
In $\Sinf[X]$, we further have an analogue of property \eqref{lemInfpow-inf} for infima:
\begin{enum}
	\setcounter{enum}{2}
	\item \label{lemInfpow-sup}
	$\big( \Sup S )^\infty = \Sup S^\infty$,
	where we write $S^\infty = \{ P^\infty \mid P \in S \}$,
	for any set $S \subseteq \Sinf[X]$.
\end{enum}
\elem

\begin{proof}
For the first statement in (1), let $a,b \in K$.
We clearly have $(a+b)^n \ge a^n + b^n$ (for all $n < \omega$) and hence $(a+b)^\infty \ge a^\infty + b^\infty$.
For the other direction, fix $n$ and consider $(a+b)^{2n} = \sum_{i=0}^{2n} {\binom{2n}{i}} a^{2n-i} b^i$.
Each summand is absorbed by either $a^n$ (if $i \le n$) or by $b^n$ (if $i \ge n$), hence
$a^n + b^n \ge (a+b)^{2n} \ge (a+b)^\infty$ and the claim follows.
The second statement follows by continuity of multiplication:
$a^n \cdot a^\infty = a^n \cdot \Inf_{k < \omega} a^k = \Inf_{k < \omega} a^{k+n} = a^\infty$.

For (2), we use the splitting lemma (in $(*)$) and the fact that we can swap infima:
\begin{align*}
\Inf_{i < \omega} a_i^\infty =
\Inf_{i < \omega} \Inf_{n < \omega} a_i^n =
\Inf_{n < \omega} \Inf_{i < \omega} a_i^n \overset{(*)}{=}
\Inf_{n < \omega} \Big(\Inf_{i < \omega} a_i\Big)^n =
\Big(\Inf_{i < \omega} a_i\Big)^\infty
\end{align*}

For the last statement, we first note that for $a,b \in K$ with $a \le b$,
we always have $a^\infty \le b^\infty$. That is, the infinitary power is monotone.
This follows directly from the definition, as $a \le b$ implies $a^n \le b^n$
and thus $\Inf_{n < \omega} a^n \le \Inf_{n < \omega} b^n$.

For statement (3) in $\Sinf[X]$, we compare the two sides of the equation.
The direction $(\Sup S)^\infty \ge \Sup S^\infty$ follows from the aforementioned monotonicity.
For the other direction, let $\Sup S = m_1 + \dots + m_k$ for a finite number of monomials $m_1, \dots, m_k$.
By statement (2), $(\Sup S)^\infty = m_1^\infty + \dots + m_k^\infty$.
Fix one monomial $m_i$. As $\Sup S = \mxls(\bigcup S)$, there is a $P \in S$ with $m_i \in P$.
Hence $m_i \le P$ and thus $m_i^\infty \le P^\infty \le \Sup S^\infty$ by monotonicity.
As this holds for each $1 \le i \le k$, we can conclude $m_1^\infty + \dots + m_k^\infty \le \Sup S^\infty$.
\end{proof}

\lem[Countable Chains]
\label{lemCountableChains}
Let $K$, $K'$ be fully chain-complete semirings and $C \subseteq K$ a countable chain.
Then there is a descending $\omega$-chain $(x_i)_{i < \omega}$ such that $\Inf C = \Inf_{i} x_i$.
Moreover, if $f : K \to K'$ is a monotone function, then additionally
$\Inf f(C) = \Inf_{i} f(x_i)$.
Analogue statements hold for suprema.
\elem

\begin{proof}
We only show the statement involving $f$, as it implies the first,
and only consider infinite $C$ (otherwise the statement is trivial).
Fix a bijection $g : \omega \to C$ and recursively define
$x_0 = g(0)$ and $x_{i+1} = \min(g(i+1), x_i)$.
This defines an $\omega$-chain with $x_i \in C$ and thus $\Inf_i f(x_i) \ge \Inf f(C)$.
Conversely, for every $c \in C$ there is an $i$ with $g(i) = c$ and thus $c \ge x_i$.
By monotonicity, $f(c) \ge f(x_i)$ and thus $\Inf f(C) \ge \Inf_i f(x_i)$.
\end{proof}

\begin{proposition}
\label{propSinfSemiring}
$(\Sinf[X],+,\cdot,0,1)$ is absorptive, fully continuous, and chain-positive.
\end{proposition}

\begin{proof}
Absorption is clear from the definition.
We first prove that the natural order on $\Sinf[X]$ forms a complete lattice, implying chain-completeness.
For $S \subseteq \Sinf[X]$,
\[
  \Sup S = \mxls(\bigcup S)
\]
where $\bigcup S$ are all monomials occurring in some polynomial of $S$
and $\mxls(M)$ denotes the set of maximal monomials (w.r.t.\ $\preceq$) in the set $M$.
For each $P \in S$, we have $P \subseteq \bigcup S$ and hence $P \le \mxls(\bigcup S)$,
so $\mxls(\bigcup S)$ is an upper bound for $S$.
To see that it is the least upper bound, let $Q$ be any upper bound for $S$, so $Q \ge P$ for all $P \in S$.
For each $m \in \mxls(\bigcup S)$ there is a $P \in S$ with $m \in P$ and hence $m \le P \le Q$.
It follows that $\mxls(\bigcup S) \le Q$.

For chain-positivity, consider the monomial $m_\infty$ with $m_\infty(x) = \infty$ for all $x \in X$.
Then $m_\infty$ is the smallest monomial with respect to $\preceq$.
Given a descending $\omega$-chain $(P_i)_{i < \omega}$ in $\Sinf[X]$ with $P_i > 0$ for all $i$,
we know that each $P_i$ must contain some monomial.
These monomials must be at least as large as $m_\infty$.
Hence $P_i \ge m_\infty$ for all $i$ and thus $\Inf_{i < \omega} P_i \ge m_\infty > 0$.

What remains is to show that $\Sinf[X]$ is fully continuous.
To this end, we have to prove that the two semiring operations preserve both suprema and infima of nonempty chains.
In the following, let $C \subseteq \Sinf[X]$ be such a chain
and let $p \in \Sinf[X]$ be a polynomial.
\begin{itemize}
	\item
	We first consider addition.
	Due to idempotency of $\Sinf[X]$, addition corresponds to the supremum
	and we have
	$p + \Sup C = \Sup \{ p, \Sup C \} = \Sup \{ \Sup \{ p, c \} \mid c \in C \} = \Sup (p + C)$.
	
	\item
	For infima, we show $\Inf (p+C) \le p + \Inf C$.
	The other direction follows from monotonicity of addition.
	Let $m \in \Inf (p+C)$ be a monomial.
	Then $m \le \Inf (p+C)$ and thus $m \le p+c$ for all $c \in C$.
	So $m$ is absorbed by a monomial in $p+c$ which originates
	either from $p$ or from $c$.
	If $m \le p$, then also $m \le p + \Inf C$ and we are done.
	Otherwise, we have $m \le c$ for all $c \in C$
	and hence $m \le \Inf C \le p + \Inf C$.
	It follows that $\Inf (p+C) \le p + \Inf C$.
	
	\item
	We now turn to the continuity of multiplication.
	We first show that $p \cdot \Sup C \le \Sup (p \cdot C)$.
	The other direction holds by monotonicity of multiplication
	(which follows from distributivity).
	Ascending chains are finite, so there is a $c \in C$ with $\Sup C = c$.
	Then $p \cdot \Sup C = p \cdot c \le \Sup (p \cdot C)$.
	
	\item
	It remains to show that $\Inf (p \cdot C) \le p \cdot \Inf C$
	(again, the other direction follows from monotonicity).
	We first consider the case where $p$ consists of a single monomial $m$.
	Let $q$ be a monomial of $\Inf (m \cdot C)$.
	Due to absorption, we have $q \le m \cdot c \le m$ (for any $c \in C$).
	Hence $q(x) \ge m(x)$ for all $x \in X$ and we can thus write $q$
	as $q = m \cdot q'$ with $q'(x) = q(x) - m(x)$
	(where we set $\infty - n = \infty$ for all $n \in \N^\infty$).
	We claim that $q' \le \Inf C$.
	To see this, let $c \in C$.
	Then $q \le m \cdot c$ and thus $m \cdot q' \le m \cdot c$.
	By comparing the exponents, we see that $q' \le c$ and the claim follows.
	Hence $q = m \cdot q' \le m \cdot \Inf C$.
	As this argument applies to all monomials of $\Inf (m \cdot C)$,
	we have shown $\Inf (m \cdot C) \le m \cdot \Inf C$.
	For the case where $p$ consists of several monomials, so $p = m_1 + \dots + m_k$,
  we exploit the continuity of addition and apply the Splitting Lemma
  (together with \cref{lemCountableChains}):  
	$\Inf (p \cdot C) = \Inf_{c \in C} (m_1 c + \dots + m_k c) = %\overset{\eqref{lemSplitting}}{=}
	(\Inf m_1 C) + \dots + (\Inf m_k C) \le
	(m_1 \cdot \Inf C) + \dots + (m_k \cdot \Inf C) =
	p \cdot \Inf C$. \qedhere
\end{itemize}
\end{proof}

The central property of $\Sinf[X]$
is the following universal property which says
that it is the absorptive fully continuous semiring freely generated by $X$ 
for fully continuous homomorphisms.
These homomorphisms enable us to apply the fundamental property.
The main difficulty in the proof of this statement
is the continuity requirement on infima of chains,
for which we make use of K\H onig's lemma.
%Together with the fundamental property,
%this makes $\Sinf[X]$ the most general fully continuous provenance semiring for LFP under the assumption of absorption.

\begin{theorem}[Universality]
\label{universality-of-Sinf}
Every mapping $h : X \to K$ into an absorptive, fully continuous semiring $K$
uniquely extends to a fully continuous semiring homomorphism $h : \Sinf[X] \to K$.
\end{theorem}

\begin{proof}
	Due to the additivity and multiplicity requirements for homomorphisms,
	$h$ uniquely extends to monomials. For the exponent $\infty$,
	notice that continuity requires
	$h(x^\infty) = \Inf_{n < \omega} h(x)^n$ for $x \in X$.
	It further follows that $h(m_1 + m_2) = h(m_1) + h(m_2)$,
	hence $h$ is uniquely defined on $\Sinf[X]$.
	Care has to be taken regarding absorption.
	If $m_1 \preceq m_2$,	then $m_1 + m_2 = m_2$.
	Since $h$ preserves the order and $K$ is absorptive,
	we also have $h(m_1 + m_2) = h(m_1) + h(m_2) = h(m_2)$.
	It follows by induction that $h$ is well-defined.
  
	It remains to show that $h$ is fully continuous.
	Ascending chains are always finite, so we only have to consider descending chains.
  By \cref{lemCountableChains}, it further suffices to consider $\omega$-chains.
	The only remaining observation is that
	\[
	\Inf_{i < \omega} h(P_i) = h\big(\Inf_{i < \omega} P_i\big)
	\]
	for any descending $\omega$-chain $(P_i)_{i < \omega}$ in $\Sinf[X]$.
	The homomorphism $h$ preserves addition and is thus monotone,
	which entails the direction ``$\ge$''.
	
	For the other direction, we first consider the case of single monomials.
	Let $(m_i)_{i < \omega}$ be a descending $\omega$-chain of monomials.
	Recall that $X$ is finite, so we can write $m_i = \prod_{x \in X} x^{m_i(x)}$.
	As the $m_i$ form a descending chain,
	the exponents $(m_i(x))_{i < \omega}$ form an ascending chain for each $x \in X$.
	By \cref{lemSplitting} and the definition of $h$,
	\[
	\Inf_{i < \omega} h(m_i) =
	\prod_{x \in X} \Inf_{i < \omega} h(x)^{m_i(x)} \overset{(*)}{=}
	\prod_{x \in X} h(x)^{\Sup_{i} m_i(x)} =
	% 	\prod_{x \in X} h(x)^{\big(\Inf_{i} m_i\big)(x)} =
	h(\Inf_{i < \omega} m_i).
	\]
	where $(*)$ can easily be seen by case distinction
	whether $\Sup_{i < \omega} m_i(x)$ is finite or $\infty$.

	For the general case of polynomials, let $P_\omega = \Inf_{i < \omega} P_i$
	be the infimum, which is of the form $P_\omega = m_1 + \dots + m_n$.
	We define a second, canonical $\omega$-chain $(P^*_i)_{i < \omega}$ with the same infimum.
	To this end, we define the canonical monomial chain $(m^*_j)_{j < \omega}$ of a given monomial $m$ as follows
  (see \cref{fig:sorbCanonicalProof} for an example),
	\[
	m^*_j(x) = \min(j, m(x)), \quad \text{ for all } x \in X,
	\]
	which satisfies the following properties needed for the proof:
	
	\begin{enumerate}
		\item
		If $m$, $v$ are two monomials with $m \preceq v$, then $m^*_j \preceq v^*_j$ for all $j < \omega$.
		\label{universalityMonomialOrder}
		
		\item
		If $m = \Inf_{i < \omega} m_i$ for an $\omega$-chain $(m_i)_{i < \omega}$ of monomials,
		then $\A j \, \E i: m^*_j \succeq m_i$.
		\label{universalityMonomialInfimum}
		
		\item
		In particular, $\Inf_{j < \omega} m^*_j = m$.
		\label{universalityMonomialInfimumTrivial}
	\end{enumerate}
	
	The canonical polynomial chain $(P^*_j)_{j < \omega}$ is then defined by
	$P^*_j = (m_1)^*_j + \dots + (m_n)^*_j$ for each $j < \omega$.
	We make the following observation:
	\[
	\textbf{Claim: } \A j \, \E i: P^*_j \ge P_i.
	\]
	We first show that the claim implies the theorem:
	\begin{align*}
	\Inf_{i < \omega} h(P_i) \overset{(1)}{\le} \Inf_{j < \omega} h(P^*_j)
	&= \Inf_{j < \omega} \Big( h((m_1)^*_j) + \dots + h((m_n)^*_j) \Big) \\
	&\overset{(2)}{=} \Inf_{j < \omega} h((m_1)^*_j) + \dots + \Inf_{j < \omega} h((m_n)^*_j) \\
	&\overset{(3)}{=} h\Big(\Inf_{j < \omega} (m_1)^*_j\Big) + \dots + h\Big(\Inf_{j < \omega} (m_n)^*_j \Big) \\
	&\overset{(4)}{=} h(m_1) + \dots + h(m_n) = h(P_\omega),
	\end{align*}
	where (1) follows from the claim,
	(2) holds by \cref{lemSplitting},
	(3) was shown above
	and (4) holds due to property \ref{universalityMonomialInfimumTrivial} above.
	Hence the claim suffices to prove the theorem.
	
	To prove the claim, assume towards a contradiction that there is a $j$
	such that $P^*_j \ngeq P_i$ for all $i < \omega$.
	Let us fix an $i < \omega$ for the moment.
	Because of $P^*_j \ngeq P_i$, there is a monomial $m_i \in P_i$ with $P^*_j \ngeq m_i$.
	Because of $P_{i-1} \ge P_i$, there is further $m_{i-1} \in P_{i-1}$ with $m_{i-1} \succeq m_i$.
	But then also $P^*_j \ngeq m_{i-1}$ (as otherwise $P^*_j \ge m_{i-1} \ge m_i$).
	By repeating this argument,
	we obtain a finite chain $m_0 \succeq m_1 \succeq \dots \succeq m_i$ of monomials
	with the property that $m_k \in P_k$ and $P^*_j \ngeq m_k$ for all $0 \le k \le i$.
	
	This argument applies to all $i < \omega$,
	so we obtain arbitrarily long finite chains with this property.
	By K{\H o}nig's lemma (recall that all polynomials $P_i$ are finite),
	there must be an infinite monomial chain $(m_i)_{i < \omega}$ with $m_i \in P_i$
	and $P^*_j \ngeq m_i$ for all $i < \omega$.
	Let $m_\omega = \Inf_{i < \omega} m_i$.
	Because of $m_i \le P_i$ for all $i$, we have $m_\omega \le P_\omega$,
	so there is a monomial $v \in P_\omega$ with $m_\omega \preceq v$.
	By considering the corresponding canonical monomial chains
	$(v^*_k)_{k < \omega}$ and $((m_\omega)^*_k)_{k < \omega}$
	at $k = j$, we obtain a contradiction:
	We know from the above properties that there is an $i$ with $(m_\omega)^*_j \succeq m_i$
	and further $v^*_j \ge (m_\omega)^*_j$.
	Because of $v^*_j \in P^*_j$, we obtain $P^*_j \ge v^*_j \ge (m_\omega)^*_j \ge m_i$, contradicting our assumption.
	The claim follows, closing the overall proof.
\end{proof}

\begin{figure}%[!b]
	\tikzset{
		mm/.style={anchor=base}, shorten/.style={shorten >=2pt, shorten <=2pt},
		arr/.style={>=stealth'},
	}
	\newcommand{\poly}[5]{
		\node [mm,align=right] at (-1.2,-1.2*#1) (P#1) {$#2:$};
		\node [mm] at (0,-1.2*#1) (a#1) {$#3$};
		
		\ifx#4\empty
		\else
		\node [mm] at (1.5,-1.2*#1) (b#1) {$#4$};
		\node [mm,black] at ($(a#1.base) !0.5! (b#1.base)$) {$+$};
		\fi
		
		\ifx#5\empty
		\else
		\node [mm] at (3,-1.2*#1) (c#1) {$#5$};
		\node [mm,black] at ($(b#1.base) !0.5! (c#1.base)$) {$+$};
		\fi
	}
	\centering
	\begin{tikzpicture}[font=\small]
		\poly 0 {P_0} {x} {y} {};
		\poly 1 {P_1} {x^\infty} {y} {};
		\poly 2 {P_2} {x^\infty} {y^2z} {xy};
		\poly 3 {P_3} {x^\infty} {y^2z} {};
		\poly 4 {P_4} {x^\infty} {y^3z^2} {xy^2z};
		\begin{scope}[yshift=1cm]
			\poly 6 {P_\omega} {x^\infty} {y^\infty z^2} {};
		\end{scope}
		
    \path (P0.west) to node [sloped, anchor=south, yshift=3pt,font=\scriptsize] {$\ge$} (P1.west);
    \path (P1.west) to node [sloped, anchor=south, yshift=3pt,font=\scriptsize] {$\ge$} (P2.west);
    \path (P2.west) to node [sloped, anchor=south, yshift=3pt,font=\scriptsize] {$\ge$} (P3.west);
    \path (P3.west) to node [sloped, anchor=south, yshift=3pt,font=\scriptsize] {$\ge$} (P4.west);
    
		% hack for top margin
		\node [draw=none] at (0,.4cm) {};
		
		\draw [gray, thick, ->, shorten >=2pt, shorten <=2pt, line cap=round]
		%(b0) edge [-, line width=3pt, gray!20] (b1)
		(a0) edge (a1) (a1) edge (a2) (a2) edge (a3) (a3) edge (a4)
		(b0) edge (b1) (b1) edge (b2) (b2) edge (b3) (b3) edge (b4)
		(b1) edge (c2) (b3) edge (c4)
		(a4) edge [dotted,-] (a6)
		(b4) edge [dotted,-] (b6)
		;
		
	\end{tikzpicture}
	\qquad \quad \quad
	\vrule
	\qquad \quad \quad
	\begin{tikzpicture}[font=\small]
		\poly 0 {P^*_0} {1} {} {};
		\poly 1 {P^*_1} {x} {yz} {};
		\poly 2 {P^*_2} {x^2} {y^2z^2} {};
		\poly 3 {P^*_3} {x^3} {y^3z^2} {};
		\poly 4 {P^*_4} {x^4} {y^4z^2} {};
		\begin{scope}[yshift=1cm]
			\poly 6 {\smash{P^*_\omega}\vphantom{P_\omega}} {x^\infty} {y^\infty z^2} {};
		\end{scope}
		
    \path (P0.west) to node [sloped, anchor=south, yshift=3pt,font=\scriptsize] {$\ge$} (P1.west);
    \path (P1.west) to node [sloped, anchor=south, yshift=3pt,font=\scriptsize] {$\ge$} (P2.west);
    \path (P2.west) to node [sloped, anchor=south, yshift=3pt,font=\scriptsize] {$\ge$} (P3.west);
    \path (P3.west) to node [sloped, anchor=south, yshift=3pt,font=\scriptsize] {$\ge$} (P4.west);
    
		\draw [gray, thick, ->, shorten >=2pt, shorten <=2pt, line cap=round]
		(a0) edge (a1) (a1) edge (a2) (a2) edge (a3) (a3) edge (a4)
		(a0) edge (b1)
		(b1) edge (b2) (b2) edge (b3) (b3) edge (b4)
		%(b0) edge (c1) (b1) edge (c2) (b2) edge (c3) (b3) edge (c4)
		%(c1) edge (c2) (c3) edge (c4)
		(a4) edge [dotted,-] (a6)
		(b4) edge [dotted,-] (b6)
		;
		
	\end{tikzpicture}
	\qquad
	
	\caption{An example of a polynomial $\omega$-chain (left) and the corresponding canonical chain (right)
		for the proof of \cref{universality-of-Sinf}. The arrows indicate absorption between
		monomials of consecutive polynomials and induce a directed graph which justifies
		our application of K\H onig's lemma.}
	\label{fig:sorbCanonicalProof}
\end{figure}
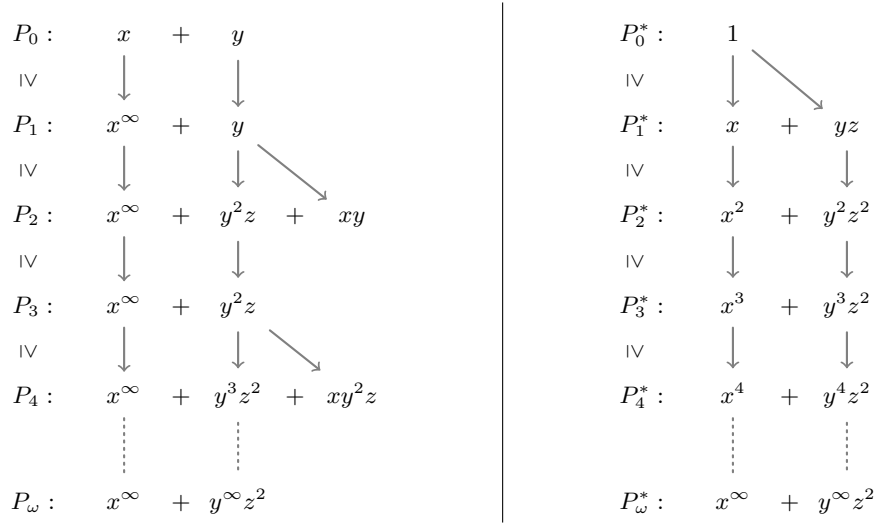

The idea to apply K\H onig's lemma to monomial chains
can also be applied to infima of chains in general
and is useful for some of the later proofs.

\begin{proposition}[Characterization of Infima]
	\label{propCharacterizationInfima}
	Let $(P_i)_{i < \omega}$ be a descending $\omega$-chain in $\Sinf[X]$.
	Let further $\MM$ be the set of descending $\omega$-chains $(m_i)_{i < \omega}$ of monomials
	with the property that $m_i \in P_i$ for all $i$. Then,
	\[ \Inf_{i < \omega} P_i = \Sup \big\{ \Inf_{i < \omega} m_i \mid (m_i)_{i < \omega} \in \MM \big\}. \]
\end{proposition}

\begin{proof}
	By definition, $m_i \le P_i$ and thus $\Inf_i m_i \le \Inf_i P_i$ for every chain $(m_i)_{i < \omega} \in \MM$.
	Hence direction ``$\ge$'' of the proposition follows.
	
	For the other direction, consider the infimum $P_\omega = \Inf_i P_i$.
	We claim that for every monomial $m_\omega \in P_\omega$, there is a monomial chain $(m_i)_{i < \omega} \in \MM$ with $\Inf_i m_i \succeq m_\omega$.
	This is sufficient to close the proof.
	
	To prove the claim, we use a similar argument as in the proof of the universal property of $\Sinf[X]$.
	Fix a monomial $m_\omega \in P_\omega$ and, for the moment, an $i < \omega$.
	We have $P_\omega \le P_i$, so there is a monomial $m_i \in P_i$ with $m_\omega \preceq m_i$.
	As $P_i \le P_{i-1}$, there must further be a monomial $m_{i-1} \in P_{i-1}$ with $m_i \preceq m_{i-1}$.
	Iterating this argument yields a sequence $m_i \preceq m_{i-1} \preceq m_{i-2} \preceq \dots \preceq m_0$ of monomials with $m_\omega \preceq m_j$ and $m_j \in P_j$ (for all $j \le i$).
	This construction is possible for each $i$, so by K\H onig's lemma (recall that all polynomials are finite),
	there must be an infinite monomial chain $(m_i)_{i < \omega}$ with $m_i \in P_i$ and $m_i \succeq m_\omega$ for each $i$.
	Hence $(m_i)_{i < \omega} \in \MM$ and the infimum is $\Inf_{i < \omega} m_i \succeq m_\omega$ as claimed.
\end{proof}

The fact that the universal property guarantees fully continuous homomorphisms should not be taken lightly:
We have seen in Example \ref{ex:infpathViterbi} that this is not the case for formal power series.
There, polynomial evaluation induces homomorphisms that are, in general, not fully continuous
and hence do not preserve greatest fixed points.
The following example shows how we can specialize provenance values in $\Sinf[X]$ to application semirings.

\ex
\label{ex:sinfUniversality}
We recall the setting from Example~\ref{ex:infpathViterbi}
and first consider the model-defining $\Sinf[X]$-interpretation tracking
the two edges labelled $x$ and $y$, as indicated in the left graph.
\begin{center}
$\displaystyle \phi(u) = \gfpfml R x {\E y (E x y \land R y)} u$
\hspace{1cm}
\begin{tikzpicture}[baseline, yshift=2.5pt]
\node [dot, label={below:$u$}] (u) {};
\node [dot, label={below:$v$}, right of=u] (v) {};
\path [draw, ->, >=stealth', shorten <=2pt, shorten >=2pt]
(u) edge node [above] {$x$} (v)
(v) edge [in=-30, out=35, looseness=15] node [right] {$y$} (v);
\end{tikzpicture}
\hspace{1cm}
\begin{tikzpicture}[baseline, yshift=2.5pt]
\node [dot, label={below:$u$}] (u) {};
\node [dot, label={below:$v$}, right of=u] (v) {};
\path [draw, ->, >=stealth', shorten <=2pt, shorten >=2pt]
(u) edge node [above] {$x$} (v)
(v) edge [in=-30, out=35, looseness=15] node [right] {$y$} (v)
(u) edge [in=-150, out=145, looseness=15] node [left] {$z$} (u);
\end{tikzpicture}
\end{center}
We obtain $\pi \ext {\phi(u)} = x y^\infty$
corresponding to the infinite path $uvvv\dots$.
The confidence values from Example~\ref{ex:infpathViterbi}
can be obtained by polynomial evaluation:
For $h(x) = h(y) = 1$, we get
$(h \circ \pi) \ext {\phi(u)} = 1 \cdot 1^\infty = 1$
and for $h'(x) = 1$, $h'(y) = 1-\epsilon$ we get
$(h' \circ \pi) \ext {\phi(u)} = 1 \cdot (1-\epsilon)^\infty = 0$.

Let us next consider the graph on the right by setting $\pi(E u u) = z$.
There are now infinitely many infinite paths from $u$ to $v$.
However, we obtain only finitely many monomials due to absorption:
$\pi \ext {\phi(u)} = x y^\infty + z^\infty$.
These correspond to the \emph{simplest} infinite paths
since monomials such as $z^2 x y^\infty$ (corresponding to the path $uuuvvv\dots$)
are absorbed by $x y^\infty$.
\eex

One consequence of the universal property is the existence
of a \emph{most general} $\Sinf[X]$-\emph{interpretation} $\mgpi$
by introducing variables $X = \{ x_L \mid L \in \Atoms_A(\tau) \cup \NegAtoms_A(\tau) \}$
for all literals and setting $\mgpi(L) = x_L$.
Any other $K$-interpretation $\pi$ (where $K$ is fully continuous and absorptive)
results from $\mgpi$ by the evaluation $x_L \mapsto \pi(L)$
which lifts to a fully continuous homomorphism $h$.
After computing $\mgpi \ext \phi$ once,
the computation for any $\pi$
is then simply a matter of applying polynomial evaluation,
since $\pi \ext \phi = h(\mgpi \ext \phi)$.

The most general $\Sinf[X]$-interpretation can also be used
to prove that the update operators $F_\pi^\phi$ induced by
LFP-formulae in $\Sinf[X]$ are fully continuous.
Hence Kleene's Fixed-Point Theorem applies and guarantees
that the fixed-point iterations for $\lfp(F_\pi^\phi)$
and $\gfp(F_\pi^\phi)$ have closure ordinal at most $\omega$.
Using the universal property,
the statement on the closure ordinal generalizes
to all absorptive, fully continuous semirings --
even to semirings in which update operators
are not fully continuous in general,
such as the semiring $\Luk$ in \cref{ex:noncontinuousLukasiewicz}.
To see how the example is related,
consider the formula $\phi(R,x) = \gfpfml P y {R x \land P y} x$
over a singleton universe $A$. We can then identify functions
$A \to K$ with elements of $K$, so that $F_\pi^\phi : K \to K$.
Setting $G = F_\pi^\phi$ in \cref{ex:noncontinuousLukasiewicz}
shows that the update operator $F_\pi^\phi$ is not fully continuous.
In $\Sinf[X]$, on the other hand, infinitely ascending chains such as
$(1 - \frac{1}{1+n})_{n < \omega}$ used in the example cannot be defined,
and $F_\pi^\phi$ is fully continuous.

\begin{proposition}
\label{propSinfKleene}
Given a $\Sinf[X]$-interpretation $\pi$ and an LFP-formula $\phi(R,\tup x)$,
the associated update operator $F_\pi^\phi$ is a fully continuous function.
\end{proposition}

\begin{proof}
 In order to prove that $F_\pi^\phi$ is fully continuous,
 we show the more general statement that for any LFP-sentence $\phi$,
 the mapping $\pi \mapsto \pi \ext \phi$ is fully continuous.
 That is, for a chain $C$ of $\Sinf[X]$-interpretations,
 we have $(\Sup C) \ext \phi = \Sup \{ \pi \ext \phi \mid \pi \in C \}$
 (and the same for infima).
 The continuity of $F_\pi^\phi$ follows by unraveling the definition
 of the update operator.

 As the set of $\Sinf[X]$-interpretations is countable,
 it suffices to consider $\omega$-chains $(\pi_i)_{i < \omega}$
 due to \cref{lemCountableChains}.
 To simplify notation, let $\pi_\omega = \Sup_{i < \omega} \pi_i$.
 Now let $X' = \{ x_L \mid L \in \Atoms_A(\tau) \cup \NegAtoms_A(\tau) \}$ and
 consider the most general $\Sinf[X']$-interpretation
 $\mgpi$ with $\mgpi(L) = x_L$.
 By the universal property, the mapping $x_L \mapsto \pi(L)$
 extends to a fully continuous homomorphism $h_\pi: \Sinf[X'] \to \Sinf[X]$
 with $\pi = h_\pi \circ \mgpi$,
 for any $\Sinf[X]$-interpretation $\pi$.
 Using these homomorphisms, we can reformulate
 the continuity statement we want to prove:
 \[
 h_{\pi_\omega}(\mgpi \ext \phi) =
 \Sup_{i < \omega} h_{\pi_i}(\mgpi \ext \phi).
 \]
 
 We claim that $h_{\pi_\omega}(m) = \Sup_{i < \omega} h_{\pi_i}(m)$
 for all monomials $m$ over $X'$.
 Since $\mgpi \ext \phi$ consists of finitely many monomials,
 this implies the statement above by applying the Splitting Lemma \ref{lemSplitting}.
 Monomials in $\Sinf[X']$ are products consisting of factors
 of the form $x_L^n$ for $x_L \in X'$ and $n \in \N \cup \{\infty\}$.
 By again resorting to \cref{lemSplitting},
 it suffices to show the claim for monomials of the form $x_L^n$.
 For such monomials,
 \[
 h_{\pi_\omega}(x_L^n) =
 h_{\pi_\omega}(x_L)^n =
 \pi_\omega(L)^n =
 \Big(\Sup_{i < \omega} \pi_i(L)\Big)^n \overset{(*)}{=}
 \Sup_{i < \omega} (\pi_i(L)^n) =
 \Sup_{i < \omega} h_{\pi_i}(x_L^n),
 \]
 where $(*)$ can be seen by case distinction.
 For $n < \infty$, it follows (once again) from \cref{lemSplitting}.
 For $n = \infty$, we can apply \cref{lemInfpow} \eqref{lemInfpow-sup}.
 This proves the statement about suprema.
 For infima, i.e.\ $(\Inf C) \ext \phi = \Inf \{ \pi \ext \phi \mid \pi \in C \}$,
 the argument is analogous, except that $(*)$ now requires \cref{lemInfpow} \eqref{lemInfpow-inf}.
\end{proof}

By Kleene's Fixed-Point Theorem (and its dualized version for greatest fixed points),
the fixed-point iterations for $\lfp(F_\pi^\phi)$ and $\gfp(F_\pi^\phi)$
both terminate at step $\omega$ (or earlier).
We can generalize this observation to other semirings by the fundamental property.
To be precise, we need the slightly stronger statement that fully continuous homomorphisms
preserve not only the fixed-points, but also all steps of the fixed-point iterations.
The proof of fundamental property (see \cref{propFundamentalProperty}) in fact establishes this stronger statement.

\begin{corollary}
\label{propClosureOrdinal}
Given a $K$-interpretation $\pi$ into an absorptive, fully continuous semiring,
all fixed-point iterations for $\lfp(F_\pi^\phi)$ and $\gfp(F_\pi^\phi)$ have closure ordinal at most $\omega$.
\end{corollary}

\begin{proof}
The statement follows from \cref{propSinfKleene} by considering the most general interpretation
$\mgpi$ defined above and observing that the fully continuous homomorphism $h : \Sinf[X'] \to K$
induced by the mapping $x_L \mapsto \pi(L)$ preserves all steps of the fixed-point iteration.
\end{proof}

What we still have to provide for an adequate provenance analysis is a proper treatment of negation:
If we track a literal and its negation by different variables $x$ and $y$, respectively,
we may obtain inconsistent monomials such as $xy$.
As in other semirings of polynomials and power series we can also here 
take pairs of positive and negative indeterminates, with
a correspondence $X\iff \nnX$, and build
the quotient with respect to the congruence generated by 
the equation $x\cdot\nnx =0$.
We thus obtain a new semiring $\Sinf[X,\nnX]$ which, as a quotient,
retains the properties of being absorptive, fully continuous and chain-positive.
Of course, $\Sinf[X, \nnX]$ is no longer positive, as $x$ and $\nnx$ are divisors of $0$.
%(to avoid inconsistent monomials such as $x \nnx$).
%
Most importantly, $\Sinf[X, \nnX]$ inherits the universal property:
Given a mapping $h : X \cup \nnX \to K$, we obtain $h : \Sinf[X \cup \nnX] \to K$ by \cref{universality-of-Sinf};
if $h$ respects dual-indeterminates, so $h(x) \cdot h(\nnx) = 0$ for all $x \in X$,
then $h$ factors through the quotient and induces $h : \Sinf[X,\nnX] \to K$.
Together with the fundamental property, $\Sinf[X,\nnX]$ is thus
the most general appropriate provenance semiring for LFP that can represent negation,
hence providing a natural framework for a provenance analysis for
LFP and other fixed point calculi.

\begin{corollary}[Universality]
Every mapping $h : X \cup \nnX \to K$ into an absorptive, fully continuous semiring $K$
that satisfies $h(x) \cdot h(\nnx) = 0$ for all $x \in X$
uniquely extends to a fully continuous semiring homomorphism $h : \Sinf[X,\nnX] \to K$.
\end{corollary}

Instead of model-defining interpretations,
we consider \emph{model-compatible} interpretations $\pi$.
That is, for each atom $R \tup a$ we either have
$\pi(R \tup a) = x$ and $\pi(\neg R \tup a) = \nnx$,
or $\{ \pi(R \tup a), \pi(\neg R \tup a) \} = \{ 0, 1 \}$.
Additionally, $\pi$ must not use the same indeterminate for two different atoms.
We say that a model $\AA$ is \emph{compatible} with $\pi$ if $\AA \models L$
for all literals $L$ with $\pi(L) = 1$
and denote the set of compatible models by $\Mod_\pi$.
Model-compatible interpretations can be used to reason about several models at once.
Mapping certain literals to indeterminate pairs $x$ and $\nnx$
leaves open the truth of these literals,
but still encodes the semantics of opposing literals:
%We can then formulate the following analogue of truth-preservation.

\begin{proposition}
\label{propModelCompatible}
Let $\pi$ be a model-compatible $\Sinf[X, \nnX]$-interpretation.
An LFP-formula $\phi$ is
$\Mod_\pi$-satisfiable ($\Mod_\pi$-valid) if, and only if, $\pi \ext \phi \neq 0$ ($\pi \ext {\neg \phi} = 0$).
%Moreover, $\phi$ is $\Mod_\pi$-valid if, and only if, $\pi \ext {\neg \phi} = 0$.
\end{proposition}

\begin{proof}
The statement on satisfiability implies the one on validity, so we only consider the former.
If $\AA \in \Mod_\pi$ and $\AA \models \phi$, we consider the
model-defining $\B$-interpretation $\pi_\AA$ corresponding to the model $\AA$.
We can obtain $\pi_\AA$ from $\pi$ by instantiating the indeterminates with values from $\B$.
Let $h : X \cup \nnX \to \B$ be this instantiation
and observe that, since $\pi_\AA$ is model-defining, $h(x) \cdot h(\nnx) = 0$.
By the universal property, this induces a fully continuous homomorphism
$h : \Sinf[X, \nnX] \to \B$ such that $\pi_\AA = h \circ \pi$.
It follows from the fundamental property that
$h(\pi \ext \phi) = \pi_\AA \ext \phi = 1$
and, since $h(0) = 0$, we thus have $\pi \ext \phi \neq 0$.

For the other direction, assume that $\pi \ext \phi \neq 0$.
Then there is a monomial $m \in \pi \ext \phi$.
This monomial induces an instantiation $h : X \cup \nnX \to \B$ such that $h(m) = 1$ as follows:
\begin{itemize}
\item If $m(x) > 0$, then $h(x) = 1$ and $h(\nnx) = 0$,
\item if $m(\nnx) > 0$, then $h(x) = 0$ and $h(\nnx) = 1$,
\item otherwise, $h(x) = 1$ and $h(\nnx) = 0$ (this is an arbitrary choice).
\end{itemize}
By construction, $h$ respects dual-indeterminates and thus
lifts to a fully continuous homomorphism $h : \Sinf[X, \nnX] \to \B$.
Moreover, $h \circ \pi$ is a model-defining $\B$-interpretation.
It follows the induced model $\AA_{h \circ \pi}$ satisfies $\phi$,
since $(h \circ \pi) \ext \phi = h(\pi \ext \phi) \ge h(m) = 1$.
\end{proof}

\section{Game-theoretic analysis}
\label{Sect:Games}

It has been shown in \cite{GraedelTan20} that the provenance analysis for $\FO$ and 
$\posLFP$ is intimately connected 
with the provenance analysis of reachability games. 
Evaluation strategies to establish the truth of first-order formulae 
are really winning strategies for reachability games on acyclic game graphs.
%(all whose plays are finite).
For $\posLFP$ the situation is similar, but the 
associated model checking games may have cycles and thus admit infinite plays, but 
the winning plays for the verifying player have to reach a winning position 
(a true literal) in a finite number of steps.
%and the infinite plays are always losing plays for that player. 
By annotating such terminal positions with semiring values
and propagating these values along the edges to the remaining positions,
one obtains provenance values that coincide
with the syntactically defined semantics $\pi \ext \psi$.

%The provenance analysis for reachability games has been systematically developed
%in \cite{GraedelTan20}.
%By annotating terminal positions with semiring values and
%then propagating these values along the edges to non-terminal positions,
%one can compute semiring provenance values for 
%each position of a reachability game, and when applied to
%model checking games, this leads to a definition
%of semiring provenance for formulae in $\FO$ and $\posLFP$ 
%which coincides with the inductive one via the syntax of formulae.

For full LFP or the modal $\mu$-calculus, the model checking games are parity games
which are considerably more complex and do not allow
for a simple propagation of values from terminal positions.
We do not present here a general provenance
analysis of parity games, but we show how 
provenance values $\pi \ext \phi$ for fixed-point formulae
can be understood from a game-theoretic point of view.
For first-order logic or $\posLFP$, provenance values $\pi \ext \phi$ in $\N[X,\nnX]$
or $\N^\infty[\![X,\nnX]\!]$ are sums of monomials
that correspond to the evaluation strategies for $\phi$
and provide information about the literals used by these strategies.
%in $\N^\infty[\![X,\nnX]\!]$ are sums of monomials
%$x_1^{e_1} \cdots x_k^{e_k}$
%that correspond to evaluation strategies in
%which precisely the literals labelled by $x_1, \dots, x_k$ are used,
%and the strategy admits precisely $e_i$ plays ending at the literal
%labelled by $x_i$. 
We present an analogue of this statement for full fixed-point logic
and the semiring $\Sinf[X,\nnX]$. 

\medskip\noindent{\bf Model-checking games for LFP. }
Model checking games are classically defined for a formula and a fixed structure $\AA$
(see e.g. \cite[Chap.~4]{AptGraedel11}).
However, the \emph{game graph} of such a game depends only on
the formula $\psi$ and the \emph{universe} of the given structure, and
it is only the labelling of the terminal positions as winning 
for either the Verifier (Player~0) or the Falsifier (Player~1), that depends
on which of the literals in $\Lit_A(\tau)$ are true in $\AA$.
%Hence the definition of a model checking game readily generalizes to a more abstract provenance scenario. 
Hence the definition readily generalizes to a more abstract provenance
scenario where we instead label terminal positions by semiring values.

\bdefn Let $\psi$  be an LFP-sentence in negation normal form
with a relational vocabulary $\tau$, and let $A$ be a (finite) universe.
The model checking game $\Gg(A,\psi)$ has
positions $\phi(\tup a)$, obtained from a subformula $\phi(\tup x)$ of $\psi$, 
by instantiating the free variables $\tup x$ by a tuple $\tup a$ of elements of $A$. 
At a disjunction $(\psi\lor\phi)$, Player~0 (Verifier) moves to either $\psi$ or $\phi$,
and at a conjunction, Player~1 (Falsifier) makes an analogous move.
At a position $\E x\phi(\tup a,x)$, Verifier selects an element $b$ and moves
to $\phi(\tup a,b)$, whereas at positions $\A x \phi(\tup a,x)$ the move to
to the next position $\phi(\tup a,b)$ is done by Falsifier.  
For every subformula of $\psi$ of form
$\theta\coloneqq[\lfp R\tup x \st \phi(R,\tup x)](\tup x)$ or 
$\theta\coloneqq[\gfp R\tup x \st \phi(R,\tup x)](\tup x)$
we add moves from positions $\theta(\tup a)$ to $\phi(\tup a)$, and from
positions $R\tup a$  to $\phi(\tup a)$ for  every tuple $\tup a$. 
Since these moves are unique it makes no difference to which of the two players
we assign the positions $\theta(\tup a)$ and $R\tup a$.
The resulting game graphs $\Gg(A,\psi)$ may contain cycles,
but the set $T$ of terminal nodes is again a subset of $\Lit_A(\tau)$. 
The terminal positions of $\Gg(A,\psi)$ are  literals in $\Lit_A(\tau)$.
\edefn

These games may have cycles and thus admit infinite plays. The  
winning condition for infinite plays is the \emph{parity condition}:  We assign to 
each fixed-point variable a priority, which is even for greatest fixed-points and odd 
for least fixed points, satisfying the condition if a variable $R$ depends
on another variable $T$ then the priority of $R$ is smaller or equal to the priority of $T$.
An infinite play is won by Player~0 (the Verifier) if the least priority occurring infinitely often in
the play is even, otherwise it is won by Player~1 (the Falsifier).

\medskip\noindent{\bf Provenance values for plays and strategies.} 
Given a parity game $\Gg(A,\psi)$, every $K$-interpretation $\pi : \Lit_A(\tau) \to K$
provides a valuation of the terminal positions.
Based on this, we define provenance values for plays and strategies.

\bdefn A \emph{finite} play $\rho=(\phi_0,\dots,\phi_t)$ ends in a terminal position
$\phi_t\in\Lit_A(\tau)$ which we call the \emph{outcome} of $\rho$. We simply identify the
provenance value of $\rho$ with the value of its outcome, i.e. we
put $\pi\ext\rho\coloneqq\pi\ext{\phi_t}$. For an \emph{infinite} play $\rho$ we put
$\pi\ext\rho\coloneqq1$ if $\rho$ is a wining play for the Verifier, and $\pi\ext\rho\coloneqq0$ 
otherwise.  
\edefn

We denote by $\Strat(\phi)$ the set of evaluation strategies for the subformula $\phi$ of $\psi$,
i.e. the set of all (not necessarily positional) strategies that the Verifier has from position
$\phi$ in the parity game $\Gg(A,\psi)$.
Every strategy $\Ss\in\Strat(\phi)$ induces the set $\Plays(\Ss)$ of plays that are consistent with
$\Ss$. Intuitively, the provenance value of a strategy is simply the product over the
provenance values of all plays that it admits.
However, a strategy may well admit an infinite set of plays
and while it is possible to define infinite products in our setting
(we refer to the appendix for details),
we instead observe that the set of possible outcomes is of course finite,
since there exist only finitely many literals.
As a consequence, we define the provenance value for a strategy
by grouping those plays with identical outcome.

\bdefn
For any strategy $\Ss$ and any literal $L \in \Lit_A(\tau)$, we write 
$\litcount \Ss L \in \N\cup\{\infty\}$ for the number of plays $\rho\in\Plays(\Ss)$ with outcome 
$L$. We then define the provenance value
\[
 \pi \ext \Ss \coloneqq \begin{cases} \prod_{L \in \Lit_A(\tau)}  \pi(L) ^{\litcount \Ss L} &\text{ if all infinite $\rho \in \Plays(\Ss)$ are winning for Verifier},\\
 0 &\text{ otherwise.} \end{cases}
\]
\edefn

The case for $\#_\Ss(L) = \infty$ is well-defined,
as the infinitary power $a^\infty = \Inf_n a^n$ can be defined in
all absorptive, fully continuous semirings.
For model-compatible interpretations in $\Sinf[X, \nnX]$,
the value $\pi \ext \Ss$ is a single monomial.
The following central result justifies our game-theoretic analysis
and precisely characterizes provenance semantics $\pi \ext \psi$
in terms of strategies in the associated model checking game.

\begin{theorem}
\label{thmStrategyCharacterization}
 Let $\psi\in\LFP$, and and let $\pi\colon\Lit_A(\tau)\ra K$  be a $K$-interpretation 
 into an  absorptive, fully continuous semiring $K$.
 Then $\pi \ext \psi = \Sup \{ \pi \ext \Ss \mid \Ss \in \Strat(\psi) \}$.
\end{theorem}

Examples of model checking games given in the following section,
the proof is discussed in detail in Sect.~{\ref{Sect:Games-Proof}}.
The key idea is to view strategies $\Ss$ in the game of, say, $\gfpfml R {\tup x} \phi {\tup a}$
as trees and then define prefixes $\Ss \cut n$ of these trees based on the number of fixed-point literals $R \tup b$ along a path.
We prove by induction that these prefixes of increasing size correspond exactly to the steps of the fixed-point iteration via $F_\pi^\phi$.
For greatest fixed points, strategies can be infinite which leads to subtle obstacles.
Perhaps the most challenging step is the so-called \emph{puzzle lemma}
which shows that, roughly speaking, computing infima of
strategy prefixes leads to meaningful values corresponding to actual (infinite) strategies.

%Model checking games become large even for simple formulae
%over a small universe $A$; we thus refer to \cite[Sect.~XX]{DannertGraNaaTan19} for illustrating examples.
%The proof of this result \cite[Sect.~YY]{DannertGraNaaTan19} is not short either.

Consider now specifically the semiring $\Sinf[X,\nnX]$ and
model-compatible interpretations.
By the above theorem,
the provenance value of a sentence $\psi$ is then a sum of monomials
$x_1^{e_1} \cdots x_k^{e_k}$, each of which corresponds to a strategy $\Ss$ for Verifier
that uses precisely the literals labelled by $x_1, \dots, x_k$, and
each literal $x_i$ is used precisely $e_i$ many times,
that is, there are $e_i$ plays consistent with $\Ss$ that have outcome $x_i$.
By using dual indeterminates, % $X$ and $\nnX$,
we make sure that these literals are consistent
and hence represent actual evaluation strategies.
%By using dual indeterminates $X$ and $\nnX$, we further see whether a literal or
%its negation are used and we disregard inconsistent strategies that
%use both an atom and its negation.
%We may thus view the strategy $\Ss$ as a witness for the truth of $\psi$.
%

In this sense, provenance semantics in absorptive, fully continuous semirings,
and most prominently in $\Sinf[X, \nnX]$,
provide detailed information about evaluation strategies.
Because of absorption, we do not obtain information about all evaluation strategies,
as in first-order logic and $\N[X, \nnX]$, but instead only about the absorption-dominant
strategies, corresponding to absorption-maximal monomials.
These are strategies that allow the fewest different possible outcomes
and are thus the simplest or canonical evaluation strategies.

\subsection{Examples}
\label{Sect:Games-Examples}

Before we prove \cref{thmStrategyCharacterization},
let us illustrate provenance values for strategies with two examples.
Since model checking games become large even for simple formulae and
small universes, we only consider a small graph with two nodes.
The formula, on the other hand, features alternating least and greatest fixed points
which is arguably the most difficult case to analyse
and leads to more complicated parity games that need several different priorities.

\ex
\label{exModelChecking}
Consider the formula $\phi(u)$ below which expresses that
there is a path from $u$ on which $P$ holds infinitely often.
We evaluate $\phi(u)$ using the model-compatible $\Sinf[X, \nnX]$-interpretation $\pi$
over $A = \{ u, v \}$ indicated on the right, with $\pi(Pu) = 0$ and $\pi(P v) = 1$.

\begin{center}
$\displaystyle \phi(u) = \Gfpfml X x {\lfpfml Y x {\E y \big(Exy \land ((Xy \land Py) \lor Yy)\big)} x \,} u$
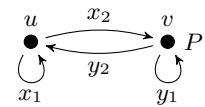
\hspace{1cm}
\begin{tikzpicture}[baseline, yshift=2.5pt, node distance=1.5cm]
\node [dot,label={above:$u$}] (u) {};
\node [dot,label={above:$v$},label={right:{$P$}}, shift=(0:1.8cm)] at (u) (v) {};

\path [draw, ->, >=stealth', shorten >=3pt, shorten <=3pt]
(u) edge [bend left=15pt] node [above] {$x_2$} (v)
(v) edge [bend left=15pt] node [below] {$y_2$} (u)
(u) edge [in=-60,out=-120,looseness=18] node [below] {$x_1$} (u)
(v) edge [in=-60,out=-120,looseness=18] node [below] {$y_1$} (v)
;
\end{tikzpicture}
\end{center}

Intuitively, witnesses for $\phi(u)$ are simply infinite paths that infinitely often visit $v$.
There are infinitely many such paths, but the simplest ones (in terms of the different edges they use)
are the paths $uvvvv\dots$ and $uvuvuv\dots$ which correspond to the monomials $x_2 y_1^\infty$
and $x_2^\infty y_2^\infty$. And indeed, $\pi \ext {\phi(u)} = x_2 y_1^\infty + x_2^\infty y_2^\infty$.
Notice that the edge $x_1$ does not appear in the result and we can conclude
that its existence does not affect the truth of $\phi(u)$.

Let us now consider the evaluation strategies for $\phi(u)$ from the game-theoretic perspective.
The complete model checking game (with abbreviated node labels) is shown below,
where rounded nodes belong to Verifier, rectangular nodes to Falsifier
and the small numbers indicate the priorities assigned to fixed-point relations.
Terminal positions are highlighted by dashed borders and include the value assigned by $\pi$.
There are four positions for which Verifier can make a decision:
The two nodes labeled $\E y (\dots)$ and the two disjunctions in the center of the figure.
Hence there are 16 positional strategies in total.
One of these strategies is highlighted in gray and has the provenance value $x_2 y_1^\infty$,
as there is one play ending in $Euv$ and there are arbitrarily long plays ending
either in $Evv$ or in $Pv$, depending on the choices of Falsifier.
Most of the other 15 strategies allow infinite paths with least priority $1$
and thus have provenance value $0$ (for instance by choosing the cycle
$\E y (\dots) \;\to\; E u u \land \ldots \;\to\; (Xu \land Pu) \lor Yu \;\to\; Yu$).
The only remaining strategy has the provenance value $x_2^\infty y_2^\infty$.
One can further observe that non-positional strategies only lead to
monomials with additional variables which are then absorbed,
so we indeed obtain $\pi \ext {\phi(u)} = x_2 y_1^\infty + x_2^\infty y_2^\infty$.

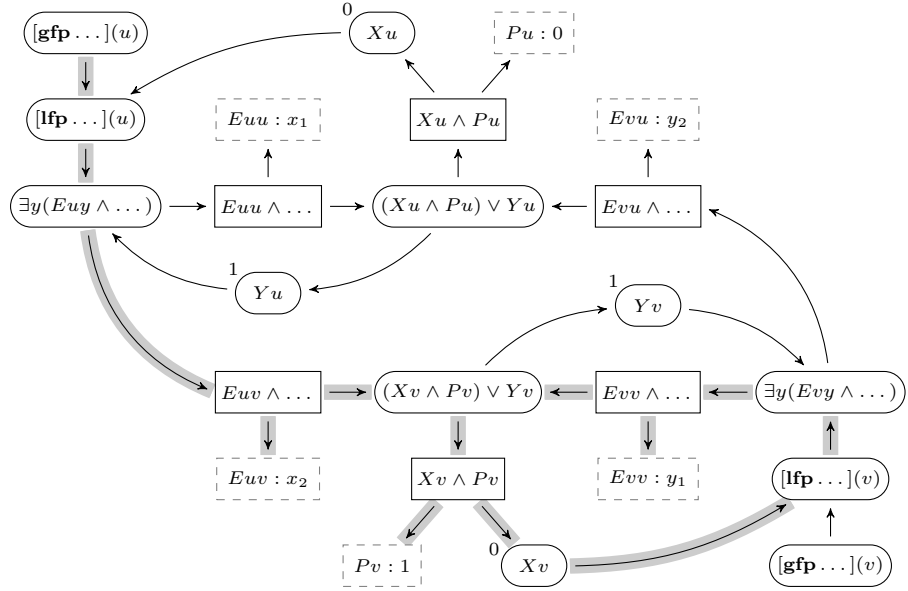
\begin{center}
\begin{tikzpicture}[font=\scriptsize, node distance=1.15cm]
\node [p0] (gfpu) {$[\gfp \dots ](u)$};
\node [p0, fill=white, below of=gfpu] (lfpu) {$[\lfp \dots ](u)$};
\node [p0, below of=lfpu] (exu) {$\E y (E u y \land \dots)$};

\node [p1, right of=exu, xshift=1.25cm] (Euu) {$E u u \land \dots$};
\node [p1, below=1.9cm of Euu] (Euv) {$E u v \land \dots$};

\node [p0, right of=Euu, xshift=1.35cm] (oru) {$(X u \land P u) \lor Y u$};
\node [p1, above of=oru] (andu) {$X u \land P u$};
\node [p0, above of=andu, xshift=-1cm] (Xu) {$X u$};
\node [tt, above of=andu, xshift=+1cm] (Pu) {$P u: 0$};

\node [tt,above of=Euu] (EuuT) {$E u u : x_1$};
\node [tt,below of=Euv] (EuvT) {$E u v : x_2$};
\node [p0, below of=Euu] (Yu) {$Y u$};

\node [p0, right of=Euv, xshift=1.35cm] (orv) {$(X v \land P v) \lor Y v$};
\node [p1, below of=orv] (andv) {$X v \land P v$};
\node [p0, below of=andv, xshift=+1cm] (Xv) {$X v$};
\node [tt, below of=andv, xshift=-1cm] (Pv) {$P v: 1$};

\node [p1, right of=orv, xshift=1.35cm] (Evv) {$E v v \land \dots$};
\node [p1, above=1.9cm of Evv] (Evu) {$E v u \land \dots$};

\node [tt,below of=Evv] (EvvT) {$E v v : y_1$};
\node [tt,above of=Evu] (EvuT) {$E v u : y_2$};
\node [p0, above of=Evv] (Yv) {$Y v$};

\node [p0, right of=Evv, xshift=1.25cm] (exv) {$\E y (E v y \land \dots)$};
\node [p0, below of=exv] (lfpv) {$[\lfp \dots ](v)$};
\node [p0, below of=lfpv] (gfpv) {$[\gfp \dots ](v)$};

\begin{scope}[on background layer]
\path [->, >=stealth', shorten >=3pt, shorten <=3pt]
(gfpu) edge [win] (lfpu)
(lfpu) edge [win] (exu)
(exu) edge (Euu) (exu) edge [bend right, win] (Euv.west)
(Euu) edge (EuuT) (Euu) edge (oru)
(oru) edge [bend left=20pt] (Yu) (oru) edge (andu)
(Yu) edge [bend left=20pt] (exu)
(andu) edge (Xu) (andu) edge (Pu)
(Xu) edge [bend right=15pt] (lfpu.north east)
(Euv) edge [win] (EuvT)

(gfpv) edge (lfpv)
(lfpv) edge [win] (exv)
(exv) edge [win] (Evv) (exv) edge [bend right] (Evu.east)
(Evv) edge [win] (EvvT) (Evv) edge [win] (orv)
(orv) edge [bend left=20pt] (Yv) (orv) edge [win] (andv)
(Yv) edge [bend left=20pt] (exv)
(andv) edge [win] (Xv) (andv) edge [win] (Pv)
(Xv) edge [bend right=15pt, win] (lfpv.south west)
(Evu) edge (EvuT)

(Evu) edge (oru)
(Euv) edge [win] (orv)
;
\end{scope}

\begin{scope}[node distance=0.5cm, every node/.style={prio, xshift=-.1cm, yshift=-1pt}]
\node [above left of=Yu] {1};
\node [above left of=Yv] {1};
\node [above left of=Xu] {0};
\node [above left of=Xv, yshift=-3pt, xshift=-2pt] {0};
\end{scope}
\end{tikzpicture}
\end{center}

If we are just interested in the question \emph{which} literals are needed to satisfy $\phi(u)$
or, equivalently, in which models $\phi(u)$ holds, we can drop all exponents and obtain
$\pi \ext {\phi(u)} = x_2 y_1 + x_2 y_2$.
This is the same result we would obtain in the semiring $\PosBool(X,\nnX)$
(the dual-indeterminate version of $\PosBool(X)$).
We see that $\phi(u)$ holds in all models that satisfy $P v$, $\neg P u$ and additionally
contain at least the edges $x_2$ and $y_1$, or at least $x_2$ and $y_2$.
\eex

\ex
\label{exModelCheckingNegation}
In the previous example, the interpretation of the evaluation strategies
in terms of infinite paths was straightforward.
If we instead consider the negated formula $\neg \phi(u)$, which states that there
are only finitely many occurrences of $P$ on all paths from $u$, witnesses for the
truth of $\phi(u)$ are more complex and are best understood through the model checking game.
We first bring $\neg \phi(u)$ into negation normal form using the duality laws of LFP:
\begin{center}
$\displaystyle \neg \phi(u) \equiv \Lfpfml X x {\gfpfml Y x {\A y \big(\neg Exy \lor ((Xy \lor \neg Py) \land Yy)\big)} x \,} u$
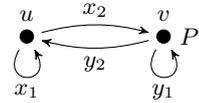
\hspace{.7cm}
\begin{tikzpicture}[baseline, yshift=2.5pt, node distance=1.5cm]
\node [dot,label={above:$u$}] (u) {};
\node [dot,label={above:$v$},label={right:{$P$}}, shift=(0:1.8cm)] at (u) (v) {};

\path [draw, ->, >=stealth', shorten >=3pt, shorten <=3pt]
(u) edge [bend left=15pt] node [above] {$x_2$} (v)
(v) edge [bend left=15pt] node [below] {$y_2$} (u)
(u) edge [in=-60,out=-120,looseness=18] node [below] {$x_1$} (u)
(v) edge [in=-60,out=-120,looseness=18] node [below] {$y_1$} (v)
;
\end{tikzpicture}
%
%\vspace{-1ex}
\end{center}

We analyse the game as in the previous example
(the game graph is shown below).
Again, the provenance values of non-positional strategies
are absorbed by the values of positional ones.
The players have basically switched roles, but Verifier can still make
relevant decisions for only four nodes.
One possible strategy is highlighted in gray above
and has the provenance value $\nn x_1 \nn y_1^2 \nn y_2^2$.
Notice the exponent $2$, as the position $\A y (\neg E vy \lor \ldots)$
can be reached in two different ways and there are hence two plays
with value $\nn y_1$ and two with value $\nn y_2$.
The other positional strategy with only finite plays has the provenance value $\nn x_1 \nn x_2$.
Most of the positional strategies with infinite plays
admit an infinite play with priority $1$ and thus have value $0$,
except for the two strategies with values $\nn x_2^\infty$
and $\nn y_1^\infty \nn y_2^\infty$, respectively.
We thus obtain
$\pi \ext {\neg \phi(u)} = \nn{x}_{1} \nn{x}_{2} + \nn{x}_{1} \nn{y}_{1}^2 \nn{y}_{2}^2 + \nn{x}_{2}^\infty + \nn{y}_{1}^\infty \nn{y}_{2}^\infty$.

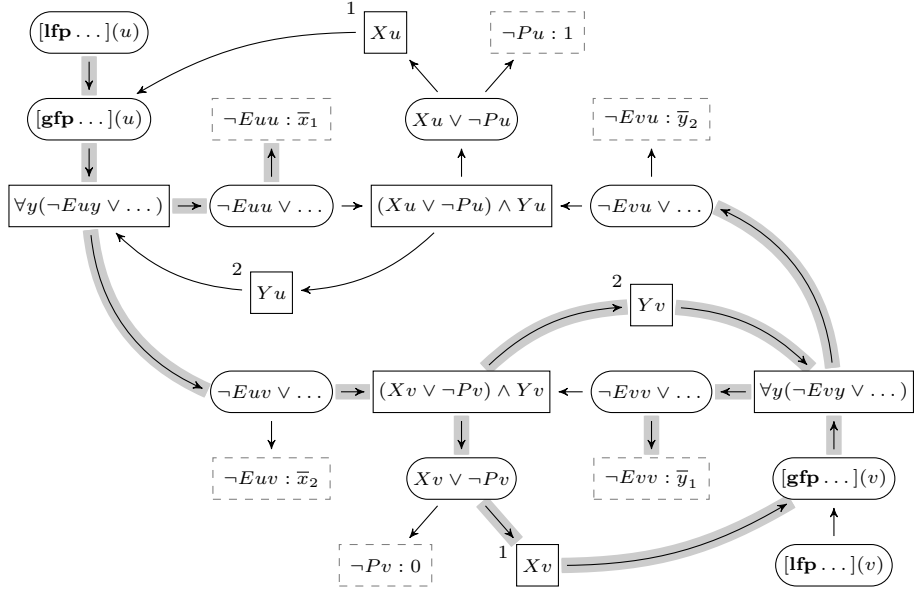
\begin{center}
\begin{tikzpicture}[font=\scriptsize, node distance=1.15cm]
\node [p0] (gfpu) {$[\lfp \dots ](u)$};
\node [p0, fill=white, below of=gfpu] (lfpu) {$[\gfp \dots ](u)$};
\node [p1, below of=lfpu] (exu) {$\A y (\neg E u y \lor \dots)$};

\node [p0, right of=exu, xshift=1.25cm] (Euu) {$\neg E u u \lor \dots$};
\node [p0, below=1.9cm of Euu] (Euv) {$\neg E u v \lor \dots$};

\node [p1, right of=Euu, xshift=1.35cm] (oru) {$(X u \lor \neg P u) \land Y u$};
\node [p0, above of=oru] (andu) {$X u \lor \neg P u$};
\node [p1, above of=andu, xshift=-1cm] (Xu) {$X u$};
\node [tt, above of=andu, xshift=+1cm] (Pu) {$\neg P u: 1$};

\node [tt,above of=Euu] (EuuT) {$\neg E u u : \nn{x}_{1}$};
\node [tt,below of=Euv] (EuvT) {$\neg E u v : \nn{x}_{2}$};
\node [p1, below of=Euu] (Yu) {$Y u$};

\node [p1, right of=Euv, xshift=1.35cm] (orv) {$(X v \lor \neg P v) \land Y v$};
\node [p0, below of=orv] (andv) {$X v \lor \neg P v$};
\node [p1, below of=andv, xshift=+1cm] (Xv) {$X v$};
\node [tt, below of=andv, xshift=-1cm] (Pv) {$\neg P v: 0$};

\node [p0, right of=orv, xshift=1.35cm] (Evv) {$\neg E v v \lor \dots$};
\node [p0, above=1.9cm of Evv] (Evu) {$\neg E v u \lor \dots$};

\node [tt,below of=Evv] (EvvT) {$\neg E v v : \nn{y}_{1}$};
\node [tt,above of=Evu] (EvuT) {$\neg E v u : \nn{y}_{2}$};
\node [p1, above of=Evv] (Yv) {$Y v$};

\node [p1, right of=Evv, xshift=1.25cm] (exv) {$\A y (\neg E v y \lor \dots)$};
\node [p0, below of=exv] (lfpv) {$[\gfp \dots ](v)$};
\node [p0, below of=lfpv] (gfpv) {$[\lfp \dots ](v)$};

\begin{scope}[on background layer]
\path [->, >=stealth', shorten >=3pt, shorten <=3pt]
(gfpu) edge [win] (lfpu)
(lfpu) edge [win] (exu)
(exu) edge [win] (Euu) (exu) edge [bend right, win] (Euv.west)
(Euu) edge [win] (EuuT) (Euu) edge (oru)
(oru) edge [bend left=20pt] (Yu) (oru) edge (andu)
(Yu) edge [bend left=20pt] (exu)
(andu) edge (Xu) (andu) edge (Pu)
(Xu) edge [bend right=15pt] (lfpu.north east)
(Euv) edge (EuvT)

(gfpv) edge (lfpv)
(lfpv) edge [win] (exv)
(exv) edge [win] (Evv) (exv) edge [bend right,win] (Evu.east)
(Evv) edge [win] (EvvT) (Evv) edge (orv)
(orv) edge [bend left=20pt,win] (Yv) (orv) edge [win] (andv)
(Yv) edge [bend left=20pt,win] (exv)
(andv) edge [win] (Xv) (andv) edge (Pv)
(Xv) edge [bend right=15pt,win] (lfpv.south west)
(Evu) edge (EvuT)

(Evu) edge (oru)
(Euv) edge [win] (orv)
;
\end{scope}

\begin{scope}[node distance=0.5cm, every node/.style={prio, xshift=-.1cm, yshift=-1pt}]
\node [above left of=Yu] {2};
\node [above left of=Yv] {2};
\node [above left of=Xu] {1};
\node [above left of=Xv, yshift=-4pt, xshift=0pt] {1};
\end{scope}
\end{tikzpicture}
\end{center}

It is possible, albeit tedious and not straightforward, to verify this result
by manually determining the infimum of the nested fixed-point iteration.
The interpretation of the result beyond the model checking game
is not as clear as in the previous example.
We can, however, again omit the exponents and, due to absorption, obtain the value
$\pi \ext {\neg \phi(u)} = \nn x_2 + \nn y_1 \nn y_2$ in $\PosBool(X,\nnX)$.
This tells us that $\neg \phi(u)$ is satisfied in all models that lack at least edge $x_2$
or have no outgoing edges from $v$ -- exactly complementary to the models
we determined for $\phi(u)$.
\eex

\subsection[Proof]{Proof of \cref{thmStrategyCharacterization}}
\label{Sect:Games-Proof}

The proof of \cref{thmStrategyCharacterization} is more involved
than the proofs presented so far.
We first show that it holds in the semiring $\Sinf[X]$
and afterwards make use of the universal property
to generalize it to all absorptive, fully continuous semirings.
We begin with some notation for model checking games and strategies.
The inductive proof is then based on the notion of strategy
\emph{truncations}, which are essentially prefixes of strategy trees.
We extend the definition of the value $\pi \ext \Ss$ of a strategy
to these prefixes and show by induction that truncations of increasing size
correspond to the steps of the fixed-point iteration for fixed-point formulae.

For the proof in $\Sinf[X]$, we always fix a universe $A$, a signature $\tau$ (which we usually omit),
a $\Sinf[X]$-interpretation $\pi$ and consider model checking games $\Game(A, \phi)$
which we abbreviate by just $\Game(\phi)$.
To make the positions of the game precise, let $\Game(\phi) = (V, V_0, V_1, E, \Omega)$,
where $(V, E)$ is a directed graph and $V = V_0 \dcup V_1$ are the positions owned by Verifier and Falsifier, respectively.
The priorities are given by the node labeling $\Omega : V \to \N$.
To avoid the special case of infinite plays that are losing and thus lead to $\pi \ext \Ss = 0$,
we only consider strategies for Verifier in which all infinite plays are winning for Verifier
and denote their set by $\WinStrat{\phi} \subseteq \Strat(\phi)$.
For $v \in V$, we denote the set of successor positions by $vE = \{ w \mid (v,w) \in E \}$.
Positions $v$ with $vE = \emptyset$ are called \emph{terminal} positions.
A play from position $v_0$ is a (finite or infinite) sequence $\rho = v_0 v_1 v_2 \dots$
such that $(v_i, v_{i+1}) \in E$ for all $i$.
If $\rho$ is finite, it must end in a terminal position (which we call the outcome of $\rho$).

The tree unraveling of a game $\Game(\phi)$, as defined in \cite{GraedelTan20},
is the tree $\TT(\Game(\phi), v_0) = (\hash V, \hash V_0, \hash V_1, \hash E)$ where
$\hash V$ is the set of all finite paths $\tau$ from $v_0$,
$\hash V_\sigma \subseteq V$ is the set of those finite paths ending in a node of player $\sigma$
and $\hash E = \{ (\tau v, \tau v v') \mid \tau v \in \hash V \text{ and } (v,v') \in E \}$.
For $\tau,\tau' \in \hash V$,
we write $\tau \prefix \tau'$ if $\tau$ is a prefix of $\tau'$.
For a node $\tau v$, we call $\Pos(\tau v) = v$ the \emph{position} of $\tau v$.
It is often convenient to identify a node $\tau v$ in the tree unraveling with its position in the original game.
For $\tau = v_0 \dots v_k$, we write $|\tau|=k+1$ for the length of $\tau$.
Following \cite{GraedelTan20}, we view strategies as subtrees of the tree unraveling.

\bdefn
\label{def:strategyAsTree}
A \emph{strategy} $\Ss$ of player $\sigma \in \{0,1\}$ from $v_0$ in $\Game$ is a subtree of $\TT(\Game, v_0)$ of the form $\Ss = (W,F)$ with $W \subseteq \hash V$ and $F \subseteq (W \times W) \cap \hash E$ that satisfies the following conditions. Let $\hashhat V_\sigma$ be the set $\hash V_\sigma$ without terminal nodes (i.e., leaves).

\begin{enumerate}
\item $W$ is closed under predecessors: if $\tau v \in W$, then also $\tau \in W$,
\item player $\sigma$ makes unique choices: if $\tau \in W \cap \hashhat V_\sigma$, then $|\tau F| = 1$,
\item all choices of the opponent are considered: if $\tau \in W \cap \hash V_{1-\sigma}$, then $\tau F = \tau \hash E$.
\end{enumerate}

A play $\rho$ is \emph{consistent} with $\Ss$ if the corresponding path in $\TT(\Game, v_0)$ is contained in $\Ss$.
The strategy $\Ss$ is \emph{winning} if all plays consistent with $\Ss$ are winning (for player $\sigma$).
\edefn

\subsubsection*{Strategy Truncations}

\begin{defn}
Let $\Ss = (W,F)$ be a strategy in $\Game(\phi) = (V, V_0, V_1, E, \Omega)$ and let $R$ be a relation symbol of arity $r$.
Nodes $v \in V$ with $\Pos(v) = R \ta$ (for some $\ta \in A^r$) are called \emph{$R$-nodes}.
For $\tau = v_0 v_1 \dots v_k \in W$, we define
\[
    |\tau|_R = \big| \{ i \mid \Pos(v_i) = R \ta \text{ for some } \ta \in A^r \} \big|
\]
as the number of $R$-nodes occurring along the path.
The \emph{$(R,n)$-truncation} of $\Ss$ is the tree $(W', F')$ defined as follows. Its nodes are finite sequences over $V \cup \{ \Cutsym \}$, where $\Cutsym$ is a special symbol which marks the nodes at which we cut off subtrees of $\Ss$. For $n \ge 1$, we define

\begin{itemize}
\item $W' = \{ \tau \in W \mid |\tau|_R < n \} \;\cup\; \{ \tau \, \Cutsym \mid \text{$\tau \, R \ta \in W$, $|\tau|_R = n-1$} \}$,
\item $F' = F \cap (W' \times W') \; \cup \; \{ (\tau, \tau \, \Cutsym) \mid \tau \, \Cutsym \in W' \}$.
\end{itemize}

\noindent
For $n=0$, we instead set $W' = \{ \Cutsym \}$ and $F' = \emptyset$.
If $R$ is clear from the context, we write $\Ss \cut n$ for the $(R,n)$-truncation of $\Ss$.
\end{defn}

\begin{figure}
\tikzset{
    p0/.style={draw=black, rounded rectangle, inner sep=1pt, minimum size=.35cm, minimum width=.8cm},
    p1/.style={draw=black, rectangle, inner sep=1pt, minimum size=.35cm, minimum width=.6cm},
    tt/.style={draw=black, rectangle, dashed, inner sep=1pt, minimum size=.35cm, minimum width=.6cm},
%    mm/.style={fill=gray!30},
    xx/.style={inner sep=1pt, text depth=.25ex},
    align at top/.style={baseline=(current bounding box.north)},
}
\centering
\begin{tikzpicture}[font=\scriptsize, node distance=.8cm, align at top]
\node at (0,.1cm) {};
\node [p0] (c0) {$\phi(\ta)$};
\node [p1, below of=c0] (c1) {};
\node [p0, below of=c1, xshift=-1.5cm] (l0) {$R\ta$};
\node [p1, below of=l0] (l1) {};
\node [p0, below of=l1, xshift=-0.6cm] (ll0) {$R\ta$};
\node [p1, below of=ll0] (ll1) {};
\node [p0, below of=l1, xshift=+0.6cm] (lr0) {};
\node [p0, below of=lr0] (lr1) {};
\node [p0, below of=lr1] (lr2) {};
\node [p0, below of=lr2] (lr3) {};
\node [p0, below of=c1, xshift=+1.5cm] (r0) {};
\node [p1, below of=r0] (r1) {};
\node [p0, below of=r1, xshift=-0.7cm] (rl0) {$R\tb$};
\node [p1, below of=rl0] (rl1) {};
\node [p0, below of=rl1, xshift=-.5cm] (xl0) {$R \tb$};
\node [p1, below of=xl0] (xl1) {};
\node [p0, below of=rl1, xshift=+.5cm] (xr0) {};
\node [p0, below of=xr0] (xr1) {$R\ta$};
\node [p0, below of=r1, xshift=+.7cm] (rr0) {};
\node [tt, below of=rr0] (rr1) {};

\path [draw=black, dotted, shorten <=3pt]
(ll1) edge ++(-70:.5cm) (ll1) edge ++(-110:.5cm)
(xl1) edge ++(-70:.5cm) (xl1) edge ++(-110:.5cm)
(lr3) edge ++(-90:.5cm)
(xr1) edge ++(-90:.5cm)
;

\path [draw=black, ->, >=stealth', shorten <=2pt, shorten >=2pt]
(c0) edge (c1)
(c1) edge (l0) (c1) edge (r0)
(l0) edge (l1)
(l1) edge (ll0) (l1) edge (lr0)
(ll0) edge (ll1)
(lr0) edge (lr1) (lr1) edge (lr2) (lr2) edge (lr3)
(r0) edge (r1)
(r1) edge (rl0) (r1) edge (rr0)
(rl0) edge (rl1)
(rl1) edge (xl0) (rl1) edge (xr0)
(xl0) edge (xl1)
(xr0) edge (xr1)
(rr0) edge (rr1)
;
\end{tikzpicture}
\quad
\vrule
\quad
\begin{tikzpicture}[font=\scriptsize, node distance=.8cm, align at top]
\node at (0,.1cm) {};
\node [p0] (c0) {$\phi(\ta)$};
\node [p1, below of=c0] (c1) {};
\node [p0, below of=c1, xshift=-1.5cm] (l0) {$R\ta$};
\node [p1, below of=l0] (l1) {};
\node [xx, below of=l1, xshift=-0.6cm] (ll0) {$\Cutsym$};
\node [p0, below of=l1, xshift=+0.6cm] (lr0) {};
\node [p0, below of=lr0] (lr1) {};
\node [p0, below of=lr1] (lr2) {};
\node [p0, below of=lr2] (lr3) {};
\node [p0, below of=c1, xshift=+1.5cm] (r0) {};
\node [p1, below of=r0] (r1) {};
\node [p0, below of=r1, xshift=-0.7cm] (rl0) {$R\tb$};
\node [p1, below of=rl0] (rl1) {};
\node [xx, below of=rl1, xshift=-.5cm] (xl0) {$\Cutsym$};
\node [p0, below of=rl1, xshift=+.5cm] (xr0) {};
\node [xx, below of=xr0] (xr1) {$\Cutsym$};
\node [p0, below of=r1, xshift=+.7cm] (rr0) {};
\node [tt, below of=rr0] (rr1) {};

\path [draw=black, dotted, shorten <=3pt]
(lr3) edge ++(-90:.5cm)
;

\path [draw=black, ->, >=stealth', shorten <=2pt, shorten >=2pt]
(c0) edge (c1)
(c1) edge (l0) (c1) edge (r0)
(l0) edge (l1)
(l1) edge (ll0) (l1) edge (lr0)
(lr0) edge (lr1) (lr1) edge (lr2) (lr2) edge (lr3)
(r0) edge (r1)
(r1) edge (rl0) (r1) edge (rr0)
(rl0) edge (rl1)
(rl1) edge (xl0) (rl1) edge (xr0)
(xr0) edge (xr1)
(rr0) edge (rr1)
;
\end{tikzpicture}
\caption{A visualization of a strategy $\Ss$ and its $(R,2)$-truncation.}
\label{fig:strategyTruncations}
\end{figure}

We lift the definition of the provenance value $\eval \lm \Ss$ to truncations $\eval \lm {\Ss \cut n}$ by treating $\Cutsym$ as an additional literal. 
That is, given some value $\pi(\Cutsym)$,
\[
    \eval \lm {\Ss \cut n} =
    \lm(\Cutsym)^{\litcount \Ss {\Cutsymsmall}}
    \;\cdot\;
    \prod_{\mathclap{L \in \Lit_A}} \lm(L)^{\litcount \Ss L}
\]

\begin{lemma}
\label{thm:truncationYieldsStrategy}
Let $\lm$ be an $\Sinf[X]$-interpretation.
\begin{enumerate}
\item Let $\phi = \lfpfml R \tx \theta \ty$ and let $\Ss \in \WinStrat{\phi(\ta)}$. If we extend $\lm$ by $\lm(\Cutsym) = 0$, then
\[
    \Sup_{n < \omega} \eval \lm {\Ss \cut n} = \eval \lm \Ss
\]

\item Let $\phi = \gfpfml R \tx \theta \ty$ and let $\Ss \in \WinStrat{\phi(\ta)}$. If we extend $\lm$ by $\lm(\Cutsym) = 1$, then
\[
    \Inf_{n < \omega} \eval \lm {\Ss \cut n} = \eval \lm \Ss
\]
\end{enumerate}
\end{lemma}
\begin{proof}
For (1), note that $\eval \lm {\Ss \cut n} = 0$ whenever $\Ss$ has a path with at least $n$ $R$-nodes
(so $\Ss \cut n$ contains a $\Cutsym$-node).
We show that there is a $k$ such that all paths of $\Ss$ have less than $k$ $R$-nodes.
Assume towards a contradiction that this is not the case.
Then consider the subgraph of $\Ss$ consisting of all nodes from which a path to an $R$-node exists.
This subgraph must then have paths of arbitrary length and by K\H onigs lemma (note that $\Ss$ is finitely branching), it must have an infinite path. This is a contradiction, as this infinite path would contain an infinite number of $R$-nodes and would thus be losing.
Hence $\eval \lm {\Ss \cut n} = \eval \lm \Ss$ for all $n \ge k$ and the claim follows.

For $(2)$, we first note that the truncations $\eval \lm {\Ss \cut n}$ indeed form a chain.
The reason is that $\lm(\Cutsym) = 1$ is the greatest element, so replacing subtrees of $\Ss$ by $\Cutsym$ leads to a larger provenance value.
Using the splitting lemma, the infimum can be written as follows (we can ignore the value $\lm(\Cutsym)$ appearing in the provenance value, as $1$ is also the neutral element).
\begin{align*}
    \Inf_{n < \omega} \eval \lm {\Ss \cut n} =
%    \Inf_{n < \omega} \prod_{L \in \Lit_A} \lm(L)^{\litcount {\Ss \cut n} L} &=
    \prod_{L \in \Lit_A} \Inf_{n < \omega} \lm(L)^{\litcount {\Ss \cut n} L} =
    \prod_{\mathclap{L \in \Lit_A}} \lm(L)^{c_L},
    \quad \text{where }
    c_L = \Sup_{n < \omega} \Litcount {\Ss \cut n} L
\end{align*}

The main observation is that each node of $\Ss$ is eventually contained in $\Ss \cut n$ (for sufficiently large $n$).
Consider a literal $L$.
If $\litcount \Ss L$ is finite, then for sufficiently large $n$, we have $\litcount {\Ss \cut n} L = \litcount \Ss L$ and thus $c_L = \litcount \Ss L$.
If $\litcount \Ss L = \infty$, then for each $k$ there is a sufficiently large $n$ such that $\litcount {\Ss \cut n} L \ge k$ and thus $c_L = \infty$, which closes the proof.
\end{proof}

\subsubsection*{The Puzzle lemma}

\begin{lemma}[Puzzle Lemma]
\label{thm:puzzlelemma}
Let $\phi = \gfpfml R \tx \theta \ty$, let $r$ be the arity of $R$ and let $\ta \in A^r$.
Let $\lm$ be an $\Sinf[X]$-interpretation extended by $\lm(\Cutsym) = 1$.
Let further $(\Ss_i)_{i < \omega}$ be a family of strategies in $\WinStrat{\phi(\ta)}$ such that $(\eval \lm {\Ss_i \cut i})_{i < \omega}$ is a descending chain.
Then there is a winning strategy $\Ss \in \WinStrat{\phi(\ta)}$ with $\eval \lm \Ss \ge \Inf_i {} \eval \lm {\Ss_i \cut i}$.
\end{lemma}

For an intuition why this result is not obvious, we consider the following example.
The key problem is that the strategies $\Ss_i$ can all be different.
In particular, it can happen that for every $i$, the provenance value of the truncation $\Ss_i \cut i$ is larger than the value of the full strategy $\Ss_i$.
The insight of the lemma is that we can always use one of the truncations $\Ss_i \cut i$ (for sufficiently large $i$) to construct a strategy $\Ss$ with the desired property.
This construction has to be done carefully to ensure that the resulting strategy $\Ss$ is winning.

\ex
Consider the following setting:

\vspace{-.9\nbs}
\begin{center}
$\displaystyle \phi_{\text{infpath}}(u) = \gfpfml R x {\E y (E x y \land R y)} u$
\qquad\qquad
\begin{tikzpicture}[baseline, font=\small, node distance=1.5cm, yshift=1pt]
\node [anchor=south, dot, label={below:$u$}] (u) {};
\node [dot, right of=u, label={below:$v$}] (w) {};
\path[draw, ->, >=stealth', shorten <=2pt, shorten >=2pt]
(u) edge node [above] {$y$} (w)
(u) edge [loop above, out=-240, in=60, looseness=15] node [above] {$x$} (u)
(w) edge [loop above, out=-240, in=60, looseness=15] node [above] {$z$} (w);
\end{tikzpicture}
\end{center}

Let $\Ss_i$ be the strategy corresponding to the infinite path that cycles $i-1$ times via $x$, then uses edge $y$ and finally cycles via $z$.
The $i$-truncation then cuts off $\Ss_i$ after taking the edge $y$ and we obtain the provenance values

\vspace{-.5\nbs}
\[
    \eval \lm {\Ss_i} = x^i y z^\infty
    \qquad
    \text{and}
    \qquad
    \Inf_{i < \omega} \eval \lm {\Ss_i \cut i} = \Inf_{i < \omega} x^i y = x^\infty y.
\]

We see that the infimum only contains the variables $x$ and $y$, although there is no winning strategy with this provenance value.
Instead, we obtain $\Ss$ by repeating the cycling part of any truncation $\Ss_i \cut i$ (without the problematic literal $y$).
This results in the strategy $\Ss$ with value $\eval \lm \Ss = x^\infty$ that corresponds to the path always cycling via $x$.
This path is not consistent with any of the strategies $\Ss_i$.
In general, we have to make sure that the additional plays in $\Ss$ (which result from the repetition of $\Ss_i \cut i$) are always winning.
\eex

As a first step to prove the Puzzle Lemma, we apply the splitting lemma to the infimum and obtain:
\[
    \Inf_{i < \omega} \eval \lm {\Ss_i \cut i} =
    \prod_{\mathclap{L \in \Lit_A}} \lm(L)^{n_L}
    \quad \text{with} \quad
    n_L = \Sup_{i < \omega} \Litcount {\Ss_i \cut i} L
\]
Literals with $n_L = \infty$ (such as the edge $x$ in the example) can appear arbitrarily often in $\Ss$, so they do not impose any restrictions.
If $n_L < \infty$, then we must have $\litcount S L \le n_L$ to guarantee that the provenance value of $\Ss$ is larger than the infimum. We therefore call literals $L$ with $n_L < \infty$ (such as the edge $y$ in the example above) \emph{problematic}.
The outline of the proof is as follows:
\begin{itemize}
\item We decompose the trees $\Ss_i \cut i$ into \emph{layers} based on the appearance of $R$-nodes.
\item We choose a sufficiently large $i$ such that there is one such layer in $\Ss_i \cut i$ which does not contain any problematic literals at all.
\item We construct $\Ss$ by first following $\Ss_i \cut i$ and then repeating this layer ad infinitum.
For the construction, we collect several subtrees (which we call \emph{puzzle pieces}) from this layer which we can then join together to form the repetition.
\item The form of the puzzle pieces ensures that $\Ss$ is winning. In particular, we only join the pieces at $R$-nodes. Paths through infinitely many pieces are thus guaranteed to satisfy the parity condition.
\end{itemize}

\smallskip\noindent{\bf Decomposition into layers.}
Fix an $i$ and let $\Ss_i \cut i = (W,F)$.
We call each node $\tau \in W$ with $\Pos(\tau) = R\ta$ (for any $\ta \in A^r$) an $R$-node.
If an $R$-node happens to be a leaf, we call it an $R$-leaf.
For each $n \ge 0$, we define the sets
\begin{align*}
    W_{\le n} = \{ \tau \in W \mid |\tau|_R \le n \}, \quad
    W_{\le n}^+ = W_{\le n} \cup \{ \tau v \in W \mid \tau \in W_{\le n}\text{, } v \in V \}
\end{align*}

We sort the nodes $\tau \in W$ into layers based on the number of $R$-nodes on the path to $\tau$.
For now, think of a layer as a forest in which all roots and most of the leaves are $R$-nodes. The $R$-leaves of one layer are the root nodes of the next layer; apart from this layers do not overlap.
The constant $k$ controls the thickness of the layer (the maximal number of $R$-nodes that can occur on paths through the layer).
For any $j \ge 1$, the \emph{$j$-th layer} is the subgraph of $\Ss_i \cut i$ induced by the node set
\[
    W_j = W^+_{\le \, j \cdot k} \setminus W_{\le (j-1) \cdot k}
    \quad \text{where} \quad
    k = |A|^r + 2.
\]

Note that each tree in a layer is a strategy (i.e., satisfies conditions (1)-(3) of \cref{def:strategyAsTree}) except for its leaves.
See \cref{fig:strategyLayerPieces} for a visualization.

\begin{figure}
\tikzset{
    p0/.style={draw=black, rounded rectangle, inner sep=1pt, minimum size=.35cm, minimum width=.8cm},
    p1/.style={draw=black, rectangle, inner sep=1pt, minimum size=.35cm, minimum width=.6cm},
    tt/.style={draw=black, rectangle, dashed, inner sep=1pt, minimum size=.35cm, minimum width=.6cm},
    mm/.style={fill=bggray},
}
\centering
\begin{tikzpicture}[font=\scriptsize, node distance=.8cm]
\node [p0] (c0) {$\phi(\ta)$};
\node [p1, below of=c0] (c1) {};
\node [p0, mm, below of=c1, xshift=-2.5cm] (l0) {$R\ta$};
\node [p1, mm, below of=l0] (l1) {};
\node [p0, mm, below of=l1, xshift=-1cm] (ll0) {$R\ta$};
\node [p1, below of=ll0] (ll1) {};
\node [p0, mm, below of=l1, xshift=+1cm] (lr0) {};
\node [p0, mm, below of=lr0] (lr1) {};
\node [p0, mm, below of=lr1] (lr2) {};
\node [p0, mm, below of=lr2] (lr3) {};
\node [p0, below of=c1, xshift=+2.5cm] (r0) {};
\node [p1, below of=r0] (r1) {};
\node [p0, mm, below of=r1, xshift=-1cm] (rl0) {$R\tb$};
\node [p1, mm, below of=rl0] (rl1) {};
\node [p0, mm, below of=rl1, xshift=-.6cm] (xl0) {$R \tb$};
\node [p1, below of=xl0] (xl1) {};
\node [p0, mm, below of=rl1, xshift=+.6cm] (xr0) {};
\node [p0, mm, below of=xr0] (xr1) {$R\ta$};
\node [p0, below of=r1, xshift=+1cm] (rr0) {};
\node [tt, below of=rr0] (rr1) {};

\path [draw=black, dotted, shorten <=3pt]
(ll1) edge ++(-60:.7cm) (ll1) edge ++(-120:.7cm)
(xl1) edge ++(-60:.7cm) (xl1) edge ++(-120:.7cm)
%(lr3) edge [densely dotted, thick, gray] node [sloped,anchor=base,xshift=1pt,yshift=2pt,font=\tiny] {inf} ++(-90:.6cm)
(xr1) edge ++(-90:.6cm)
(lr3) edge ++(-90:.6cm)
;

\path [draw=black, ->, >=stealth', shorten <=2pt, shorten >=2pt]
(c0) edge (c1)
(c1) edge (l0) (c1) edge (r0)
(l0) edge (l1)
(l1) edge (ll0) (l1) edge (lr0)
(ll0) edge (ll1)
(lr0) edge (lr1) (lr1) edge (lr2) (lr2) edge (lr3)
(r0) edge (r1)
(r1) edge (rl0) (r1) edge (rr0)
(rl0) edge (rl1)
(rl1) edge (xl0) (rl1) edge (xr0)
(xl0) edge (xl1)
(xr0) edge (xr1)
(rr0) edge (rr1)
;
\end{tikzpicture}
\caption{A visualization of a strategy. The gray nodes form the first layer (for $k=1$).
The two trees in this layer are puzzle pieces, the left one has an infinite winning path.}
\label{fig:strategyLayerPieces}
\end{figure}

\smallskip\noindent{\bf Avoiding problematic literals.}
Let $n = \sum \{n_L \mid L \in \Lit_A\text{, } n_L < \infty\}$ be the sum of the problematic $n_L$,
which is an upper bound on the number of problematic literals appearing in any truncation $\Ss_i \cut i$.
Note that $n$ is always finite. We now choose any $i$ such that:
\[
    i \;\ge\; (n+1) \cdot k = (n+1) \cdot (|A|^r + 2)
\]

From now on, we only work with $\Ss_i \cut i = (W, F)$.
Consider the layers $W_1, \dots, W_{n+1}$ of $\Ss_i \cut i$.
First assume that there is a $j$ such that $W_j = \emptyset$.
By definition of the layers, we thus have $|\tau|_R \le (j-1) \cdot k < i$ for all $\tau \in \Ss_i \cut i$.
But this means that each path in $\Ss_i \cut i$ has less than $i$ $R$-nodes.
By definition of the truncation, this means that $\Ss_i \cut i = \Ss_i$.
In this case we can simply set $\Ss = \Ss_i$ and are done.

Otherwise, there are $n+1$ nonempty layers and at most $n$ occurrences of problematic literals.
Hence there must be a layer $j$ such that $W_j$ does not contain any problematic literals.
In the following, we concentrate only on this layer $W_j$.

\smallskip\noindent{\bf Collecting puzzle pieces.}
We want to build the strategy $\Ss$ from the prefix of $\Ss_i \cut i$ up to layer $W_j$ and then continue by always repeating the layer $W_j$. Because $W_j$ does not contain any problematic literals, this eventually yields $\eval \lm \Ss \ge \eval \lm {\Ss_i \cut i}$ as required.

Let $T$ be one of the components in $W_j$, so $T$ is a tree.
We call a path in $T$ \emph{winning} if it is infinite or ends in a terminal position, so it corresponds to a (suffix of a) play consistent with $\Ss_i$.
%We call a maximal path through $T$ \emph{winning} if the corresponding play consistent with $\Ss_i$ is winning.
Paths ending in $R$-leaves of $T$ (which could be continued by leaving the layer $W_j$) are not considered to be winning.

\begin{defn}
A \emph{puzzle piece} $P = (W', F')$ is a subtree of $W_j$ such that
\begin{enum}
\item[(a)] The root of $P$ is an $R$-node,
\item[(b)] For each inner node $\tau \in P$, we have $\tau F' = \tau F$ ($P$ contains all successors),
\item[(c)] Each maximal path through $P$ is either winning or ends in an $R$-node.
\end{enum}
A puzzle piece $P$ with root $\tau$ \emph{matches} a node $\tau' \in W_j$ if $\Pos(\tau) = \Pos(\tau')$.
A \emph{complete puzzle} is a set of puzzle pieces such that for each piece in the set and all $R$-leaves $\tau$ of this piece, the set contains a puzzle piece that matches $\tau$.
\end{defn}

First observe that $T$ itself is a puzzle piece:
Each maximal path through $T$ which does not end in an $R$-node must visit less than $k$ $R$-nodes.
If we append this path to the unique path from the root of $\Ss_i \cut i$ to the root of $T$, then the resulting path contains less than $i$ $R$-nodes.
Hence the path is not truncated in $\Ss_i \cut i$, so it is also a maximal path of $\Ss_i$ and thus winning.
However, a single piece does not make a complete puzzle.
Instead, we collect smaller pieces from $T$ by the following process:

\begin{enumerate}
\item Initialize $L = \{ \hat \tau \}$ where $\hat \tau$ is the root of $T$ (which is an $R$-node).
\item Pick a node $\tau \in L$ and remove it from $L$ (if $L$ is empty, terminate).
\item If we have already found a puzzle piece matching $\tau$, go back to step (2).

\item Let $P$ be the subgraph of $T$ induced by the following set $W'$ of nodes.
Then $P$ is a puzzle piece matching $\tau$ and we add it to our set of pieces. We set
\[
W' \coloneqq \{ \tau' \in T \mid \tau \prefix \tau' \text{ and there is no $R$-node $\tau''$ with $\tau \prefixneq \tau'' \prefixneq \tau'$)} \}.
\]
\item For each $\ta \in A^r$: If $P$ has a leaf $\tau'$ with $\Pos(\tau') = R \ta$, add one such leaf $\tau'$ to $L$.
\item Go back to step (2).
\end{enumerate}

If the definition of $P$ in step (4) is correct, then this process clearly terminates after finding at most $|A|^r$ puzzle pieces and the resulting set of pieces is a complete puzzle.
For step (4), recall the definition of $W_j$.
For the root of $T$, we have $|\hat \tau|_R = (j-1)k + 1$ and $W_j$ contains in particular the nodes $\tau$ with $(j-1)k < |\tau|_R \le jk$.

Assume that in (2), we picked a node $\tau$ with $|\tau|_R = n$ (for some $n$).
By definition of $W'$, the piece $P$ only contains nodes $\tau'$ with $|\tau'|_R \le n+1$.
In particular, the leaves that we add to $L$ in step (5) all satisfy $|\tau'|_R \le n+1$.
We start with $|\hat \tau|_R = (j-1)k + 1$ and perform at most $k-2 = |A|^r$ iterations, hence we always have $|\tau|_R < jk$ for all $\tau \in L$.

This guarantees that $P$ is always a puzzle piece in step (4):
Inner nodes of $P$ cannot be $R$-nodes and hence $P$ always contains all successors of inner nodes, so (b) is satisfied.
For (c), assume towards a contradiction that there is a maximal path through $P$ which does not end in an $R$-node and is not winning.
This path is also a path in $\Ss_i$ and because all infinite paths of $\Ss_i$ are winning, the path must be finite. Because all terminal positions in $\Ss_i$ are winning, the path must end in a leaf of $T$ which is not a leaf of $\Ss_i$. But such leaves of $T$ must be $R$-nodes by definition of the layers, which is a contradiction. Hence (c) holds as well and $P$ is a puzzle piece.

We proceed in the same way for all other components of $W_j$ and obtain a complete puzzle for each component.
The overall result is the union of all these puzzles, which is again a complete puzzle.
An illustration of such a puzzle (as individual pieces and in assembled form) is shown in \cref{fig:puzzlePieces}; the next step is to perform the assembly.

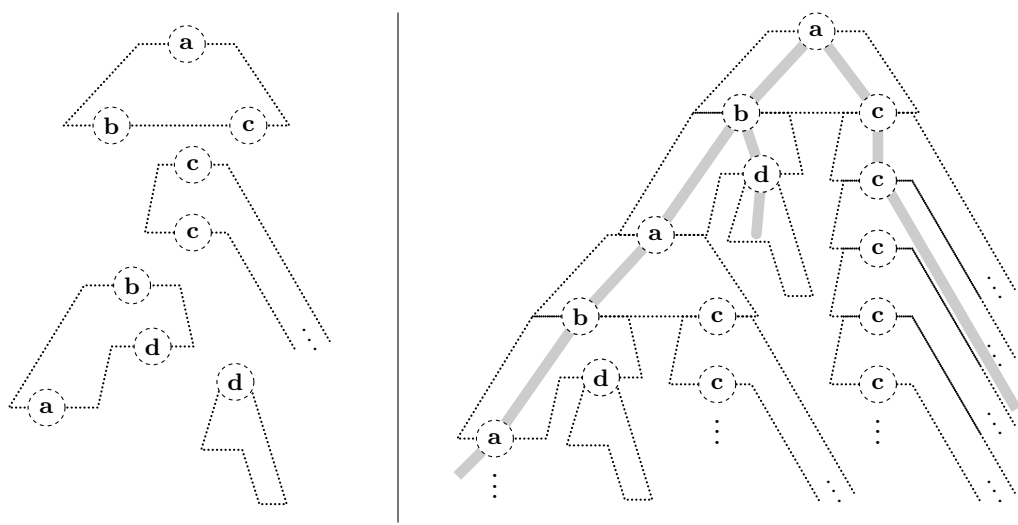
\begin{figure}
\tikzset{
    border/.style={black, line cap=round, dash pattern=on 0pt off 1.5pt, thick},
    connector/.style={black, dash pattern=on 1.5pt off 1.5pt}, %densely dashed},
    root/.pic={
        \draw[draw=none, fill=white] (0,0) circle (.275cm);
        \draw[connector] (0,0) circle (.27cm);
        \node at (0,0) {#1};
    },
    leaf/.pic={
        \draw[draw=none, fill=white] (0,0) circle (.275cm);
        \draw[connector] (0,0) circle (.27cm);
        \node at (0,0) {#1};
    },    
    pieceA/.pic={
        \draw[border] (0,0) -- (.7,0) -- (1.5,-1.2) -- (-1.8,-1.2) -- (-.7,0) -- cycle;
        \coordinate (-B) at (-1.1,-1.2);
        \coordinate (-C) at (0.9,-1.2);
        \pic at (0,0) {root=$\ta$};
        \pic at (-B) {leaf=$\tb$};
        \pic at (-C) {leaf=$\tuple c$};
    },
    pieceB/.pic={
        \draw[border] (0,0) -- (0.7,0) -- (0.9,-.9) -- (-0.3,-.9) -- (-0.5,-1.8) -- (-1.8,-1.8) -- (-.7,0) -- cycle;
        \coordinate (-D) at (0.3,-0.9);
        \coordinate (-A) at (-1.25,-1.8);
        \pic at (0,0) {root=$\tb$};
        \pic at (-D) {leaf=$\tuple d$};
        \pic at (-A) {leaf=$\ta$};
    },
    pieceC/.pic={
        \draw[border] (0,0) -- (0.5,0) -- ++(-60:3cm);
        \draw[border] (0,0) -- (-0.5,0) -- (-0.7,-1) -- (0.5,-1) -- node [sloped,yshift=.22cm,pos=1] {$\cdots$} ++(-60:2cm);
        \coordinate (-C) at (0,-1);
        \pic at (0,0) {root=$\tuple c$};
        \pic at (-C) {leaf=$\tuple c$};
    },
    pieceD/.pic={
        \draw [border] (0,0) -- (0.2,0) -- (0.75,-1.8) -- (0.35,-1.8) -- (0.1,-1) -- (-0.5,-1) -- (-0.2,0) -- cycle;
        \pic at (0,0) {root=$\tuple d$};
    },
    dots/.pic={
        \node [anchor=north,font=\Large] at (0,-.15cm) {$\vdots$};
    }
}
\centering
\begin{tikzpicture}[font=\small,scale=0.8,every pic/.style={scale=0.9},baseline,yshift=-6pt]
\pic at (0,0) {pieceA};
\pic at (0.1,-2) {pieceC};
\pic at (-.9,-4) {pieceB};
\pic at (0.8,-5.6) {pieceD};
\end{tikzpicture}
\qquad
\vrule
\qquad
\begin{tikzpicture}[font=\small,every pic/.style={scale=0.9},baseline]
%\node (spacer) at (0,.2cm) {};
\pic (A) at (0,0) {pieceA};
\pic (B) at (A-B) {pieceB};
\pic (C) at (A-C) {pieceC};
\pic (CC) at (C-C) {pieceC};
\pic (CCC) at (CC-C) {pieceC};
\pic (CCCC) at (CCC-C) {pieceC};
\pic (AA) at (B-A) {pieceA};
\pic (D) at (B-D) {pieceD};
\pic (BB) at (AA-B) {pieceB};
\pic (C2) at (AA-C) {pieceC};
\pic (DD) at (BB-D) {pieceD};
\pic at (BB-A) {dots};
\pic at (C2-C) {dots};
\pic at (CCCC-C) {dots};

\begin{scope}[on background layer]
\draw [draw=wingray, line width=4pt] (0,0) -- (A-B) -- (B-A) -- (AA-B) -- (BB-A) -- ++(-135:.7cm);
\draw [draw=wingray, line width=4pt, line cap=round] (0,0) -- (A-B) -- (B-D) -- ++(-95:.8cm);
\draw [draw=wingray, line width=4pt, line cap=butt] (0,0) -- (A-C) -- (C-C) -- ++(0.07,0) -- ++(-60:3.5cm);
\end{scope}
\end{tikzpicture}
\qquad
\caption{A schematic illustration of the pieces in a complete puzzle and their infinite repetition.
We abbreviate $R$-nodes $R \ta$ (at which we join pieces) by just $\ta$. The gray lines indicate three paths: One through infinitely many pieces, a finite one and an infinite one that only visits finitely many pieces (from left to right). All three are winning by construction.}
\label{fig:puzzlePieces}
\end{figure}

\smallskip\noindent{\bf Completing the puzzle.}
We now have a complete puzzle with a matching piece for all root nodes of $W_j$ (these are precisely the $R$-leaves of the preceding layer $W_{j-1}$).
All that remains is to join the pieces together to form the strategy $\Ss$.
Note that puzzle pieces can contain infinite paths or even infinitely many $R$-leaves.
We therefore construct $\Ss$ recursively layer by layer:
\begin{itemize}
\item $S_0$ is the subgraph induced by $W^+_{\le (j-1)k}$, i.e., the prefix of $\Ss_i \cut i$ up to layer $W_j$.
By definition of the layers, all leaves of $S_0$ are either leaves of $\Ss_i$ or $R$-leaves of $W_{j-1}$.

\item Given $S_{n}$, we construct $S_{n+1}$ as follows.
Recall that for $\tau = v_0 \dots v_l \in S_n$, we write $|\tau| = l$ for its length (which equals the depth of $\tau$ in $S_n$).
Consider the set
\[ X = \{ \tau \in S_n \mid \tau \text{ is an $R$-leaf of $S_n$ with } |\tau| = n \} \]

Because $S_n$ is finitely branching (as we construct it from subtrees of $\Ss_i \cut i$), this set is finite.
Moreover, each $\tau \in X$ is either the $R$-leaf of a puzzle piece or, initially, the root of one component of $W_j$. In both cases, the complete puzzle contains a piece matching $\tau$.
The tree $S_{n+1}$ results from $S_n$ by replacing all leaves $\tau \in X$ with the unique puzzle piece matching $\tau$. Then $S_{n+1}$ has no more $R$-leaves at depth $n$ (note that the puzzle pieces we collected always consist of at least two nodes).

A technical remark: To fit our definition of strategies, our construction must yield a subtree of the tree unraveling.
To see that this is the case we can, whenever we replace $\tau \in X$ by a piece $P$, rename the nodes of $P$ accordingly: If $\hat \tau$ is the root of $P$, we rename each node $\hat \tau \tau' \in P$ to $\tau \tau'$ when adding it to $S_{n+1}$.
Since $\tau$ and $\tau'$ are paths in the game graph, also $\tau \tau'$ is a path in the game and thus a node of the tree unraveling.

\item We define $\displaystyle \Ss = \bigcup_{n < \omega} S_n$, so $\Ss$ contains no more $R$-leaves.
\end{itemize}

Then $\Ss$ is a strategy: Each node $\tau \in \Ss$ corresponds to a node $\tau' \in \Ss_i \cut i$ (either $\tau'$ is an inner node of a puzzle piece, or $\tau' \in S_0$) and $\tau$ has the same successors as $\tau'$.
Moreover, the provenance value is $\eval \lm \Ss \ge \eval \lm {\Ss_i \cut i} \ge \Inf_i \eval \lm {\Ss_i \cut i}$ as desired, because the repetition of puzzle pieces does not contain any problematic literals.
Lastly, $\Ss$ is a winning strategy: Consider a play consistent with $\Ss$ and the corresponding maximal path through $\Ss$. If the path is finite, it ends in a leaf of $\Ss$ which corresponds to a leaf of $\Ss_i$ and is therefore winning.
If the path visits infinitely many puzzle pieces (whose root nodes are $R$-nodes), then it visits infinitely many $R$-nodes and is thus winning by the parity condition (as $R$ belongs to the outermost fixed-point formula in $\phi$).
If the path is infinite and stays in $S_0$, then it corresponds to an infinite path of $\Ss_i \cut i$ and is thus winning.
Otherwise, the path is infinite, leaves $S_0$ at some point and visits only finitely many puzzle pieces.
This means that it must from some point on stay in one piece, so it is winning by definition of puzzle pieces.
We have therefore completed the Puzzle Lemma.
\qed

\subsubsection*{The Main Proof}

\tikzset{
%    shadow/.style={drop shadow={shadow scale=1, shadow xshift=1pt, shadow yshift=-.8pt, opacity=.8, fill=gray!80}},
    p0/.style={draw=black, rounded rectangle, inner sep=2pt, minimum size=.55cm, minimum width=1cm},
    p1/.style={draw=black, rectangle, inner sep=2pt, minimum size=.50cm},
    prio/.style={font=\scriptsize},    
    tt/.style={draw=gray, fill=white, rectangle, dashed, inner ysep=2pt, inner xsep=4pt, minimum size=.55cm},
    gg/.style={draw=none, rectangle, inner sep=2pt},
}

We are now ready to prove the central result for $\Sinf[X]$:
\[
    \eval \lm {\phi(\ta)} = \Sup \{ \eval \lm \Ss \mid \Ss \in \WinStrat{\phi(\ta)} \}
\]

The interesting part is the proof for fixed-point formulae
such as $\phi(\ta) = \lfpfml R \tx \theta \ta$ or $\gfpfml R \tx \theta \ta$.
A strategy $\Ss \in \WinStrat{\phi(\ta)}$ may then look as in the picture below.
We write $L/\!\dots$ to denote either a literal or an infinite path (without occurrences of $R$-nodes).

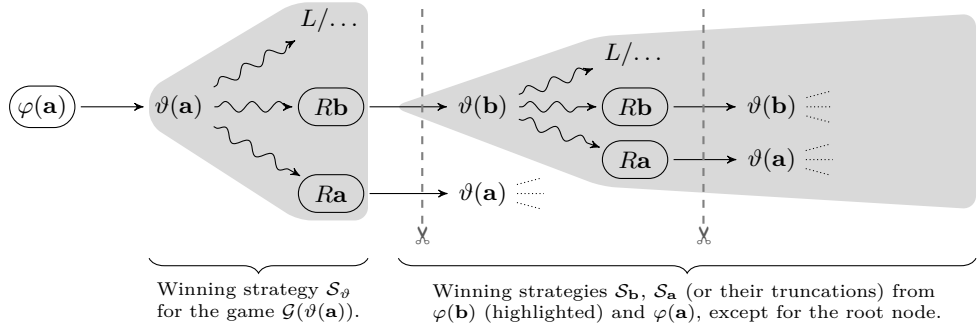
\begin{center}
\begin{tikzpicture}[font=\small]
\node [p0] (0) {$\phi(\ta)$};
\node [gg] (1) at ($(0)+(1.8cm,0)$) {$\theta(\ta)$}; 
\node [gg] (2L) at ($(1)+(2cm,1.15)$) {$L/\!\dots$}; 
\node [p0,minimum height=.5cm] (2b) at ($(1)+(2cm,0)$) {$R\tb$}; 
\node [p0,minimum height=.5cm] (2a) at ($(1)+(2cm,-1.15)$) {$R\ta$}; 

\node [gg] (3b) at ($(2b)+(2cm,0)$) {$\theta(\tb)$}; 
\node [gg] (3a) at ($(2a)+(2cm,0)$) {$\theta(\ta)$}; 

\node [gg] (4L) at ($(3b)+(2cm,0.7cm)$) {$L/\!\dots$}; 
\node [p0,minimum height=.5cm] (4b) at ($(3b)+(2cm,0)$) {$R\tb$}; 
\node [p0,minimum height=.5cm] (4a) at ($(3b)+(2cm,-0.7)$) {$R\ta$}; 

\node [gg] (5b) at ($(4b)+(1.8cm,0)$) {$\theta(\tb)$}; 
\node [gg] (5a) at ($(4a)+(1.8cm,0)$) {$\theta(\ta)$}; 

\path [->, >=stealth', shorten <=2pt, shorten >=2pt]
(0) edge (1)
(2a) edge (3a) (2b) edge (3b)
(4a) edge (5a) (4b) edge (5b);

\path [->, >=stealth', shorten <=2pt, shorten >=2pt,
       decoration = {snake,amplitude=2pt,pre length=3pt,post length=4pt},
       every edge/.append style={decorate}]
(1) edge (2L) (1) edge (2b) (1) edge (2a)
(3b) edge (4L) (3b) edge (4b) (3b) edge (4a);

\path [-, shorten <=1.5pt, densely dotted]
\foreach \x in {3a,5a,5b} {
    (\x) edge ++(15:.8cm)
    (\x) edge ++(0:.8cm)
    (\x) edge ++(-15:.8cm)
};

\coordinate (mid1) at ($(2b)!0.6!(3b)$);
\coordinate (top1) at ($(mid1)+(0,1.3cm)$);
\coordinate (bot1) at ($(mid1)+(0,-1.55cm)$);
\coordinate (bot1s) at ($(mid1)+(0,-1.7cm)$);
\draw [-, thick, gray, dashed] (bot1) -- (top1);
\node [rotate=90,yshift=-0.5pt,gray!70!black] at (bot1s) {$\Cutsym$};

\coordinate (mid2) at ($(4b)!0.5!(5b)$);
\coordinate (top2) at ($(mid2)+(0,1.3cm)$);
\coordinate (bot2) at ($(mid2)+(0,-1.55cm)$);
\coordinate (bot2s) at ($(mid2)+(0,-1.7cm)$);
\draw [-, thick, gray, dashed] (bot2) -- (top2);
\node [rotate=90,yshift=-0.5pt,gray!70!black] at (bot2s) {$\Cutsym$};

\begin{scope}[on background layer]
\coordinate (2asouth) at ($(2a.south)+(0,-.1cm)$);
\draw [draw=none, fill=bggray, rounded corners=5pt]
(1.north west) -- (2L.north west) -- (2L.north east) -- (2asouth -| 2L.east) -- (2asouth -| 2a.west) -- (1.south west) -- cycle;

\coordinate (4asouth) at ($(4a.south)+(0,-.1cm)$);
\coordinate (tr) at ($(4L.north east)+(4cm,.3cm)$);
\coordinate (br) at ($(4asouth -| 4L.east)+(4cm,-.3cm)$);
\draw [draw=none, fill=bggray, rounded corners=5pt]
($(3b)!0.6!(2b)$) -- (4L.north west) -- (tr) -- (br) -- (4asouth -| 4a.west) -- cycle;
%(3b.north west) -- (4L.north west) -- (tr) -- (br) -- (4asouth -| 4a.west) -- (3b.south west) -- cycle;
\end{scope}

\coordinate (brace1right) at (2L.east |- 2a.south);
\coordinate (brace1left) at (brace1right -| 1.west);
\draw[decoration={brace, mirror, raise=0.6cm, amplitude=.2cm}, decorate]
(brace1left) --
node[below,anchor=north,align=left,font=\scriptsize,yshift=-.8cm]
{Winning strategy $\Ss_\theta$\\for the game $\Game(\theta(\ta))$.}
(brace1right);

\coordinate (brace2right) at (2a.south -| br);
\coordinate (brace2left) at ($(brace1right)+(0.4cm,0)$); %(brace2right -| 3b.west);
\draw[decoration={brace, mirror, raise=0.6cm, amplitude=.2cm}, decorate]
(brace2left) -- 
node[below,anchor=north,align=left,font=\scriptsize,yshift=-.8cm]
{Winning strategies $\Ss_\tb$, $\Ss_\ta$ (or their truncations) from\\
$\phi(\tb)$ (highlighted) and $\phi(\ta)$, except for the root node.}
(brace2right);
\end{tikzpicture}
\end{center}

\vspace{-.5\nbs}
The strategy $\Ss$ must first move to $\theta(\ta)$ and thus contains a winning strategy from $\theta(\ta)$ in $\Game(\phi)$. If we only consider the strategy from $\theta(\ta)$ up to the first occurrence of an $R$-node, as indicated above, we obtain a winning strategy for the game $\Game(\theta(\ta))$.
Note that in $\Game(\theta(\ta))$, $R$-nodes are terminals and are thus winning.

In $\Game(\phi(\ta))$, these $R$-nodes are not terminals.
Hence $\Ss$ must further contain substrategies for these $R$-nodes ($R \tb$ and $R \ta$ above).
Because the positions $\phi(\tb)$ and $R \tb$ must both be followed by $\theta(\tb)$, we can view the substrategy from $R\tb$ as a winning strategy $\Ss_\tb$ for $\Game(\phi(\tb))$ (as highlighted above).

We thus see that each winning strategy $\Ss$ in $\Game(\phi(\ta))$ can be decomposed into a prefix $\Ss_\theta$ which we can identify with a winning strategy for $\Game(\theta(\ta))$ and, for all $R$-leaves of $\Ss_\theta$, substrategies which are winning strategies in $\Game(\phi)$.
Conversely, every winning strategy $\Ss_\theta$ for $\Game(\theta(\ta))$ can be combined with winning strategies from $\phi(\ta), \phi(\tb), \dots$ for all the $R$-leaves of $\Ss_\theta$ to form a winning strategy in $\Game(\phi(\ta))$.

If we build the strategy by starting with $\Ss_\theta$ but then appending the truncations $\Ss_\ta \cut n$ and $\Ss_\tb \cut n$ instead of $\Ss_\ta, \Ss_\tb$ (as indicated by the dashed lines in the picture),
then the result is the $n+1$-truncation $\Ss \cut {n+1}$ of a winning strategy $\Ss \in \WinStrat{\phi(\ta)}$,
because $\Ss_\theta$ contains at most one $R$-node on each path.
We exploit this observation in an inductive proof that relates the $n$-truncations of winning strategies with the $n$-th step of the fixed-point iteration.

\begin{proof}[Proof of \cref{thmStrategyCharacterization} in {$\Sinf[X]$}]
Induction on the negation normal form of $\phi(\ta)$:
\begin{itemize}%[leftmargin=1.5em]
\vspace{.25\nbs}
\item $\phi(\ta) = L \in \Lit$: Then $\Game(\phi(\ta))$ consists only of a terminal position (which is winning). There is only one (trivial) strategy with $\eval \lm \Ss = \lm(L) = \eval \lm {\phi(\ta)}$.

\vspace{.5\nbs}
\item \begin{minipage}[t]{.72\linewidth}
$\phi(\ta) = \phi_1(\ta) \lor \phi_2(\ta)$.
The game $\Game(\phi(\ta))$ is shown on the right.
Each strategy $\Ss$ for $\Game(\phi(\ta))$ makes a unique choice at $\phi(\ta)$ and thus either consists of a strategy $\Ss_1$ for $\Game(\phi_1(\ta))$ or a strategy $\Ss_2$ for $\Game(\phi_2(\ta))$, but not both.
Conversely, a strategy $\Ss_i$ for $\Game(\phi_i(\ta))$ lifts to a strategy $\Ss$ from $\phi(\ta)$ (for $i \in \{0,1\}$). We thus have:
\end{minipage}
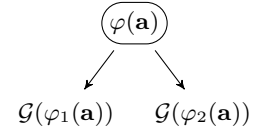
\hfill
\begin{minipage}[t]{.23\linewidth}
\centering
\raisebox{-1.1\height}{\begin{tikzpicture}[node distance=1.2cm, font=\small]
\node [p0] (0) {$\phi(\ta)$};
\node [gg,below of=0, xshift=-.9cm] (1) {$\Game(\phi_1(\ta))$};
\node [gg,below of=0, xshift=+.9cm] (2) {$\Game(\phi_2(\ta))$};
\path [->, >=stealth', shorten <=3pt, shorten >=3pt] (0) edge (1) (0) edge (2);
\end{tikzpicture}}
\end{minipage}
\begin{align*}
\Sup \{ \eval \lm \Ss \mid \Ss \in \WinStrat{\phi(\ta)} \} &=
\Sup \{ \eval \lm {\Ss_i} \mid \Ss_i \in \WinStrat{\phi_i(\ta)} \text{, } i \in \{0,1\} \} \\ &=
\Sup \{ \eval \lm {\Ss_1} \mid \Ss_1 \in \WinStrat{\phi_1(\ta)} \} \; \join \; 
\Sup \{ \eval \lm {\Ss_2} \mid \Ss_2 \in \WinStrat{\phi_2(\ta)} \} \\ &\eqIH
\eval \lm {\phi_1(\ta)} \join \eval \lm {\phi_2(\ta)} =
\eval \lm {\phi_1(\ta)} + \eval \lm {\phi_2(\ta)} =\eval \lm {\phi(\ta)}
\end{align*}

\vspace{.3\nbs}
\item \begin{minipage}[t][][t]{.72\linewidth}
$\phi(\ta) = \phi_1(\ta) \land \phi_2(\ta)$.
The reasoning is similar:
Each strategy $\Ss$ for $\Game(\phi(\ta))$ consists of both a strategy $\Ss_1$ for $\Game(\phi_1(\ta))$ and a strategy $\Ss_2$ for $\Game(\phi_2(\ta))$.
The converse direction (all $\Ss_1$ and $\Ss_2$ together induce a strategy $\Ss$) holds as well.
\end{minipage}
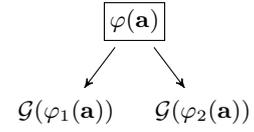
\hfill
\begin{minipage}[t][][t]{.23\linewidth}
\centering
\raisebox{-0.9\height}{\begin{tikzpicture}[node distance=1.2cm, font=\small]
\node [p1] (0) {$\phi(\ta)$};
\node [gg,below of=0, xshift=-.9cm] (1) {$\Game(\phi_1(\ta))$};
\node [gg,below of=0, xshift=+.9cm] (2) {$\Game(\phi_2(\ta))$};
\path [->, >=stealth', shorten <=3pt, shorten >=3pt] (0) edge (1) (0) edge (2);
\end{tikzpicture}}
\end{minipage}

\vspace{\nbs}
If $\Ss$ consists of the two strategies $\Ss_1$ and $\Ss_2$, then we further have $\eval \lm \Ss = \eval \lm {\Ss_1} \cdot \eval \lm {\Ss_2}$ by definition of the provenance value and \cref{lemInfpow}.
The claim follows by induction and continuity%
\footnote{Notice that here we consider suprema of arbitrary sets.
It follows from the considerations in \cite{Markowsky76}
that if multiplication preserves suprema of finite sets and of chains, it also preserves suprema of arbitrary sets.
In idempotent semirings, finite suprema are simply finite sums which are preserved by distributivitiy.}
of $\Sinf[X]$:
\begin{align*}
\eval \lm {\phi_1(\ta)} \cdot \eval \lm {\phi_1(\ta)} &\eqIH
\Sup \{ \eval \lm {\Ss_1} \mid \Ss_1 \in \WinStrat{\phi_1(\ta)} \} \cdot \Sup \{ \eval \lm {\Ss_2} \mid \Ss_2 \in \WinStrat{\phi_2(\ta)} \} \\
&= \Sup \{ \eval \lm {\Ss_1} \cdot \eval \lm {\Ss_2} \mid \Ss_i \in \WinStrat{\phi_i(\ta)}\text{ for }i \in \{1,2\} \} \\
&= \Sup \{ \eval \lm \Ss \mid \Ss \in \WinStrat{\phi(\ta)} \}
\end{align*}

\item
The cases for $\E$ and $\A$ follow by the same arguments (with $|A|$ instead of 2 child nodes).

\item
For fixed-point formulae, we use the decomposition into $\Ss_\theta$ and $\Ss_\ta, \Ss_\tb$ as motivated above. Consider the fixed-point iteration $(f_\beta)_{\beta \in \On}$ for $\phi = \lfpfml R \tx \theta \ty$ or $\phi = \gfpfml R \tx \theta \ty$, where $R$ has arity $r$.
We relate this iteration to strategy truncations.
To this end, we set $\lm(\Cutsym) = 0$ in case of $\phi = \lfpfml R \tx \theta \ty$ and $\lm(\Cutsym) = 1$ for $\phi = \gfpfml R \tx \theta \ty$.
We split the remaining proof into two claims:
\begin{claim*}[1]
  For all $n < \omega$ and $\ta \in A^r$, we have
  $\displaystyle f_n(\ta) = \Sup \{ \eval \lm {\Ss \cut n} \mid \Ss \in \WinStrat{\phi(\ta)} \}$.
\end{claim*}
\begin{claim*}[2]
  For all $\ta \in A^r$, we further have
  $\displaystyle f_\omega(\ta) = \Sup \{ \eval \lm \Ss \mid \Ss \in \WinStrat{\phi(\ta)} \}$.
\end{claim*}
We have already shown that the fixed-point iteration has closure ordinal $\omega$,
hence $\pi \ext {\phi(\ta)} = f_\omega(\ta)$ and the claims suffice to close the proof.
We prove both claims below.
\end{itemize}

\begin{claimproof}[Proof of Claim (1)]
For $n=0$, we trivially have $f_0 = 0$ and $\eval \lm {\Ss \cut 0} = 0$ for least fixed points and $f_0 = 1$ and $\eval \lm {\Ss \cut 0} = 1$ for greatest fixed points. For the induction step $n \to n+1$, we first rewrite $f_{n+1}$.
To simplify notation, we set $\Lit_A^* = \Lit_A \setminus \{ R \ta \mid \ta \in A^r \}$.
By the induction hypothesis on $\theta$ and on $f_n$, we can write
\begin{align*}
&f_{n+1}(\ta) =
\eval {\lm[R/f_n]} {\theta(\ta)} \\
={} &\Sup \big\{ \eval {\lm[R/f_n]} {\Ss_\theta} \bigmid \Ss_\theta \in \WinStrat{\theta(\ta)} \big\} \\
={} &\Sup \Big\{
    \prod_{L \in \Lit_A^*} \lm(L)^{n_L} \;\cdot
    \prod_{\tb \in A^r}
    \underbrace{\big( \Sup \{ \eval \lm {\Ss_\tb \cut n} \mid \Ss_\tb \in \WinStrat{\phi(\tb)} \} \big)^{n_\tb}}_{(*)}
\Bigmid \Ss_\theta \in \WinStrat{\theta(\ta)} \Big\}, \tag{$\dagger$}
\end{align*}
where we set $n_L = \litcount {\Ss_\theta} L$ and $n_\tb = \litcount {\Ss_\theta} {R \tb}$.
Notice that both values depend on $\Ss_\theta$; we omit this dependence in the notation to simplify the presentation.
Let us first fix a strategy $\Ss_\theta$ and $\tb \in A^r$ and consider the term $(*)$.
Recall that absorptive polynomials are finite and that for a set $S$, $\Sup S = \mxls(\bigcup S)$.
We can thus write
\[
    (*) = \big( \eval \lm {\Ss_\tb^1 \cut n} + \dots + \eval \lm {\Ss_\tb^k \cut n} \big)^{n_\tb}
\]
for some $\Ss_\tb^1, \dots, \Ss_\tb^k \in \WinStrat{\phi(\tb)}$.
If $n_\tb = \infty$, then by \cref{lemInfpow},
\[
    (*) = {\eval \lm {\Ss_\tb^1 \cut n}}^\infty + \dots + {\eval \lm {\Ss_\tb^k \cut n}}^\infty.
\]
Otherwise, $n_\tb = l < \infty$. Then each monomial of $(*)$ is of the form
\[
    \eval \lm {\Ss_\tb^{i_1} \cut n} \cdots \eval \lm {\Ss_\tb^{i_l} \cut n},
    \quad
    \text{where $i_1, \dots, i_l \in \{1, \dots, k\}$}.
\]

\medskip\noindent{\bf Direction 1:} $f_{n+1}(\ta) \le \Sup \{ \eval \lm {\Ss \cut {n+1}} \mid \Ss \in \WinStrat{\phi(\ta)} \}$.
To prove this direction, we show that each monomial of $f_{n+1}(\ta)$ is absorbed by the right-hand side.
Given the above considerations, we know that these monomials are of the form
\begin{align*}
m = \prod_{\mathclap{L \in \Lit_A^*}} \lm(L)^{n_L'}, \quad 
n_L' = n_L \,+
\sum_{\substack{\tb \in A^r \\ n_\tb = \infty}} \litcount {\Ss_\tb \cut n} {L} \cdot \infty \,+
\sum_{\substack{\tb \in A^r \\ \mathclap{n_\tb = l < \infty}}}
\litcount {\Ss_\tb^1 \cut n} {L} \cdots \litcount {\Ss_\tb^l \cut n} {L},
\end{align*}
for some $\Ss_\theta \in \WinStrat{\theta(\ta)}$ (which defines $n_L$) and some $\Ss_\tb$, $\Ss_\tb^i \in \WinStrat{\phi(\tb)}$ (for all $\tb$, $i$).

Now consider the strategy which starts with $\Ss_\theta$ and for each $\tb \in A^r$, we replace all $R\tb$-leaves of $\Ss_\theta$ by either $\Ss_\tb$ (if there are infinitely many such leaves) or by the strategies $\Ss_\tb^1, \dots, \Ss_\tb^l$ if there are $l < \infty$ such leaves (it does not matter in which order these strategies are assigned to the leaves).
The result is a strategy $\Ss \in \WinStrat{\phi(\ta)}$.

We further see that $\Ss \cut {n+1}$ results from $\Ss_\theta$ in the same way if we replace leaves by the truncations $\Ss_\tb \cut n$ instead of $\Ss_\tb$ (and $\Ss_\tb^i \cut n$ instead of $\Ss_\tb^i$), because $\Ss_\theta$ contains $R$-nodes only as leaves.
By this construction, we see that $\litcount \Ss L = n_L'$ for each $L \in \Lit_A^*$ and hence $m = \eval \lm {\Ss \cut {n+1}}$ and thus $m \le \Sup \{ \eval \lm {\Ss \cut {n+1}} \mid \Ss \in \WinStrat{\phi(\ta)} \}$ as claimed.

\medskip\noindent{\bf Direction 2:} $f_{n+1}(\ta) \ge \Sup \{ \eval \lm {\Ss \cut {n+1}} \mid \Ss \in \WinStrat{\phi(\ta)} \}$.
For this direction, we fix any strategy $\Ss \in \WinStrat{\phi(\ta)}$ and show that $\eval \lm {\Ss \cut {n+1}} \le f_{n+1}(\ta)$.
In the case that $\lm(\Cutsym) = 0$ and $\Cutsym$ appears in $\Ss \cut {n+1}$, we have $\eval \lm {\Ss \cut {n+1}} = 0$ and there is nothing to show. If $\lm(\Cutsym) = 1$, the appearance of $\Cutsym$ does not affect the provenance value $\eval \lm {\Ss \cut {n+1}}$.
We again decompose $\Ss$ into a prefix $\Ss_\theta$ corresponding to a winning strategy in $\Game(\theta(\ta))$ and substrategies from all $R$-leaves of $\Ss_\theta$.

We first consider any $\tb \in A^r$ for which $\litcount {\Ss_\theta} {R\tb} = \infty$ such that we have infinitely many such substrategies from $R\tb$-leaves. We need a preliminary observation:

\begin{observation}
There is a strategy $\Ss'_\tb \in \WinStrat{\phi(\tb)}$ such that the strategy $\Ss'$ which is like $\Ss$ but uses $\Ss'_\tb$ for all of the infinitely many $R\tb$-leaves of $\Ss_\theta$ satisfies $\eval \lm {\Ss \cut {n+1}} \le \eval \lm {\Ss' \cut {n+1}}$.
\end{observation}

\begin{claimproof}
Let $(\Ss_\tb^i)_{i < \omega}$ be the family of all substrategies $\Ss$ uses from $R\tb$-leaves of $\Ss_\theta$.
As in the proof of the Puzzle Lemma, we call a literal $L$ problematic if $\Litcount {\Ss \cut {n+1}} L < \infty$.
Let $L$ be a problematic literal. Then there can only be finitely many $i$ with $\Litcount {\Ss_\tb^i \cut n} L > 0$ (otherwise $L$ would occur infinitely often in $\Ss \cut {n+1}$).
As $\omega$ is infinite while the number of literals is finite, there is an $i < \omega$ such that $\Ss_\tb^i \cut n$ contains no problematic literals at all.
We then set $\Ss'_\tb = \Ss_\tb^i$ and the claim follows.
\end{claimproof}

Due to this observation and because $A^r$ is finite, we obtain a strategy $\Ss'$ with $\eval \lm {\Ss \cut {n+1}} \le \eval \lm {\Ss' \cut {n+1}}$ such that for all $\tb \in A^r$ with $\litcount {\Ss_\theta} {R \tb} = \infty$, the strategy $\Ss'$ uses the same substrategy from all $R\tb$-leaves of $\Ss_\theta$.
From $(\dagger)$, we know that
\begin{align*}
f_{n+1}(\ta) \ge
\prod_{L \in \Lit_A^*} \lm(L)^{n_L} \;\cdot
\prod_{\tb \in A^r}
\big(
\underbrace{
    \Sup \{ \eval \lm {\Ss_\tb \cut n} \mid \Ss_\tb \in \WinStrat{\phi(\tb)} \}}
_{\eqqcolon P_\tb}
\big)^{n_\tb},
\end{align*}
where $n_L = \litcount {\Ss_\theta} L$ and $n_\tb = \litcount {\Ss_\theta} {R \tb}$.
Consider the strategies $\Ss'$ uses from the $R$-leaves of $\Ss_\theta$:
For $\tb \in A^r$ with $\litcount {\Ss_\theta} {R\tb} = \infty$, let $\Ss_\tb$ be the strategy that $\Ss'$ uses from all $R\tb$-leaves. For $\tb$ with $\litcount {\Ss_\theta} {R\tb} = l < \infty$, let $\Ss_\tb^1, \dots, \Ss_\tb^l$ be the strategies $\Ss'$ uses from the $R\tb$-leaves of $\Ss_\theta$.
Let further
\[
n_L' = n_L +
\sum_{\substack{\tb \in A^r \\ n_\tb = \infty}} \litcount {\Ss_\tb \cut n} {L} \cdot \infty \,+
\sum_{\substack{\tb \in A^r \\ \mathclap{n_\tb = l < \infty}}}
\litcount {\Ss_\tb^1 \cut n} {L} \cdots \litcount {\Ss_\tb^l \cut n} {L}.
\]
We can apply commutativity and the lemma on infinitary powers to conclude
\begin{align*}
\eval \lm {\Ss' \cut {n+1}} =
\prod_{\mathclap{L \in \Lit_A^*}} \lm(L)^{n_L'} &=
\prod_{\mathclap{L \in \Lit_A^*}} \lm(L)^{n_L} \cdot
    \prod_{\substack{\tb \in A^r \\ \mathclap{n_\tb = \infty}}}
        {\eval \lm {\Ss_\tb \cut n}}^\infty
\cdot
    \prod_{\substack{\tb \in A^r \\ \mathclap{n_\tb = l < \infty}}}
        {\eval \lm {\Ss_\tb^1 \cut n}} \cdots {\eval \lm {\Ss_\tb^l \cut n}} \nsmall &\le
\prod_{\mathclap{L \in \Lit_A^*}} \lm(L)^{n_L} \cdot
    \prod_{\substack{\tb \in A^r \\ \mathclap{n_\tb = \infty}}}
        \mathrlap{P_\tb^\infty}
        \phantom{{\eval \lm {\Ss_\tb \cut n}}^\infty}
\cdot
    \prod_{\substack{\tb \in A^r \\ \mathclap{n_\tb = l < \infty}}}
        P_\tb^l
\quad \le \; f_{n+1}(\ta).
\end{align*}
We have thus shown both directions and \textsf{Claim (1)} follows.
\end{claimproof}

\begin{claimproof}[Proof of Claim (2)]
We now prove, assuming \textsf{Claim (1)}, that
\[
    f_\omega(\ta) = \Sup \{ \eval \lm \Ss \mid \Ss \in \WinStrat{\phi(\ta)} \}.
\]

\begin{itemize}%[leftmargin=1.5em]
\item For $\phi = \lfpfml R \tx \theta \ty$, this follows via \cref{thm:truncationYieldsStrategy} by swapping suprema:
\begin{align*}
f_\omega(\ta) =
\Sup_{n < \omega} f_n(\ta) &=
\Sup_{n < \omega} \Sup \{ \eval \lm {\Ss \cut n} \mid \Ss \in \WinStrat{\phi(\ta)} \} \\ &=
\Sup \big\{ \Sup_{n < \omega} \eval \lm {\Ss \cut n} \bigmid \Ss \in \WinStrat{\phi(\ta)} \big\} =
\Sup \{ \eval \lm \Ss \mid \Ss \in \WinStrat{\phi(\ta)} \}.
\end{align*}

\item For $\phi = \gfpfml R \tx \theta \ty$, the proof is more difficult and requires the Puzzle Lemma.
We first note that one direction is trivial (using \cref{thm:truncationYieldsStrategy} in the last step):
\begin{align*}
    \Inf_{n<\omega} \Sup \{ \eval \lm {\Ss \cut n} \mid \Ss \in \WinStrat{\phi(\ta)} \} &\ge
    \Sup \{ \Inf_{n < \omega} \eval \lm {\Ss \cut n} \mid \Ss \in \WinStrat{\phi(\ta)} \} \\ &=
    \Sup \{ \eval \lm \Ss \mid \Ss \in \WinStrat{\phi(\ta)} \}.
\end{align*}

For the other direction, we use the characterization of infima in \cref{propCharacterizationInfima}:
\[
    f_\omega(\ta) =
    \Inf_{n<\omega} \underbrace{
        \Sup \{ \eval \lm {\Ss \cut n} \mid \Ss \in \WinStrat{\phi(\ta)} \}
    }_{\eqqcolon P_n} =
    \Sup \big\{ \Inf_{n < \omega} m_n \bigmid (m_n)_{n < \omega} \in \MM \big\}
\]
where $\MM$ is defined as in \cref{propCharacterizationInfima}, based on the chain $(P_n)_{n < \omega}$.
Consider a monomial chain $\mathbf m = (m_n)_{n < \omega} \in \MM$.
By definition of $P_n$, we have for each $n$ a strategy $\Ss_n \in \WinStrat{\phi(\ta)}$ such that $m_n = \eval \lm {\Ss_n \cut n}$.
The Puzzle Lemma further implies that there is a strategy
$\Ss_{\mathbf m} \in \WinStrat{\phi(\ta)}$ such that $\eval \lm {\Ss_{\mathbf m}} \ge \Inf_n m_n$.
Using these strategies, we get
\begin{align*}
    f_\omega(\ta) =
    \Sup \big\{ \Inf_{n < \omega} m_n \bigmid \mathbf m = (m_n)_{n < \omega} \in \MM \big\} &\le
    \Sup \{ \eval \lm {\Ss_{\mathbf m}} \mid \mathbf m \in \MM \} \\ &\le
    \Sup \{ \eval \lm \Ss \mid \Ss \in \WinStrat{\phi(\ta)} \}. \qedhere
\end{align*}
\end{itemize}
\end{claimproof}

\noindent
We have now established both claims.
This closes the proof of \Cref{thmStrategyCharacterization} for $\Sinf[X]$. \qedhere
\end{proof}

\subsubsection*{Generalization}

We generalize the result from $\Sinf[X]$ to all absorptive, fully continuous semirings.
To this end, we first need two lemmas on the properties of fully continuous homomorphisms.

\begin{lemma}
Let $h : \Sinf[X] \to K$ be a semiring homomorphism (not necessarily fully continuous)
into an absorptive, fully continuous semiring $K$.
Then $\Sup h(S) = h\big(\Sup S \big)$ holds for arbitrary sets $S \subseteq \Sinf[X]$.
\end{lemma}
\begin{proof}
We first remark that the natural order on $K$ forms a complete lattice
(by \cref{propIdempotentLattice}), so the supremum $\Sup h(S)$ is well defined.
Since $h$ preserves addition and thus the natural order,
the direction $h(\Sup S) \ge \Sup h(S)$ is trivial.
Let $\Sup S = m_1 + \dots + m_k$ for some monomials $m_1, \dots, m_k$ and
consider one monomial $m_i$. Since $\Sup S = \mxls\big(\bigcup S\big)$,
there is a $P \in S$ with $m_i \in P$. Then $\Sup h(S) \ge h(P) \ge h(m_i)$.
This holds for each $1 \le i \le k$ and implies
$\Sup h(S) \ge h(m_1) + \dots + h(m_k) = h(m_1 + \dots + m_k) = h(\Sup S)$.
\end{proof}

\begin{lemma}
Let $K_1$ and $K_2$ be absorptive, fully continuous semirings and
let $h : K_1 \to K_2$ be a fully continuous semiring homomorphism.
Let further $\lm$ be a $K_1$-interpretation, $\phi \in \LFP$ and $\Ss$ a strategy in $\Game(\phi)$.
Then $h( \eval \lm \Ss ) =  (h \circ \lm) \ext \Ss$.
\end{lemma}
\begin{proof}
Clearly, $h$ preserves finite products.
It thus suffices to prove that $h$ preserves infinitary powers.
Let $a \in K_1$. Then
$h(a^\infty) = h(\Inf_n a^n) = \Inf_n h(a^n) = \Inf_n h(a)^n = h(a)^\infty$.
\end{proof}

We can now finally prove the main result in its general formulation
by considering the most general $\Sinf[X]$-interpretation
together with the above lemmas.

\begin{proof}[Proof of \cref{thmStrategyCharacterization}]
Consider the most general $\Sinf[X]$-interpretation $\mgpi$
with $\mgpi(L) = x_L$ and $X = \{ x_L \mid L \in \Atoms_A(\tau) \cup \NegAtoms_A(\tau) \}$.
By the universal property, there is a fully continuous homomorphism
$h : \Sinf[X] \to K$ with $\lm = h \circ \mgpi$ (induced by the mapping $x_L \mapsto \pi(L)$).
Then,
\begin{align*}
    \eval \lm {\phi(\ta)} =
    h(\eval {\mgpi} {\phi(\ta)}) &=
    h\big(\Sup \{ \eval {\mgpi} \Ss \mid \Ss \in \WinStrat{\phi(\ta)} \}\big) \\ &=
    \Sup \{ h(\eval {\mgpi} \Ss) \mid \Ss \in \WinStrat{\phi(\ta)} \} =
    \Sup \{ \eval \lm \Ss \mid \Ss \in \WinStrat{\phi(\ta)} \}. \qedhere
\end{align*}
\end{proof}

\section{Related Work}

While our approach and our general project, as outlined in the introduction, is
rooted in the work on semiring provenance in databases, there have also
been a number of other areas of logic in computer science where semiring semantics 
have been used.

A prominent instance is the work on weighted automata 
(see, e.g., the Handbook  \cite{HandbookWeighted2009}).
In particular, weighted automata
over finite and infinite words are discussed in \cite{DrosteGas09, DrosteRah06},
and their expressive power is related  to
weighted monadic second-order logic (MSO) on words.
In this setting, the weight of a word is defined as the sum over the weights
of accepting paths, and the overall behaviour of an automaton is described by
a formal power series over a semiring.
To deal with infinite words, infinite sum and product operations
of the semiring are assumed \cite{DrosteRah06},
roughly comparable to our assumptions on suprema and infima.
Whereas the power series assign semiring values to \emph{words},
we instead use indeterminates to track (combinations of) \emph{literals}
which then provide us with provenance information.
%or, in the analysis of games \cite{GraedelTan20}, positions or edges.
As we have seen, formal power series are not the right tool for this purpose
when confronted with greatest fixed points, so we consider
absorptive polynomials $\Sinf[X]$ instead.
Moreover, in our setting the sum-of-strategies characterization
is not a definition, but a non-trivial result.
The definition of weighted MSO is similar to our semiring semantics for LFP
and is also based on negation normal form.
Main differences are that only logics over words are considered,
and that semiring values are part of the formulae,
whereas we assign values to literals.
This reflects the different point of view: Weighted MSO is used to define
series recognizable by weighted automata, whereas our goal is the
provenance analysis of the logic itself.

Lluch-Lafuente and Montanari \cite{LluchMon05}  have studied a semantics of CTL and 
$\mu$-calculus in so-called constraint semirings, to reason about issues of quality of service
such as delay or bandwidth. The choice of constraint semirings is motivated by applications for a particular
class of constraint satisfaction problems, called soft CSP, and by useful closure properties, such as closure under
Cartesian products, exponentials, and power constructions.  
Although this is not mentioned explicitely, constraint semirings are in
fact also absorptive and satisfy a continuity requirement for suprema. However, the approach to
negation is different from ours, requiring the extension of the semiring by new functions,
and they do not have an abstract approach on the basis of polynomials with universal properties 
and reasoning over multiple constraint semirings. 
A main result of \cite{LluchMon05} is that the usual embedding of CTL into the $\mu$-calculus
fails for this semantics, which is another instance showing that
a refined semiring semantics may distinguish between formulae that equivalent under 
Boolean semantics.

\section{Conclusion and Outlook}

Let us summarize the contributions of this paper: 
We have layed foundations for the semiring provenance analysis of full fixed-point logics, with arbitrary interleavings of least and 
greatest fixed points, as part of the general project of developing provenance semantics of logical languages used in various branches
of computer science.  We have seen that absorptive and fully continuous semirings provide an adequate framework for this.
We have identified the semiring % $\Sinf[X]$ 
of dual-indeterminate generalized absorptive polynomials $\Sinf[X,\nnX]$
as the `right'
%most general adequate
provenance semiring for LFP. It satisfies
the further algebraic property of chain-positivity which
guarantees that provenance interpretations are truth-preserving, and we
have established an important universal property of this semiring.
Finally, we have shown how provenance for LFP is related to strategies in model-checking games.

Next steps will include the specific analysis of important logics such as temporal logics, 
dynamic logics, the modal  $\mu$-calculus, description logics (see initial work in \cite{DannertGra19b}) etc.
Applications require in particular the study of \emph{algorithms} for computing provenance values -- a non-trivial task,
considering that greatest fixed-point iterations in semirings such as $\Sinf[X,\nnX]$ can be infinite.
Nevertheless, absorption and the infinitary power $a^\infty$ can be used to short-circuit these iterations;
forthcoming work will include results that show how an effective, and in important cases also efficient, computation of provenance values
is possible in absorptive, fully continuous semirings.

\bibliography{provenance}

\begin{thebibliography}{10}

\bibitem{AmsterdamerDeuTan11}
Y.~Amsterdamer, D.~Deutch, and V.~Tannen.
\newblock On the limitations of provenance for queries with difference.
\newblock In {\em 3rd Workshop on the Theory and Practice of Provenance,
  TaPP'11}, 2011.
\newblock See also CoRR abs/1105.2255.

\bibitem{AptGraedel11}
K.~Apt and E.~Gr{\"a}del, editors.
\newblock {\em Lectures in Game Theory for Computer Scientists}.
\newblock Cambridge University Press, 2011.

\bibitem{DannertGra19a}
K.~Dannert and E.~Gr\"{a}del.
\newblock Provenance analysis: A perspective for description logics?
\newblock In C.~Lutz et~al., editor, {\em Description Logic, Theory
  Combination, and All That}, Lecture Notes in Computer Science Nr. 11560.
  Springer, 2019.

\bibitem{DannertGra19b}
K.~Dannert and E.~Gr\"{a}del.
\newblock Semiring provenance for guarded logics.
\newblock In {\em Hajnal Andréka and István Németi on Unity of Science: From
  Computing to Relativity Theory through Algebraic Logic}, Outstanding
  Contribution to Logic. Springer, 2020.

\bibitem{DeutchMilRoyTan14}
D.~Deutch, T.~Milo, S.~Roy, and V.~Tannen.
\newblock Circuits for datalog provenance.
\newblock In {\em Proc. 17th International Conference on Database Theory ICDT},
  pages 201--212, 2014.

\bibitem{DrosteGas09}
M.~Droste and P.~Gastin.
\newblock Weighted automata and weighted logics.
\newblock In {\em Handbook of weighted automata}, pages 175--211. Springer,
  2009.

\bibitem{DrosteKuich09}
M.~Droste and W.~Kuich.
\newblock Semirings and formal power series.
\newblock In {\em Handbook of Weighted Automata}, pages 3--28. Springer, 2009.

\bibitem{HandbookWeighted2009}
M.~Droste, W.~Kuich, and H.~Vogler, editors.
\newblock {\em Handbook of Weighted Automata}.
\newblock Springer, 2009.

\bibitem{DrosteRah06}
M.~Droste and G.~Rahonis.
\newblock Weighted automata and weighted logics on infinite words.
\newblock In {\em International Conference on Developments in Language Theory},
  pages 49--58. Springer, 2006.

\bibitem{GeertsPog10}
F.~Geerts and A.~Poggi.
\newblock On database query languages for {K-relations}.
\newblock {\em J. Applied Logic}, 8(2):173--185, 2010.

\bibitem{GeertsUngKarFunChr16}
F.~Geerts, T.~Unger, G.~Karvounarakis, I.~Fundulaki, and V.~Christophides.
\newblock Algebraic structures for capturing the provenance of {SPARQL}
  queries.
\newblock {\em J. {ACM}}, 63(1):7:1--7:63, 2016.

\bibitem{GraedelTan17}
E.~Gr{\"a}del and V.~Tannen.
\newblock Semiring provenance for first-order model checking.
\newblock arXiv:1712.01980 [cs.LO], 2017.
\newblock URL: \url{https://arxiv.org/abs/1712.01980}.

\bibitem{GraedelTan20}
E.~Gr{\"a}del and V.~Tannen.
\newblock Provenance analysis for logic and games.
\newblock {\em Moscow Journal of Combinatorics and Number Theory}, 2020.
\newblock Special Issue on Complexity and Model Theory. To appear. See also
  arXiv: 1907.08470 [cs.LO].

\bibitem{Graedel+07}
E.~{Gr\"adel~et~al.}
\newblock {\em Finite Model Theory and Its Applications}.
\newblock Springer-Verlag, 2007.

\bibitem{GreenIveTan09}
T.~Green, Z.~Ives, and V.~Tannen.
\newblock Reconcilable differences.
\newblock In {\em Database Theory - {ICDT} 2009}, pages 212--224, 2009.

\bibitem{GreenKarTan07}
T.~Green, G.~Karvounarakis, and V.~Tannen.
\newblock Provenance semirings.
\newblock In {\em Principles of Database Systems {PODS}}, pages 31--40, 2007.

\bibitem{GreenTan17}
T.~Green and V.~Tannen.
\newblock The semiring framework for database provenance.
\newblock In {\em Proceedings of PODS}, pages 93--99, 2017.

\bibitem{LluchMon05}
A.~Lluch{-}Lafuente and U.Montanari.
\newblock Quantitative mu-calculus and {CTL} defined over constraint semirings.
\newblock {\em Theoretical Compututer Science}, 346(1):135--160, 2005.

\bibitem{Markowsky76}
G.~Markowsky.
\newblock Chain-complete posets and directed sets with applications.
\newblock {\em Algebra universalis}, 6(1):53--68, 1976.

\bibitem{Moschovakis74}
Y.~Moschovakis.
\newblock {\em Elementary induction on abstract structures}.
\newblock North Holland, 1974.

\bibitem{RamusatManSen18}
Y.~Ramusat, S.~Maniu, and P.~Senellart.
\newblock Semiring provenance over graph databases.
\newblock In {\em Proceedings of TaPP 2018}, 2018.

\bibitem{Senellart17}
P.~Senellart.
\newblock Provenance and probabilities in relational databases: From theory to
  practice.
\newblock {\em {SIGMOD} Record}, 46(4):5--15, 2017.

\bibitem{XuZhaAlaTan18}
J.~Xu, W.~Zhang, A.~Alawini, and V.~Tannen.
\newblock Provenance analysis for missing answers and integrity repairs.
\newblock {\em {IEEE} Data Eng. Bull.}, 41(1):39--50, 2018.

\end{thebibliography}
\bibliographystyle{plainurl}

\newpage
\appendix

% custom numbering in appendix
% Section: A, B, C
% Subsection: A-I, A-II, etc.
% Theorems in the appendix: A1, A2, B3, C4, etc.
\renewcommand{\thesection}{\Alph{section}}
\renewcommand{\thesubsection}{\Alph{section}-\Roman{subsection}}
\newcommand{\appendixtheorem}{\thesection\arabic{theorem}} % this is used to define \thetheorem

\renewcommand{\theExample}{\appendixtheorem}
\renewcommand{\thedefn}{\appendixtheorem}
\renewcommand{\thetheorem}{\appendixtheorem}

\setcounter{section}{0}
\setcounter{theorem}{0}

\section{Appendix: On Infinite Products of Plays}

In Sect.~\ref{Sect:Games}, we have introduced the provenance value $\pi \ext \Ss$
of a strategy $\Ss$ informally as the (possibly infinite) product over the values
$\pi \ext \rho$ of all $\rho \in \Plays(\Ss)$.
The formal definition of $\pi \ext \Ss$, on the other hand, avoids infinite products
by grouping together those plays with the same outcome.
Here we discuss how infinite products can be formulated in our setting
of absorptive, fully continuous semirings, in order to justify our definition of $\pi \ext \Ss$.
We restrict our interest to products over countable domains which are relatively easy to analyse
(in fact, finite domains would suffice as there are only finitely many different outcomes of plays).

For this section, we always assume the following setting.
Let $K$ be an absorptive, fully continuous semiring and
let $A \subseteq K$ be a countable set.
Let $I$ be an arbitrary index set
and $(x_i)_{i \in I}$ a family over $A$, that is, with $x_i \in A$ for all $i \in I$.
Notice that a single element $a \in A$ can occur several times, or even infinitely often,
in the family $(x_i)_{i \in I}$ and we allow index sets $I$ of arbitrary cardinality.
Since $A$ is countable, we fix an enumeration $A = \{ a_0, a_1, a_2, \dots \}$
(not to be confused with the elements $x_i$ of the family).
We are interested in the (possibly infinite) product of the family $(x_i)_{i \in I}$
which we denote by $\Prod_{i \in I} x_i$.

\subsection{Definition via Finite Subproducts}

In \cite{DrosteKuich09}, infinite sums in complete monoids are defined
using the supremum over all finite subsums. Here we propose an analogous
definition for infinite products.
\[
  \Prod_{i \in I} x_i \coloneqq \Inf \Big\{ \prod_{i \in F} x_i \mid \text{$F \subseteq I$, $F$ finite} \Big\}
\]
We write $\Prod$ to clearly distinguish this product from the usual finite product $\prod$.
To simplify notation, we write $F \finsub I$ to express that $F$ is a finite subset of $I$
and use the abbreviation $x_F \coloneqq \prod_{i \in F} x_i$ for finite subproducts.
Then, $\Prod_{i \in I} x_i = \Inf \{ x_F \mid F \finsub I \}$.

Some properties of infinite products are immediate, based on this definition.
Further properties and an alternative definition are discussed in the remainder of this section.

\prop
\label{propProductProperties}
The infinite product (in the setting described above)
\begin{enum}
  \item coincides with finite multiplication, i.e., $\Prod_{i \in I} x_i = \prod_{i \in I} x_i$ for finite $I$,
  \item is commutative, i.e., $\Prod_{i \in I} x_i = \Prod_{i \in I} x_{f(i)}$, where $f : I \to I$ is a bijection,
  \item satisfies $\Prod_{i \in \emptyset} x_i = 1$, and $\Prod_{i \in I} x_i = 0$ if $x_i = 0$ for some $i \in I$.
\end{enum}
\eprop

\begin{proof}
For (1), assume that $I$ is finite and $F \subseteq I$.
Then $x_F \ge x_I$ because of absorption, hence
$\Prod_{i \in I} x_i = \Inf \{ x_F \mid F \finsub I \} = x_I = \prod_{i \in I} x_i$.
For (2), consider some $F \finsub I$ and let $F' = f^{-1}(F)$.
Then $F' \finsub I$ and $\prod_{i \in F} x_i = \prod_{i \in F'} x_{f(i)}$.
Writing $x'_F$ for $\prod_{i \in F} x_{f(i)}$, it follows that
$\Prod_{i \in I} x_i = \Inf \{ x'_{F'} \mid F \finsub I$, $F' = f^{-1}(F) \} = \Inf \{ x'_F \mid F \finsub I \} = \Prod_{i \in I} x_{f(i)}$.
For (3), $\Prod_{i \in \emptyset} x_i = x_\emptyset = 1$.
If $x_i = 0$ then $x_{\{i\}} = 0$ and hence $\Inf \{ x_F \mid F \finsub I \} \le 0$.
\end{proof}

\subsection{Products as Chains}

In order to justify our definition of $\pi \ext \Ss$,
we need a form of \emph{associativity} for infinite products.
Consider a partition $I = I_1 \dcup I_2$.
We then need $\Prod_{i \in I} x_i = \big(\Prod_{i \in I_1} x_i\big) \cdot \big(\Prod_{i \in I_2} x_i\big)$
to group together plays with identical outcome.
%
%In particular, this means that infinite products are compatible
%with mutliplication by a constant $c \in A$ in the sense that
%$c \cdot \Prod_{i \in I} x_i = \Prod_{i \in I \dcup \{\star\}} x_i$
%where we set $x_\star = c$ and $\star$ does not occur in $I$.

Since the product is defined as an infimum,
this is related to the full continuity of multiplication.
However, recall that continuity applies only to infima of \emph{chains}.
We therefore reformulate the product $\Prod_{i \in I} x_i$ as the infimum of an $\omega$-chain.
Here we rely on $A$ being countable, so $A = \{ a_0, a_1, \dots \}$.
We need some further notation.
For $a \in A$ and $F \subseteq I$ we write $\cntf F a$ for
the number of occurrences of $a$ in $(x_i)_{i \in F}$, that is,
$\cntf F a = \big|\{ i \in F \mid x_i = a \}\big|$ which we understand as a number in $\N \cup \{ \infty \}$.
In case of $F = I$, we simply write $\cnt{a}$ for $\cntf I a$.
Using this notation, we define the $\omega$-chain $(\pchain x_n)_{n < \omega}$ as follows:
\[
  \pchain x_n \coloneqq \prod_{0 \le k < \min(n, |A|)} (a_k)^{\cnt {a_k}}
\]

Notice that $(\pchain x_n)_{n < \omega}$ is indeed a descending chain because of absorption.
Moreover, each finite subset of $A$ is eventually contained in the elements $a_0, \dots, a_{n-1}$
considered in $\pchain x_n$, which leads to the following observation.

\prop
\label{propProductAsChain}
It holds that $\Prod_{i \in I} x_i = \Inf_{n < \omega} \pchain x_n$.
\eprop

\begin{proof}
Recall that $\Prod_{i \in I} x_i = \Inf \{ x_F \mid F \finsub I \}$.
We show both directions of the equality.
We first fix any $F \finsub I$. Since $F$ is finite,
there is $n < \omega$ such that $\{ x_i \mid i \in F \} \subseteq \{ a_0, \dots, a_{n-1} \}$.
Then $x_F \ge \pchain x_n$, by comparing the exponents of each $a_k$ (with $k < n$)
in $x_F$ and $\pchain x_n$.
It follows that $\Prod_{i \in I} x_i \ge \Inf_{n < \omega} \pchain x_n$.

For the other direction, fix $n$ and consider $\pchain x_n$ as defined above.
Intuitively, we want to choose an appropriate set $F \finsub I$
that contains sufficiently many occurrences of the values $a_0, \dots, a_{n-1}$.
However, in case of $\cnt{a_k} = \infty$ this is not possible.
Instead, we define a family $(F_j)_{j < \omega}$ of increasingly large finite subsets of $I$.
For each $j$, choose $F_j \finsub I$ in such a way that for all $k < n$, we have
$\cntf {F_j} {a_k} = \big|\{ i \in F_j \mid x_i = a_k \}\big| = \min(j, \cnt{a_k})$.
Then, the finite products $(x_{F_j})_{j < \omega}$ form a chain
and we can apply the Splitting Lemma to see that $\Inf_{j < \omega} x_{F_j} = \pchain x_n$.
It follows that
$\Prod_{i \in I} x_i = \Inf \{ x_F \mid F \finsub I \} \le \Inf \{ x_{F_j} \mid j < \omega \} = \pchain x_n$.
This holds for all $n < \omega$, hence $\Prod_{i \in I} x_i \le \Inf_{n < \omega} \pchain x_n$.
\end{proof}

We have reformulated the infinite product as infimum of an $\omega$-chain
and can now exploit that multiplication is fully continuous
to prove further properties of infinite products.

\prop[Assocativity]
\label{propProductAssociative}
For each partition $I = I_1 \dcup I_2$,
it holds that
\[
  \Prod_{i \in I} x_i = \Big( \Prod_{i \in I_1} x_i \Big) \cdot \Big( \Prod_{i \in I_2} x_i \Big).
\]
In particular, given $c \in A$ it holds that $c \cdot \Prod_{i \in I} x_i = \Prod_{i \in I \dcup \{\star\}} x_i$ where $x_\star = c$.
\eprop

\begin{proof}
The definition of the chain $(\pchain x_n)_{n < \omega}$
for the product $\Prod_{i \in I} x_i$ on the left-hand side
depends on the values $\cnt a = \cntf I a$ for $a \in A$.
For the two products on the right-hand side,
let $(\pchain x_n^{(1)})_{n < \omega}$ and $(\pchain x_n^{(2)})_{n < \omega}$
be the corresponding chains that are defined analogously using
$\cntf {I_1} a$ and $\cntf {I_2} a$, respectively.

First observe that $\cnt a = \cntf {I_1} a + \cntf {I_2} a$ for each $a \in A$.
Since we fixed an enumeration of $A$, it follows from the definitions
of the three chains that
$\pchain x_n = \pchain x_n^{(1)} \cdot \pchain x_n^{(2)}$ for each $n < \omega$
(in case of $\cnt a = \infty$, recall \cref{lemInfpow} on the infinitary power).
This observation and an application of the Splitting Lemma entail associativity:
\[
  \Big( \Prod_{i \in I_1} x_i \Big) \cdot \Big( \Prod_{i \in I_2} x_i \Big) =
  \Big( \Inf_{n < \omega} \pchain x_n^{(1)} \Big) \cdot \Big( \Inf_{n < \omega} \pchain x_n^{(2)} \Big) =
  \Inf_{n < \omega} \pchain x_n^{(1)} \cdot \pchain x_n^{(2)} =
  \Inf_{n < \omega} \pchain x_n =
  \Prod_{i \in I} x_i
\]
The statement on multiplication with $c$ follows by considering the partition $I' = \{\star\} \dcup I$.
\end{proof}

The last property we consider is compatibility with semiring homomorphisms $h : K \to K'$.
Finite products are preserved by $h$ and we can generalize this to infinite products
if we require $h$ to be fully continuous.
This applies in particular to the homomorphisms induced by polynomial evaluation in $\Sinf[X]$.
Recall that $K$, and now also $K'$, are absorptive and fully continuous.

\prop
\label{propProductHomomorphism}
Let $h : K \to K'$ be a fully continuous homomorphism. Then,
\[
  h \Big( \Prod_{i \in I} x_i \Big) = \Prod_{i \in I} h(x_i).
\]
\eprop

\begin{proof}
We revisit the proof of $\Prod_{i \in I} x_i = \Inf_{n < \omega} \pchain x_n$ in \cref{propProductAsChain}.
For each $F \finsub I$, we have shown that there is an $n$ with $x_F \ge \pchain x_n$.
Hence $h(x_F) \ge h(x_n)$ by monotonicity of semiring homomorphisms.

Conversely, for each $n$ we have constructed a family $(F_j)_{j < \omega}$
such that $\Inf_{j < \omega} x_{F_j} = \pchain x_n$.
Using the continuity of $h$, it follows that
$h(\pchain x_n) = h(\Inf_{j < \omega} x_{F_j}) = \Inf_{j < \omega} h(x_{F_j}) \ge \Inf \{ h(x_F) \mid F \finsub I \}$.
Combining both directions, we get $\Inf \{ h(x_F) \mid F \finsub I \} = \Inf_{n < \omega} h(\pchain x_n)$.

The proof of the proposition now simply follows from the continuity assumption:
\[
  h \Big( \Prod_{i \in I} x_i \Big) =
  h\Big( \Inf_{n < \omega} \pchain x_n \Big) =
  \Inf_{n < \omega} h(\pchain x_n) =
  \Inf \{ h(x_F) \mid F \finsub I \}
\]
Clearly, $h(x_F) = \prod_{i \in F} h(x_i)$
and it follows that $\Inf \{ h(x_F) \mid F \finsub I \} = \Prod_{i \in I} h(x_i)$.
\end{proof}

\subsection{Alternative Definition}

Infinite products can be defined in different ways
and while the definition in terms of finite subproducts
is natural, one could also define the product recursively.
We show that the notion of infinite products is robust
by observing that the recursive definition is equivalent.
Let $(x_\beta)_{\beta < \alpha}$ be a family
over $A$ indexed by an ordinal $\alpha \in \On$.
The (possibly infinite) product of this family can be defined recursively by
\begin{align*}
  \pi_0 &\coloneqq 1, \\
  \pi_{\beta+1} &\coloneqq x_\beta \cdot \pi_\beta, && \text{for ordinals } \beta \in \On, \\
  \pi_{\lambda} &\coloneqq \Inf \{ \pi_\beta \mid \beta < \lambda \}, && \text{for limit ordinals } \lambda \in \On,
\end{align*}
and we define the (infinite) product as $\OrdProd_{\beta < \alpha} x_\beta \coloneqq \pi_\alpha$.

\prop
In the above setting, $\OrdProd_{\beta < \alpha} x_\beta = \Prod_{\beta \in \alpha} x_\beta$.
\eprop

\begin{proof}
We prove by induction that $\pi_\delta = \Prod_{\beta \in \delta} x_\beta$ holds for all $\delta \le \alpha$.
For $\delta = 0$, we have shown $\Prod_{\beta \in \emptyset} x_\beta = 1$ above.
For $\delta+1$, we have
$x_\delta \cdot \pi_\delta = x_\delta \cdot \Prod_{\beta \in \delta} x_\beta = \Prod_{\beta \in \delta+1} x_\beta$
by \cref{propProductAssociative}.
If $\delta$ is a limit ordinal, then
$\Inf \{ \pi_\beta \mid \beta < \delta \} =
\Inf \big\{ \Prod_{\gamma \in \beta} x_\gamma \bigmid \beta < \delta \big \}$.
By applying the definition of the infinite product
in terms of the infimum over all finite subproducts,
this is equal to
$\Inf \big\{ x_F \bigmid F \finsub \beta, \beta < \delta \big\} =
\Inf \{ x_F \bigmid F \finsub \delta \} = \Prod_{\beta \in \delta} x_\beta$.
\end{proof}

\subsection{Back to Games: Products of Plays}

The infinite product we defined and analysed above
can be used to properly define the provenance values
of strategies as products over all plays.
Recall that the value $\pi \ext \rho$ of a play $\rho$
with outcome $L$ is simply the semiring value $\pi(L)$,
and the value of infinite plays is either $0$ or $1$.
The value of a strategy $\Ss$ was defined as
\[
  \pi \ext \Ss \coloneqq \begin{cases} \displaystyle\prod_{L \in \Lit_A(\tau)}  \pi(L) ^{\litcount \Ss L} &\text{ if all infinite $\rho \in \Plays(\Ss)$ are winning for Verifier}\\
  0 &\text{ otherwise.} \end{cases}
\]

We first check that our setting applies:
In Sect.~\ref{Sect:Games}, we assumed an absorptive, fully continuous semiring $K$.
We further always assume the universe to be finite.
In particular, the number of literals and hence also the number of different values
$\pi \ext \rho$ is finite and thus countable.
The number of plays, on the other hand, can well be infinite and even uncountable%
\footnote{Consider the rather artificial formula $\phi(u) = \gfpfml R x {R x \land R x} u$
which allows Falsifier to repeatedly make a binary choice in the corresponding game.}.
We can thus model the product over all values $\pi \ext \rho$ as an infinite product
which finally justifies our definition of $\pi \ext \Ss$ in all absorptive, fully continuous semirings.

\prop
In the setting of Sect.~\ref{Sect:Games}, it holds that
\[
  \pi \ext \Ss = \Prod_{\rho \in \Plays(\Ss)} \pi \ext \rho
\]
\eprop

\begin{proof}
First observe that whenever there is an infinite play that is losing for Verifier,
such that $\pi \ext \Ss = 0$, this play has the value $\pi \ext \rho = 0$
and thus also $\Prod_{\rho \in \Plays(\Ss)} \pi \ext \rho = 0$.
If no such play exists, we group plays with identical outcome,
making use of the associativity of infinite products (see \cref{propProductAssociative}).
For each literal $L$, let $\Plays_L(\Ss)$ be the set of plays with outcome $L$.
We further write $\Plays_1(\Ss)$ for the set of infinite plays that are winning for Verifier.
If we denote the finite set of literals by $\{ L_1, \dots, L_n \}$,
we obtain the partition
$\Plays(\Ss) = \Plays_1(\Ss) \dcup \Plays_{L_1}(\Ss) \dcup \cdots \dcup \Plays_{L_n}(\Ss)$
and can apply \cref{propProductAssociative}:
\begin{alignat*}{2}
  \Prod_{\rho \in \Plays(\Ss)} \!\! \pi \ext \rho \;&=
  \Prod_{\rho \in \Plays_1(\Ss)} \!\! \pi \ext \rho \;&&\cdot\;
  \prod_{L \in \Lit_A(\tau)}
  \left(
  \Prod_{\rho \in \Plays_L(\Ss)} \!\! \pi \ext \rho
  \right)
  \\ &=
  \Prod_{\rho \in \Plays_1(\Ss)} \!\! 1  &&\cdot\;
  \prod_{L \in \Lit_A(\tau)}
  \left(
  \Prod_{\rho \in \Plays_L(\Ss)} \!\! \pi(L)
  \right)
\end{alignat*}
Infinite products of a single value, that is, $\Prod_{i \in I} x_i$ with $x_i = a$ for all $i \in I$,
have the value $\Inf \{ x_F \mid F \finsub I \} = \Inf \{ a^{|F|} \mid F \finsub I \} = a^\infty$.
We can thus conclude
\[
  \Prod_{\rho \in \Plays(\Ss)} \!\! \pi \ext \rho \;=\;
  1 \; \cdot \!\!
  \prod_{L \in \Lit_A(\tau)}
  \left(
    \pi(L)^{\litcount \Ss L}
  \right) =
  \pi \ext \Ss. \qedhere
\]
\end{proof}

\end{document}